\documentclass[12pt,letterpaper]{article}
\usepackage[T1]{fontenc}
\usepackage{lmodern}
\usepackage[utf8]{inputenc}
\usepackage{geometry}
\usepackage{booktabs}
\usepackage{amsmath}
\usepackage{amsthm} 
\usepackage{amsfonts}
\usepackage{dsfont}
\usepackage{cancel}
\usepackage{amssymb}
\usepackage{subcaption}
\usepackage{bm}
\usepackage{mathtools}
\usepackage{changepage}
\usepackage{amsthm}
\usepackage{verbatim}
\usepackage{amsmath,amssymb}
\usepackage{graphicx}
\usepackage{emptypage}
\usepackage{newlfont}
\usepackage[all,cmtip]{xy}
\usepackage[bottom]{footmisc}
\usepackage[titletoc,title]{appendix}
\usepackage{chngcntr}
\usepackage{apptools}
\usepackage{listings}
\usepackage{color}
\usepackage{sgame, tikz} 
\usepackage{kpfonts}  
\usepackage{natbib}
\AtAppendix{\counterwithin{lem}{section}}
\usepackage{caption}
\usepackage{float}
\usepackage{epigraph}
\usepackage{xr-hyper}
\usepackage[colorlinks=true,linkcolor=blue, allcolors=blue]{hyperref}
\usepackage{sgamevar}

\theoremstyle{plain} 

\newtheorem{cor}{Corollary} 
\newtheorem{prop}{Proposition}
\newtheorem{claim}{Claim} 
\newtheorem{theorem}{Theorem}
\newtheorem{observation}{Observation}
\newtheorem{lemma}{Lemma}
\theoremstyle{definition}

\newtheorem{rmk}{Remark}

\makeatletter
\newcommand{\neutralize}[1]{\expandafter\let\csname c@#1\endcsname\count@}
\makeatother

\let\emptyset\varnothing

\DeclareMathOperator{\E}{\mathds{E}}

\newcommand{\R}{\mathds{R}}

\newcommand{\cl}{\text{cl}}

\renewcommand*\d{\mathop{}\!\mathrm{d}}
\renewcommand{\epsilon}{\varepsilon}
\newcommand{\argmax}{\operatornamewithlimits{argmax}}

\usepackage{setspace}

\theoremstyle{definition}

\reversemarginpar

\usepackage[shortlabels]{enumitem} 
\usepackage[nameinlink]{cleveref}
\crefname{manualasm}{assumption}{assumptions}
\crefname{cor}{corollary}{corollaries}
\crefname{observation}{Observation}{Observations}
\crefalias{prop}{proposition}
\crefname{claim}{claim}{claims}
\crefname{ex}{example}{examples}
\crefname{defn}{definition}{definitions}
\crefname{rmk}{remark}{remarks}

\makeatletter
\newcommand{\unlinkedref}[1]{\ref*{#1}}
\makeatother

\theoremstyle{plain}
\newtheorem{innercustomthm}{\customgenericname}

\newtheorem{innercustomprop}{\customgenericname}   %  ← NEW line
\newtheorem{innercustomcor}{\customgenericname}   %  ← NEW line

\providecommand{\customgenericname}{}
\newcommand{\newcustomtheorem}[2]{%
  \newenvironment{#1}[1]{%
    \renewcommand\customgenericname{#2}%
    \crefalias{#1}{innercustomthm}% use the *thm* label‑type
    \renewcommand\theinnercustomthm{\unlinkedref{##1}$'$}%
    \innercustomthm
  }{\endinnercustomthm}}

\newcustomtheorem{customtheorem}{Theorem}

\crefname{innercustomthm}{theorem}{theorems}
\Crefname{innercustomthm}{Theorem}{Theorems}

\newcommand{\newcustomproposition}[2]{%
  \newenvironment{#1}[1]{%
    \renewcommand\customgenericname{#2}%
    \crefalias{#1}{innercustomprop}% use the *prop* label‑type
    \renewcommand\theinnercustomprop{\unlinkedref{##1}$'$}%
    \innercustomprop
  }{\endinnercustomprop}}

\newcustomproposition{customproposition}{Proposition}

\newcommand{\newcustomcorollary}[2]{%
  \newenvironment{#1}[1]{%
    \renewcommand\customgenericname{#2}%
    \crefalias{#1}{innercustomcor}% use the *prop* label‑type
    \renewcommand\theinnercustomcor{\unlinkedref{##1}$'$}%
    \innercustomcor
  }{\endinnercustomcor}}

\newcustomcorollary{customcorollary}{Corollary}

\crefname{innercustomprop}{proposition}{propositions}
\Crefname{innercustomprop}{Proposition}{Propositions}
\Crefname{innercustomcor}{Corollary}{Corollary}

\theoremstyle{plain}% ensure italic body for the inner theorems

\allowdisplaybreaks

\newcommand\citetaliasyearpar[1]{\citetalias{#1} \citeyearpar{#1}}
\newcommand\citealptaliasyearpar[1]{\citetalias{#1} \citeyear{#1}}
	\defcitealias{redistribution}{Dworczak \textcircled{r} al.}
	\defcitealias{followup}{Akbarpour \textcircled{r} al.}
    \defcitealias{simplicity}{Li \textcircled{r} Dworczak}

\begin{document}

\title{Searchable Menus\thanks{We thank Nima Haghpanah, Jason Hartline, Ravi Jagadeesan, Andreas Kleiner, Elliot Lipnowski, Lea Nagel, Axel Niemeyer, Filip Tokarski, and several seminar audiences for helpful comments. Dworczak gratefully acknowledges the support received under the ERC Starting grant IMD-101040122.}}
\author{Frank Yang\thanks{Department of Economics, Stanford University. Email: shuny@stanford.edu} \quad \quad  \textsuperscript{\textcircled{r}}  \quad  Piotr Dworczak\thanks{Department of Economics, Northwestern University; Group for Research in Applied Economics. Email: piotr.dworczak@northwestern.edu}}
\date{\today}
\maketitle
\begin{abstract}
Multidimensional screening is (in)famously intractable. In this paper, we study optimal screening mechanisms subject to a tractability constraint from the \textit{agent}'s perspective. Specifically, we require that the menu of options offered by the designer can be ordered so that, regardless of her preference type, the agent can find a utility-maximizing option via greedy search: any locally optimal choice must also be globally optimal. In one-dimensional screening with the single-crossing property, this  requirement has no bite. In multidimensional environments, however, searchability restricts the set of implementable outcomes. In the multiproduct monopoly problem, the optimal searchable menu is a sparse upgrade menu: higher tiers offer higher allocation probabilities for every good, and the number of tiers is at most the number of goods. In a multidimensional screening problem with money and ordeals, the optimal searchable menu offers the agent a single way to obtain the good. In income taxation with rich multidimensional heterogeneity, a tax schedule is searchable if and only if it is progressive.

\medskip

\noindent\textbf{Keywords:}  Multidimensional screening, menu design, searchable menus, greedy search, upgrade menus, simplicity, bundling, costly screening, taxation
\end{abstract}

\setcounter{page}{1}
\newpage

\section{Introduction}

Multidimensional screening is an important yet notoriously intractable problem. Consider a seller who sells multiple goods to a single buyer with additive values. When there is one good, the optimal mechanism is simply a posted price (\citealt{myerson1981optimal}; \citealt{riley1983optimal}). As soon as there are two goods, however, the seller may optimally offer a menu with infinitely many options (\citealt*{daskalakis2013mechanism}). More starkly, for any fixed finite menu size, there are value distributions for which every menu of that size obtains an arbitrarily small fraction of the optimal revenue (\citealt{hart2013menu}). The optimal menu can also exhibit counterintuitive comparative statics: optimal revenue may decrease when buyer values shift upward in the stochastic-dominance sense (\citealt{hart2015maximal}). In fact, recent ``anything goes'' results show that, in multidimensional screening, virtually any incentive-compatible mechanism can be uniquely optimal for some value distribution (\citealt{lahr2025extreme}).

At the same time, most menus observed in practice are simple: they often contain only a few options and are organized in ways that make them easy to navigate. Complicated menus impose a search burden on agents. A benevolent social planner naturally internalizes the search costs and mistakes induced by a complicated menu. But even a revenue-maximizing seller may prefer a simple menu: a buyer who cannot easily identify a preferred option may act unpredictably, delay the decision, or abandon the purchase altogether. Thus, in practice, the effectiveness of menus as screening devices depends not only on the options they offer, but also on whether agents can navigate them to find utility-maximizing choices.

This paper formalizes this idea through a notion of search simplicity. In our single-agent  framework, any incentive-compatible allocation can be induced by a menu. We constrain the screening problem by requiring this menu to be \textit{\textbf{searchable}}. A menu is searchable if it can be endowed with a total order such that, regardless of the agent's preferences, any locally optimal choice is also globally optimal.\footnote{In the definition, we also require that, whenever an agent has multiple globally optimal options, these options form an interval in the menu order; this assumption makes the statements of some results in the paper more elegant but it does not play a substantial role---for finite menus, this interval property is implied by the local-to-global property under continuity and non-degeneracy assumptions.} Equivalently, the agent can locate a utility-maximizing choice through greedy search. Starting anywhere in the menu, the agent compares the current option with its neighbors and moves to an adjacent option whenever doing so increases utility. Once no local improvement is available, the agent is guaranteed to have reached a global optimum.\footnote{We formalize this intuition by providing a microfoundation for searchable menus via a class of extensive-form search games and the solution concept of one-step simplicity from \citet{pycia2023theory}.}

Our first result characterizes searchable menus in a fully general screening environment with abstract outcome and type spaces (\Cref{thm:main}). A menu is searchable if and only if, within every finite submenu, the incremental comparison sets are nested. That is, if $x,y,z$ are three consecutive options in any submenu, the set of types who prefer $y$ to $x$ must contain the set of types who prefer $z$ to $y$. We refer to this property as the \textit{\textbf{incremental nesting condition}}. We also show that, when the original menu is finite, it suffices to verify the condition for the menu itself (rather than for all its submenus). Intuitively, the incremental nesting condition characterizes searchability because it precisely identifies the menu structures under which greedy search is optimal: the agent can find her utility-maximizing choice by repeatedly comparing adjacent options and stopping as soon as moving to the next option would reduce her payoff.

An immediate consequence of our characterization is that, in any one-dimensional screening problem with single-crossing preferences, searchability imposes no additional restriction on the set of implementable outcomes (\Cref{cor:1d}). Indeed, under the Spence--Mirrlees condition, every implementable allocation rule is monotone in type; moreover, ordering menu options by increasing allocations  makes the incremental comparison sets nested. We view this as a clean benchmark: one-dimensional screening problems are tractable not only for the designer but also for the agent navigating the menu. By contrast, searchability restricts the set of implementable outcomes in multidimensional environments, where the single-crossing property generally fails. In fact, searchability can be seen as preserving the nesting logic of single crossing while dispensing with a primitive type order, endogenizing the menu order, and requiring nesting only among local comparison sets (see \Cref{sec:discussion}).

We apply the abstract characterization result to three canonical multidimensional screening problems. First, we consider the multiple-good monopolist problem. Specifically, the buyer has additive and non-negative values over $n \geq 1$ goods, drawn from a full-support joint distribution. The seller wants to maximize revenue among all searchable menus. Unlike the unconstrained problem, the searchability-constrained problem turns out to yield a drastically simpler class of optimal mechanisms. In particular, we show that,  regardless of the type distribution, there always exists an optimal searchable menu that is a \textit{\textbf{sparse upgrade menu}} (\Cref{thm:multiplegood}): higher tiers give weakly higher allocation probabilities for each good, the top tier is the grand bundle, and the number of priced tiers is at most the number of goods.  

We find two aspects of \Cref{thm:multiplegood} particularly striking. First, the \textit{form} of the optimal searchable menu does not depend on the distribution---paralleling the optimality of posted prices in the one-good monopolist problem (\citealt{myerson1981optimal}). Second, the result contrasts sharply with the unrestricted multidimensional problem, where optimal menus can be essentially arbitrary, possibly with infinitely many options and highly irregular structures. Searchability, an intuitive restriction that asks agents' choice problem to satisfy a local-to-global property, turns out to convert a generally intractable problem into a tractable class of tiered mechanisms. The optimal searchable menus also share intuitive features that make them attractive for practical applications: they are monotone upgrade menus, they use few tiers, and they satisfy revenue monotonicity with respect to stochastic dominance shifts in the value distribution (see \Cref{cor:fosd}).

The proof illustrates the force of the incremental nesting condition in \Cref{thm:main}. Because the agent’s preferences are linear, the incremental comparison sets are half-spaces in a Euclidean space, with boundaries given by indifference hyperplanes. Searchability rules out interior crossings among these hyperplanes. This nested geometry first allows us to transform any searchable menu into a searchable upgrade menu without reducing revenue. Conditional on the nested comparison regions, the remaining design problem can then be expressed as a linear program over weights assigned to the indifference hyperplanes. At each menu upgrade, the corresponding weight proportionally scales the upgrade probabilities for the goods and the upgrade price, thereby preserving the comparison boundary. Because there is one capacity constraint for each good---the allocation probability of a
good cannot exceed one---an extreme point argument implies that at most $n$ increments are needed. This yields an optimal searchable menu with at most $n$ tiers.

Other multidimensional screening problems with linear preferences share this structure. To illustrate, our second application considers a screening problem where a designer wants to allocate a single good but is allowed to use multiple instruments---in addition to money, the designer can also use costly screening devices (``ordeals''). The good has a constant marginal cost and the designer wants to maximize a general welfare function (possibly with redistributive preferences). The agents differ in their values for the good and costs of completing different ordeals.  Even for a profit-maximizing designer, this is a complex screening problem (\citealt{yangCostly}). However, we show that  among searchable menus, regardless of the number of costly instruments, it is always optimal for a designer to use a \textit{\textbf{posted-requirement menu}} (\Cref{thm:ordeal}): an agent receives the good if and only if she pays a posted price and completes a fixed set of ordeal requirements. Moreover, if the designer cares only about the profit, there is an optimal searchable menu that uses no ordeals and hence reduces to a posted price.

Finally, as a canonical example of a screening problem with nonlinear preferences, we consider the influential optimal income taxation problem of \cite{mirrlees1971}. We allow for unrestricted heterogeneity in preferences, imposing only that agents' utility functions are concave in consumption and labor supply (pre-tax earnings).  We call a type space \textit{\textbf{rich}} if it includes all preferences with a constant marginal rate of substitution between consumption and earnings. In any income taxation problem with a rich type space, we show that a menu of consumption-earning pairs is searchable if and only if it corresponds to a \textit{\textbf{progressive tax schedule}} (\Cref{thm:progressive}). Intuitively, if the tax schedule is progressive, disposable income is concave in earnings, and each type’s payoff over earnings is single-peaked; hence, the tax code is searchable. Conversely, for any non-progressive tax schedule, there exists a preference type for which some earnings level is locally optimal but not globally optimal regardless of how the menu is ordered; hence, searchability fails. Even though progressive tax schedules have intuitive appeal, they are not optimal in the classical income taxation framework of \cite{mirrlees1971}, except under restrictive parametric assumptions. \Cref{thm:progressive} provides a simple microfoundation for tax progressivity: increasing marginal tax rates are exactly what is needed to make the worker's problem tractable by guaranteeing that local-optimality conditions are sufficient for global optimality. Searchability also provides a technical benefit: it ensures the validity of  the perturbation method developed by \cite{Saez}---a central approach to solving modern optimal taxation problems---that relies on first-order conditions characterizing workers’ choices. 

Our main notion of searchability assumes that the menu is totally ordered, so that the agent searches along a line. This may be restrictive when products or policies naturally vary across multiple dimensions. If preferences are separable across dimensions, the agent may instead simplify her choice problem by searching over each dimension separately. We capture this possibility by extending searchability to the notion of \textit{\textbf{grid-searchability}}. We allow the designer to decompose the menu into several (totally ordered) component menus. The final allocation is determined by the agent’s choice of one option from each component menu. An allocation is locally optimal if the agent cannot benefit by replacing her choice in any single component menu with an adjacent option. We say that a menu is \textit{\textbf{grid-searchable}} if every locally optimal allocation is globally optimal. In quasilinear environments with rich additive-linear preferences, we show that grid-searchability requires \textit{\textbf{block-additive pricing}} (\Cref{thm:grid}): the price of a joint allocation must decompose across the component menus, and each component menu must itself be searchable.

In the context of monopoly pricing, this leads to a \textit{\textbf{block-separable upgrade menu}}, where the seller first partitions the goods into different brackets and then prices the goods in each bracket separately using a sparse upgrade menu (\Cref{cor:grid-multiplegood}). Depending on the type distribution, the seller would want to optimize over the bracketing structure---intuitively, finer bracketing does not lead to richer pricing because it also increases the dimensionality of the search space and hence introduces potentially more locally optimal choices for the agents. Besides monopoly pricing, we also apply \Cref{thm:grid} to revisit our other two applications---in particular, in the taxation context, grid-searchability provides a simple microfoundation for separability of taxes across income sources (\Cref{cor:grid-tax}).

\paragraph{Related Literature.}\hspace{-2mm}There is a large literature on multidimensional screening and optimal bundling (\citealt{stigler1963united}; \citealt{adams1976commodity}; \citealt*{McAfee1989MultiproductValues}; \citealt{armstrong1996multiproduct}). It is now understood that the unrestricted problem is analytically and computationally intractable and admits little general structure. Given these difficulties, one approach is to impose structure on the type distribution by identifying classes of distributions for which the optimal mechanism is tractable (\citealt{Manelli2007}; \citealt{pavlov2011optimal}; \citealt*{daskalakis2017strong}; \citealt{haghpanah2021pure}; \citealt{ghili2023characterization}; \citealt{yangCostly,yang2023nested}) or to endogenize the relevant type distribution through robustness or learning (\citealt{carroll2017robustness}; \citealt{deb2023}; \citealt{che2021robustly}; \citealt{pernoud2025bundling}). Another approach restricts the admissible mechanism class and studies the performance of simple mechanisms---such as separate sales, pure bundling, or bounded-size menus---relative to the unrestricted optimum (\citealt*{chawla2007algorithmic}; \citealt{babaioff2014simple}; \citealt*{cai2016duality}; \citealt*{babaioff2017menu}; \citealt{hart2017approximate}; \citealt*{frick2026multidimensional}). 

We take a different approach by imposing tractability on the \textit{agent's choice problem}. Rather than prescribing a particular class of simple mechanisms or directly bounding the size of the menu, we impose a restriction on its navigability: the menu must admit greedy search, so that local optimality implies global optimality. Searchability is therefore a restriction on the structure of the agent's induced optimization problem, rather than on the number or ex ante form of the options offered. It generates restrictions on mechanisms without imposing assumptions on the type distribution. Interestingly, several conditions previously identified for the exact optimality of simple mechanisms---including pure bundling (\citealt{haghpanah2021pure}; \citealt{ghili2023characterization}), nested bundling under incremental monotonicity (\citealt{yang2023nested}), and screening frontiers under comonotonic types  (\citealt{yang2026screening})---produce menus that are searchable. In those environments, primitives ensure that searchability and optimality do not conflict. We reverse the exercise: taking searchability as a premise, we study the mechanism structures implied by agent-side tractability.

Our paper therefore also relates to the literature on simplicity and complexity in choice and mechanism design. \citet{camara2025computational} characterizes preferences consistent with computationally tractable rational choice over a rich domain of menus, showing that it must effectively take the form of narrow choice bracketing. We instead characterize the menus that make choice tractable. The notions of tractability differ: \citet{camara2025computational} studies polynomial-time computation, whereas we study greedy searchability. More broadly, our focus on choice complexity distinguishes our paper from work on strategic simplicity, which asks whether agents can identify optimal strategies without sophisticated contingent reasoning or higher-order beliefs (\citealt{li2017obviously};  \citealt{borgers2019strategically};  \citealt{pycia2023theory}; \citealt{nagel2023measure}; \citealptaliasyearpar{simplicity}).\footnote{See also \cite{FefferTokarski2026}, who develop a theory of how agents can make strategic choices in complex mechanisms by recognizing analogies to mechanisms they already understand.} Even after strategic interaction is reduced to choosing from a menu, the resulting choice problem may itself be difficult; searchability addresses this residual source of complexity. 

The remainder of the paper proceeds as follows. \Cref{sec:model} presents the model and the concept of searchable menus; \Cref{sec:char} presents the characterizations of searchable menus; \Cref{sec:application} applies the characterizations to study three multidimensional screening problems; \Cref{sec:grid} generalizes the concept of searchable menus to grid-searchable menus and revisits the three applications; \Cref{sec:discussion} concludes by discussing the connection between searchability and single crossing and suggesting directions for future research. \Cref{app:proof} contains all omitted proofs. 

\section{Model}\label{sec:model}

\paragraph{Preliminaries.}\hspace{-2mm}We consider a framework with a single agent (equivalently, a unit mass of non-interacting agents).  There is a measurable \textit{\textbf{outcome space}} $\mathcal{X}$ and a measurable \textit{\textbf{type space}} $\Theta$. Each type $\theta \in \Theta$ has preferences over the outcomes $x \in \mathcal{X}$ described by a \textit{\textbf{utility function}} $u(x, \theta)$. Types are distributed according to a commonly known \textit{\textbf{type distribution}} $\mu \in \Delta(\Theta)$. We assume that there is an outside option $\underline{x} \in \mathcal{X}$ for which the utilities of all types $\theta$ are normalized to $0$.

\paragraph{Menus.}\hspace{-2mm}A \textit{\textbf{menu}} $(\mathcal{M}, \preceq)$ is any subset of outcomes that includes the outside option, totally ordered by $\preceq$. We require that $\preceq$ is order-embeddable into $\mathbb{R}$, i.e., that $(\mathcal{M}, \preceq)$ can be indexed by $s \in S \subseteq \R$ such that $x^s \preceq x^t$ for all $s \leq t$. We then write $\{x^s\}_{s\in S}$ to represent the menu. We additionally assume that (i) the index set $S$ is compact and (ii) for any $\theta$, $u(x^s,\theta)$ is upper-semicontinuous in $s \in S$. These assumptions guarantee that search over the menu is one-dimensional, and that every type has a well-defined optimal choice.\footnote{Since any set can be totally ordered, assuming that $\preceq$ is a total order does not by itself guarantee that the effective search space is well-behaved. See \cite{debreu1954representation} for primitive conditions under which $\preceq$ is order-embeddable into $\mathbb{R}$.} We say that menu $\mathcal{M}$ \textit{\textbf{implements}} an outcome (function) $\chi:\Theta\to \mathcal{X}$ if $\chi(\theta)\in \text{argmax}\{u(x,\theta):\,x\in \mathcal{M}\}$, for all $\theta\in \Theta$. Note that any outcome implemented by a menu---viewed as a direct mechanism---is incentive-compatible and individually rational for the agent.

\paragraph{Searchability.}\hspace{-2mm}Fix an index set $S\subseteq\mathbb{R}$. For any $s<t$ in $S$, define an \textit{\textbf{interval}} $[s,t]:= \{r\in S:\,s\leq r\leq t\}$; similarly, define an open interval as $(s, t) := \{r\in S:\,s <  r < t\}$. We call an interval $[\underline{s},\bar{s}] $ a \textit{\textbf{neighborhood}} of $s\in S$ if either $\underline{s}< s < \bar{s}$, or $\underline{s}< s=\bar{s}=\max(S)$, or $\min(S)=\underline{s}= s < \bar{s}$. 
We say that $s$ is a \textit{\textbf{local maximum}} of a function $g: S\to\mathbb{R}$ if $s$ maximizes $g$ on some neighborhood of $s$.\footnote{Note that for a finite menu, $s$ is a local maximum if and only if $g(s)$ is at least as high as the value of $g$ at the adjacent indices.} We say that $s\in S $ is a \textit{\textbf{global maximum}} of a function $g: S \to \R$ if $s$ maximizes $g$ on $S$. A menu $(\mathcal{M}, \preceq)$ is \textit{\textbf{searchable}} if for any type $\theta \in \Theta$: (i) 
\[\text{$s^\star \in S$ is a local maximum of } u(x^s, \theta)  \implies \text{$s^\star \in S$ is a global maximum of } u(x^s, \theta)\,\]
and (ii) the set of global maximizers of $u(x^s, \theta)$ is an interval.  The first condition is the substantive restriction; the second makes the exposition more elegant by ruling out the non-generic case in which there are several separated global maxima. In fact, for finite menus, the second condition is implied by the first under continuity and non-degeneracy assumptions defined precisely below (see \Cref{prop:interval} in \Cref{app:interval}).\footnote{Intuitively, this is the case because a violation of the interval property at some type implies that a \textit{nearby} type cannot have the local-to-global property under continuity and non-degeneracy assumptions.}

\paragraph{Microfoundation through Greedy Search.}\hspace{-2mm}Intuitively, searchability guarantees that the agent can find an optimal choice from the menu through greedy search. By ``greedy'' search, we mean that, at each step, the agent need only compare the current option with an adjacent option, without anticipating how the remainder of the search will unfold. This intuition can be formalized using the solution concept of one-step dominance from \citet{pycia2023theory}. Indeed, given a finite menu with options indexed by $\{1, \dots, m\}$, consider an \textit{\textbf{ascending search game}} where the agent searches the menu forward from index $1$ to index $m$: at each index, the agent can decide either to clinch the current option or to continue, and at the last index the agent has a choice over all options that she has seen. Similarly, define the \textit{\textbf{descending search game}} by letting the agent search in the reverse of the menu order. In \Cref{app:one-step}, we show that, under the same continuity and non-degeneracy assumptions, a finite menu is searchable if and only if both search games are one-step simple in the sense of \citet{pycia2023theory}, so that the agent can find an optimal choice using only one-step foresight (see \Cref{prop:search}).

\section{Characterizations of Searchable Menus}\label{sec:char}

For a finite menu $(\mathcal{M}, \preceq)$ of size $m$, without loss of generality, we write the index set $S = \{1, \dots, m\}$, and define its \textit{\textbf{incremental better-off sets}} as
\[
\mathcal{H}^k := \big\{\theta : u(x^{k-1}, \theta) \le u(x^k, \theta) \big\} \qquad \text{ for all $k \in \{2, \dots, m\}$}
\]
and its \textit{\textbf{strict incremental better-off sets}} as
\[
\mathcal{G}^k := \big\{\theta : u(x^{k-1},\theta) < u(x^k,\theta)\big\} \qquad \text{ for all $k \in \{2, \dots, m\}$}\,.
\]
We say that a menu $(\mathcal{M}, \preceq)$ satisfies the \textit{\textbf{incremental nesting condition}} if 
\[
\hspace*{10em}
\mathcal{H}^{k+1}\subseteq \mathcal{H}^k
\qquad\text{and}\qquad
\mathcal{G}^{k+1}\subseteq \mathcal{G}^k
\,\tag{\textbf{Incremental Nesting}}
\]
for all $k \in \{2, \dots, |\mathcal{M}| - 1\}$.

For any menu $(\mathcal{M}, \preceq)$, its \textit{\textbf{submenus}} are defined by its subsets $\mathcal{N} \subseteq \mathcal{M}$ equipped with the same order $\preceq$.  

Our main result in this section provides a characterization of searchable menus.

\begin{theorem}\label{thm:main}
A menu $(\mathcal{M},\preceq)$ is searchable if and only if every finite submenu satisfies the incremental nesting condition. If $\mathcal{M}$ is finite, this is equivalent to the incremental nesting condition for $(\mathcal{M},\preceq)$ itself.
\end{theorem}
The full proof of \Cref{thm:main} is in the appendix---the next subsection offers a sketch. An immediate consequence of \Cref{thm:main} is that any submenu of a searchable menu is also searchable. 

\subsection{Proof Sketch of \Cref{thm:main}}

The proof consists of two steps. In \textbf{Step 1}, we observe that searchability of a menu is equivalent to an appropriate notion of quasi-concavity of the agent's payoff over menu options, which is in turn equivalent to certain set relations between types who prefer different options. In \textbf{Step 2}, we show that the set relations derived in Step 1 can be simplified to the incremental nesting condition.

\textbf{Step 1.} A function $g: S \to \R$ is called  \textit{\textbf{semi-strictly quasi-concave}} if  (i) $g$ is quasi-concave, i.e., 
\[g(t) \geq \min\big\{g(s_1), g(s_2)\big\}\]
for all $t \in [s_1, s_2]$ and all $s_1 < s_2 \in S$, and (ii) whenever $g(s_1)\neq g(s_2)$, the above inequality is strict for all $t \in (s_1, s_2)$.

\begin{lemma}\label{lem:char}
For any menu $(\mathcal{M}, \preceq)$, the following are equivalent: 
\begin{itemize}
    \item[(a)] The menu is searchable. 
    \item[(b)] For any $\theta \in \Theta$, $u(x^{s}, \theta)$ is semi-strictly quasi-concave in $s \in S$. 
    \item[(c)] For any $s<t<r$ in $S$, we have
\begin{equation*}\label{eq_nesting}
\big\{\theta : u(x^t,\theta)\le u(x^r,\theta)\big\} \subseteq
\big\{\theta : u(x^s, \theta) < u(x^t, \theta) \big\} \cup
\big\{\theta : u(x^s,\theta) = u(x^t,\theta) = u(x^r,\theta) \big\}\,.
\end{equation*} 
\end{itemize}    
\end{lemma}

The proof of \Cref{lem:char} is in the appendix. Intuitively, the semi-strict quasi-concavity of $u(x^s, \theta)$ in $s$ says that type $\theta$'s payoff over menu options is single-peaked and moreover, any nontrivial flat interval must be a global-maximizer plateau, i.e., the payoff is ``almost'' strictly quasi-concave except that it may admit a flat interval consisting of global maximizers. Note that such a property is enough for type $\theta$ to determine whether an option is globally optimal by simply comparing it with its adjacent options in the menu. Moreover, combined with the interval structure of global maximizers, the sufficiency of local comparisons becomes equivalent to the semi-strict quasi-concavity property, establishing the equivalence between \textbf{condition $(a)$} and \textbf{condition $(b)$}. To understand \textbf{condition $(c)$}, suppose $S$ is finite and take any three adjacent elements $s-1 < s < s+1$ in $S$. Note that \textbf{condition $(c)$} applied to these three elements says that for every type $\theta$, if the payoff weakly goes up when moving from $s$ to $s+1$, then either the payoff strictly goes up when moving from $s-1$ to $s$ or stays flat from $s-1$ to $s+1$. This rules out precisely any ``valley'' in the payoff function as well as any flat interval that is not globally optimal, recovering the semi-strict quasi-concavity property of  \textbf{condition $(b)$}. 
  
\textbf{Step 2.} We argue that (i) \textbf{condition $(c)$} of \Cref{lem:char} is equivalent to the incremental nesting condition for every finite submenu and (ii) the condition further simplifies when the original menu is finite. Part (i) follows by unpacking \textbf{condition $(c)$} and the definition of incremental better-off sets $\mathcal{H}^k$ and $\mathcal{G}^k$ for the submenu. Indeed, for any finite submenu $\mathcal{N} = \{x^{s_1}, \dots, x^{s_m}\}$ where $s_1 < \cdots < s_m$ are in $S$, applying \textbf{condition $(c)$} to $s_{k-1} < s_k  < s_{k+1}$ gives 
\[\mathcal{H}^{k+1} \subseteq \mathcal{G}^{k} \cup \Big(\big(\mathcal{H}^{k+1} \backslash \mathcal{G}^{k+1}\big)\cap \big(\mathcal{H}^{k} \backslash \mathcal{G}^{k}\big) \Big) \subseteq \mathcal{H}^{k}\,.\]
To show nestedness of $\{\mathcal{G}^k\}_k$, note that the above also gives 
\[\mathcal{G}^{k+1} \subseteq \mathcal{H}^{k+1} \cap \mathcal{G}^{k+1} \subseteq  
\Big[\mathcal{G}^{k} \cup \Big(\big(\mathcal{H}^{k+1} \backslash \mathcal{G}^{k+1}\big)\cap \big(\mathcal{H}^{k} \backslash \mathcal{G}^{k}\big) \Big)\Big] \cap \mathcal{G}^{k+1} \subseteq \mathcal{G}^{k}\,,\]
as desired. Moreover, the proof in the appendix shows that the converse also holds by applying the incremental nesting condition to appropriate finite submenus. 

Part (ii) of the proof shows why, when the original menu is finite, it is sufficient to simply check the incremental nesting condition for the full menu instead of all submenus. The key observation is a telescoping argument. Indeed, let $\mathcal{N}$ be any submenu. We index the options in $(\mathcal{N}, \preceq)$ by $l$ with $s_l \in S=\{1,...,m\}$ being the index in the full menu $(\mathcal{M}, \preceq)$. Let $\{{\mathcal{H}}_\mathcal{N}^l\}_l$ be the incremental better-off sets for the submenu and $\{\mathcal{H}^s\}_s$ be the incremental better-off sets for the full menu. For any $\theta \in {\mathcal{H}}_\mathcal{N}^{l+1}$, we have 
\[u(x^{s_{l+1}}, \theta) - u(x^{s_{l}}, \theta)  = \sum_{s = s_l}^{s_{l+1} - 1} \Big( u(x^{s+1}, \theta) - u(x^{s}, \theta)\Big) \geq 0\,,\]
which implies at least one of the terms is non-negative, and hence $\theta \in \mathcal{H}^{s+1}$ for some $s \ge s_l$. By the incremental nesting condition for the full menu, we have \[\theta \in \mathcal{H}^{s+1} \subseteq \mathcal{H}^{s} \subseteq \cdots \subseteq \mathcal{H}^{s_l}\subseteq \cdots \subseteq \mathcal{H}^{s_{l-1}+1}\,.\]
Then, chaining the implied inequalities from $\mathcal{H}^{s_l}\subseteq \cdots \subseteq \mathcal{H}^{s_{l-1}+1}$, it follows that 
\[u(x^{s_l}, \theta) \geq u(x^{s_{l-1}}, \theta)\,,\]
and hence $\theta \in {\mathcal{H}}_\mathcal{N}^{l}$, as desired. The same telescoping argument also works for the strict incremental better-off sets, and thus the theorem follows.

\subsection{Searchability and One-dimensional Screening}

In general, searchability imposes constraints on the outcomes that a designer can implement. However, an immediate consequence of \Cref{thm:main} is that searchability has no bite if the agent has single-crossing preferences (i.e., in one-dimensional screening environments). Formally, let $\mathcal{X} \subseteq \mathbb R^2$ and $\Theta\subseteq \mathbb R$. We say that $u((q, p), \theta)$ satisfies the \textit{\textbf{strict single-crossing condition}} if, for any $(q,p), (q', p') \in \mathcal{X}$ where $q < q'$, and any $\theta < \theta'$, we have 
\[u((q, p), \theta)\,\, \leq  \,\,u((q', p'), \theta) \implies u((q, p), \theta') \,\, < \,\,u((q', p'), \theta')\,. \]
Under quasilinear preferences, where $u((q, p), \theta) = v(q, \theta) - p$, the strict single-crossing condition reduces to the strict increasing differences condition on $v(q, \theta)$. 

We say that a menu $(\mathcal{M}, \preceq)$ is \textit{\textbf{essential}} if for every $x \in \mathcal{M}$, there exists some $\theta \in \Theta$ such that $x$ is an optimal choice for $\theta$. For implementable outcomes, it is without loss of generality to consider essential menus. 

\begin{cor}[One-dimensional screening is always searchable]\label{cor:1d}
Let $\mathcal{X} \subseteq \mathbb R^2$ and $\Theta\subseteq \mathbb R$. Suppose that $u((q, p), \theta)$ is strictly decreasing in $p$ and satisfies the strict single-crossing condition. Then, for every essential menu, its $q$-ordered menu is searchable. 
\end{cor}

\Cref{cor:1d} shows that one-dimensional screening problems are not only simple for the designer  but also for the agent. Intuitively, in one-dimensional problems, the type space is ordered (by assumption) and there is a natural ordering of any menu by increasing allocations. The single-crossing property then guarantees that any type's utility function is semi-strictly quasi-concave over the menu. Searchability follows as a consequence of \Cref{thm:main}. 

In multidimensional problems, there is generally no natural ordering on either the types or the allocations. In light of \Cref{thm:main}, searchability can be seen as restoring some of the structure imposed by single crossing: the designer must choose an ordering over the offered allocations, the nested incremental better-off sets impose a coarse ordering on types, and preferences over adjacent menu options satisfy an across-type restriction. We formalize the connection between searchability and single crossing in \Cref{sec:discussion}.

\subsection{Regular Environments}

The incremental nesting condition in \Cref{thm:main} can be further simplified once $\Theta$ is equipped with a topological structure and nearby types have similar preferences. We say that a menu $(\mathcal{M}, \preceq)$ satisfies the \textit{\textbf{weak incremental nesting condition}} if $\mathcal{H}^{k+1}\subseteq \mathcal{H}^k
$ for all $k \in \{2, \dots, |\mathcal{M}| - 1\}$. We say that an environment is \textit{\textbf{regular}} if (i) $\Theta$ is a topological space, (ii) $u(x, \theta)$ is continuous in $\theta \in \Theta$, and (iii) the following reachability condition holds: for any two outcomes $x \neq x' \in \mathcal{X}$, 
\[\big\{\theta: u(x, \theta) = u(x', \theta)\big\} \subseteq \cl\big({\{\theta: u(x, \theta) > u(x', \theta)\big\}}\big),\]
whenever the strict-comparison set is non-empty (for a set $A\subseteq\Theta$, $\cl(A)$ stands for the closure of $A$ in $\Theta$). The reachability condition assumes that there is no ``thick'' region of types who are indifferent between any two options. We say that a menu is \textit{\textbf{interior}} if every incremental better-off set $\mathcal{H}^k$ is not the full type space. 

\begin{cor}[Regular environments]\label{cor:regular}
In any regular environment, a finite interior menu is searchable if and only if it satisfies the weak incremental nesting condition. 
\end{cor}

\Cref{cor:regular} states that in any regular environment the incremental nesting condition can be further simplified to only compare the weak incremental better-off sets $\mathcal{H}^k$.

An important special case of regular environments involves quasilinear preferences whose payoffs are linear in a finite-dimensional allocation. Formally, let 
\[\mathcal{X} = \mathcal{Q} \times \R \quad \text{ and } \quad u(x, \theta) = q \cdot \theta - p\,,\]
where $\mathcal{Q} \subseteq \R^d$ is any set of allocations containing the null
allocation $0$, and where $\Theta \subseteq \R^d$ is the \textit{\textbf{payoff type space}}. This formulation covers lotteries over a finite set $\mathcal{A}$ of social alternatives, by taking $\mathcal{Q}$ to be the set of sub-probability vectors on $\mathcal{A}$ and $\theta_a$ to be the agent's value for alternative $a$. We say that two menus are \textit{\textbf{outcome-equivalent}} if they can implement the same outcome Lebesgue-almost-everywhere on $\Theta$. As a consequence of \Cref{cor:regular}, the following result shows that in this setting searchability is equivalent to the weak incremental nesting condition. 

\begin{cor}[Payoff types]\label{cor:payoff}
Consider a quasilinear environment with a payoff type space $\Theta \subseteq \R^d$. Suppose $\Theta$ is convex with a non-empty interior. Then:
\begin{itemize}
    \item[(i)] Any finite menu satisfying the weak incremental nesting condition, with the outside option optimal for a positive measure of types, is outcome-equivalent to a searchable menu.
    \item[(ii)] Conversely, any finite searchable menu satisfies the weak incremental nesting condition. 
\end{itemize}
\end{cor}

\section{Applications}\label{sec:application}

In this section, we apply our framework to three applications. \Cref{subsec:multiple-good} characterizes optimal searchable menus for a revenue-maximizing multiproduct monopolist; \Cref{subsec:ordeal} characterizes optimal searchable menus for a designer allocating a single good but using both money and ordeals to screen;  \Cref{subsec:taxation} characterizes searchable menus in the income taxation model with rich heterogeneity and shows that they coincide with progressive tax schedules. 

\subsection{Multiple-Good Monopoly Pricing}\label{subsec:multiple-good}

There are $n$ goods. Agents have quasilinear preferences in money and additive values for the goods. Specifically, let $\mathcal{X} = [0, 1]^n \times \R$, $\Theta = \R^n_{+}$, and $u(x, \theta) = q \cdot \theta - p$, where $q \in [0, 1]^n$ and $p \in \R$. We assume that the type distribution $\mu$ has full support with a density and finite first moments. The designer wants to maximize expected revenue. 

A menu $(\mathcal{M}, \preceq)$ is an \textit{\textbf{upgrade menu}} if $q^s \leq q^t$ for all $s < t$. We define the number of \textit{\textbf{tiers}} of such a menu as $|\mathcal{M}| - 1$ (the number of priced options). 

\begin{theorem}\label{thm:multiplegood}
There exists an optimal searchable menu $(\mathcal{M}, \preceq)$ such that $\mathcal{M}$ is an upgrade menu with at most $n$ tiers, where the highest tier is the grand bundle. 
\end{theorem}

The proof is in the appendix; we provide an intuitive sketch in the next subsection. As \Cref{thm:multiplegood} demonstrates, searchability---which constrains the seller to offer menus that make the buyer's choice simple---also drastically simplifies the seller's problem.  In particular, \Cref{thm:multiplegood} identifies a simple form that optimal searchable menus take for \textit{any} continuous type distribution. This is in sharp contrast to the ``anything goes'' results in the multidimensional screening literature: it is known that in the multiple-good monopolist problem, the set of mechanisms that are uniquely optimal for some type distribution is dense in the set of all IC and IR mechanisms (\citealt{lahr2025extreme}). In contrast,  \Cref{thm:multiplegood} shows that there always exists an optimal searchable menu whose menu-size complexity is at most the number of goods. We explain the key reasons behind this contrast after sketching the proof (see \Cref{rmk:extreme}). 

The upgrade menus identified in \Cref{thm:multiplegood} are in fact undominated in the sense of \citet{Manelli2007} and hence must be optimal for \textit{some} type distribution even among all menus (searchable or not) by the results of \citet{Manelli2007} and \citet{lahr2025extreme}: 
\begin{prop}\label{prop:undominated}
For any upgrade menu identified in \Cref{thm:multiplegood}, there exists a type distribution with a bounded density under which it is optimal among all menus. 
\end{prop}

\Cref{prop:undominated} implies that the characterization in \Cref{thm:multiplegood} is tight: Without further assumptions on the type distribution, we cannot make a stronger prediction about the form of the optimal mechanism. In particular, it may be optimal to offer intermediate tiers in the upgrade menu (with interior probabilities of obtaining each good) in addition to the grand bundle  (see \Cref{app_example} for a simple example illustrating this possibility).  

Searchability also restores certain natural properties of the optimal mechanism. For a general bundling problem, it is known from \citet{hart2015maximal} that revenue in the optimal mechanism may decrease when the agents' values increase in the stochastic dominance sense. This cannot happen with searchable menus:

\begin{cor}\label{cor:fosd}
    If the seller is restricted to searchable menus, the optimal revenue weakly increases when the buyer's type distribution is increased in the stochastic dominance sense. 
\end{cor}

Finally, we note that searchability rules out additive pricing. The simplest way to see this is in the two-good case. Under additive pricing, the incremental better-off set for adding good 1 is a vertical half-space, while the incremental better-off set for adding good 2 is a horizontal half-space: These sets cross and are hence not nested. Nevertheless, it seems intuitive that, under additive valuations, consumers should be able to find the optimal bundle by separately deciding how much of each good to buy. We formalize this intuition in \Cref{sec:grid}, where we extend the notion of searchability to \textit{grid-searchability}.

\subsubsection{Proof Sketch of \Cref{thm:multiplegood}}

In this subsection, we sketch the proof of \Cref{thm:multiplegood}. \Cref{fig1} presents a graphical illustration (for the case of two goods) that highlights the role of the incremental nesting condition in the argument.

Take an arbitrary finite menu $\mathcal{M}$ of size $m$.\footnote{The case of an infinite menu is handled in the proof with an approximation argument.} For any $k\in \{2,...,m\}$, define the boundary between adjacent incremental better-off sets as
\[
\mathcal{I}^k := \big\{\theta : u(x^{k-1}, \theta) = u(x^k, \theta) \big\}=\big\{\theta : \theta\cdot \Delta q^{k}=\Delta p^{k}  \big\},
\]
where $\Delta q^k=q^k-q^{k-1}\in[-1,1]^n$ is the change in the allocation and $\Delta p^k=p^k-p^{k-1}\in \mathbb{R}$ is the change in the price  associated with moving from option $k-1$ to option $k$ in the menu $\mathcal{M}$. For a generic menu, these hyperplanes $\mathcal{I}^k$ may intersect in complicated ways, as illustrated in the top left panel in \Cref{fig1}. 

\begin{figure}[t]
    \centering
    \includegraphics[width=1.1\linewidth]{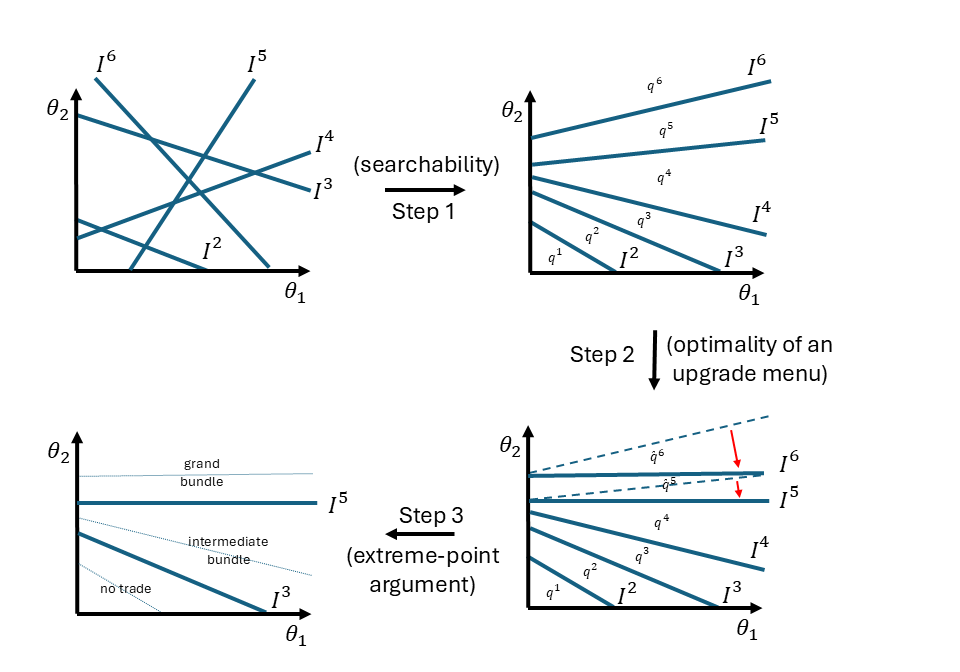}
    \caption{Illustration of the proof of \Cref{thm:multiplegood} for the case of two goods.}
    \label{fig1}
\end{figure}

\paragraph{Step 1.}\hspace{-2mm}By \Cref{thm:main} and \Cref{cor:payoff}, searchability is equivalent to the weak incremental nesting condition. Thus, if $\mathcal{M}$ is searchable, the boundary hyperplanes $\mathcal{I}^k$ \textit{cannot} intersect on the interior of the type space. The top right panel of \Cref{fig1} illustrates. The areas between the boundaries correspond to sets of types choosing the same bundle from the menu. Since $\mathcal{I}^k$'s cannot cross and the type space is assumed unbounded, for any good $i$, if $\Delta q^k_i<0$ (less of good $i$ is given out in a higher bundle), then $\Delta q_i^j\leq 0$ for all $j>k$ (the amount of good $i$ is decreasing above bundle $k$). Finally, up to reversing the order $\preceq$, it can be assumed that prices $p^k$ are increasing. Indeed, prices must be quasi-convex along the order due to the requirement of searchability for low-value types; the zero price associated with the outside option must lie at one of the endpoints of the menu; and, finally, the seller can ``chop off'' all prices that are negative to produce a monotone price schedule with greater profit.

\paragraph{Step 2.}\hspace{-2mm}As the top right panel of \Cref{fig1} illustrates, searchable menus need not be upgrade menus: Positively sloped boundary hyperplanes $\mathcal{I}^k$ (corresponding to non-monotone adjacent bundles in the menu $\mathcal{M}$) are compatible with the incremental nesting condition. In this step of the proof, we show that any non-upgrade menu can be improved for the seller to an upgrade menu. The key property is that prices $p^k$ are increasing (as established in Step 1). Thus, the seller benefits from expanding the set of types consuming a higher bundle. Graphically, as shown in the bottom right panel of \Cref{fig1}, the seller can benefit from rotating positively sloped hyperplanes $\mathcal{I}^k$ until they are flat. This modification of the better-off sets corresponds to replacing any negative increment of good $i$ across adjacent bundles with a zero increment of good $i$. This step results in an upgrade menu, $\Delta q^k\geq 0$ for all $k$, but with a possibly large number of options. 

\paragraph{Step 3.}\hspace{-2mm}The final step in the proof uses an extreme point argument. Fixing the position of the hyperplanes $\mathcal{I}^k$ does not pin down the allocation because the sets $\mathcal{I}^k$ depend only on the \textit{ratio} of $\Delta q^k$ to $\Delta p^k$, not their absolute levels. Thus, we can set up an auxiliary optimization problem of maximizing revenue over a set of non-negative weights $\lambda(k)$ multiplying both $\Delta q^k$ and $\Delta p^k$ (and thus keeping their ratio constant). Crucially, because the menu is an upgrade menu, $\lambda(k)$'s must only satisfy a single capacity constraint for each good, namely, $\sum_{k=2}^m \lambda(k) \Delta q_i^k\leq 1$. Thus, the auxiliary problem is to maximize a linear objective $\sum_{k=2}^m \lambda(k) \Delta p^k \mu(\mathcal{H}^k)$ subject to $n$ linear constraints. By a standard argument, there exists a solution  in which at most $n$ of the $\lambda(k)$'s are non-zero. This gives us an optimal menu with at most $n$ priced options. It is then easy to argue that the top option can be taken to be the grand bundle. The bottom left panel of \Cref{fig1} illustrates. 

\begin{rmk}\label{rmk:extreme}
Our argument proceeds through an extreme-point argument, as in the work showing the complexity of optimal mechanisms in the unrestricted problem (\citealt{Manelli2007}; \citealt{lahr2025extreme}). Curiously, the conclusions are opposite. In the unrestricted problem, the set of extreme points is dense (\citealt{lahr2025extreme})---in particular, the number of menu options cannot be bounded---so extremality alone places essentially no bound on menu complexity. Here, the additional structure comes from the interaction between searchability and optimality: searchability delivers nested incremental better-off sets, while \textbf{Step 1} and \textbf{Step 2} show that, given that constraint, optimality allows us to restrict attention to upgrade menus. Importantly, the perturbations in \textbf{Step 3} preserve searchability, which is what allows the extreme point argument to be run within the class of searchable menus rather than within the set of all menus. 

This connection can be made precise using results in \citet{Manelli2007}. Their Theorem 20 characterizes extremality by asking whether the associated feasible face is a singleton, while their Corollary 24.1 extracts a simpler sufficient condition from a path of market segments connecting the outside option to the grand bundle. In general, this simpler condition is not necessary because a non-trivial solution of the corresponding equality system need not generate a feasible mechanism. In our case, however, incremental nesting and the upgrade structure jointly ensure that the relevant perturbations remain feasible and preserve searchability. This simpler condition therefore becomes exact: extremality reduces to uniqueness of the solution to the system obtained by imposing our binding capacity constraints as equalities. This explains why the extreme-point structure, and hence menu complexity, becomes considerably simpler in our setting.  \qed
\end{rmk}

\subsection{Screening with Money and Ordeals}\label{subsec:ordeal}

Consider allocation of a single good $q \in [0, 1]$ with a constant marginal cost of production $C \geq 0$. Suppose that the designer can use both money $p$ as well as a set of ordeals $(y_1, \dots, y_n)$. The agent's payoffs are given by 
\[v q - p - \sum^n_{i=1} c_i y_i\,,\]
where $(v, c) \in \R^{n+1}_+$ are the types. We assume that the type distribution has a density with full support and bounded first moments. The designer can screen  using  payments $p \in \R_+$ as well as any combinations of ordeals $y \in \R^{n}_+$. The designer wants to maximize a weighted average of agent welfare and profit: 
\[\text{$\alpha$-Welfare}:= \alpha \cdot \E\Big[w \cdot \big(vq(v,c) - p(v,c) - c \cdot y(v,c)\big)\Big] + (1-\alpha) \cdot \E\Big[p(v,c) - C q(v, c)\Big]\,,\]
where the welfare weight $w \geq 0$ is jointly distributed with $(v, c)$, and $\alpha \in [0, 1]$ is the weight on agent welfare. We assume that $w$ has a bounded support. 

We say that a menu is a \textit{\textbf{posted-requirement}} menu if, besides the outside option, it consists only of $(1, p^\star, y^\star)$---the agent gets the object for sure if and only if they pay the posted price $p^\star$ and accomplish all the ordeal requirements $y^\star = (y^\star_1, \dots, y^\star_n)$. 

Note that any posted-requirement menu is a binary menu and hence searchable. It turns out that they are optimal among all searchable menus.

\begin{theorem}\label{thm:ordeal}
There exists an optimal searchable menu that is a posted requirement. Moreover, if the designer cares only about profit (i.e., $\alpha=0$), then any optimal posted-requirement menu cannot use ordeals; hence, there exists an optimal searchable menu that is a posted price. 
\end{theorem}

\Cref{thm:ordeal} is proven in a way similar to \Cref{thm:multiplegood}. The main difference is that the extreme point argument becomes simpler: In the screening problem with ordeals, there is no analog of the capacity constraint for the ordeals since the difficulties of the ordeals $(y_1,...,y_n)$ can be made arbitrarily large. As a result, the auxiliary problem in Step 3 of the proof is a linear program with a single linear constraint---hence, there exists an optimal mechanism implemented as a menu with a single non-zero option.

\subsection{Income Taxation with Multidimensional Heterogeneity}\label{subsec:taxation}

Consider a general income taxation problem. Agents have types $\theta \in \Theta$, where $\Theta$ is an abstract space. The agents' utility functions are given by  
\[
u(c,y,\theta),\]
where $c$ denotes disposable income and $y\in [0, \overline{y}]$ denotes the (pre-tax) earnings level (capturing the disutility from labor supply). We assume that, for every type $\theta$, $u$ is nondecreasing in $c$, nonincreasing in $y$, and concave in $(c,y)$.

A \textit{\textbf{tax schedule}} $T: [0, \overline{y}] \rightarrow \R$ specifies the tax at every income level. We will assume that disposable income $y - T(y)$ is strictly increasing in $y$ (the marginal tax rate is strictly below 1) to avoid the degenerate cases in which some earnings levels are dominated (and hence never chosen). 

Relative to the previous applications, non-linearity of agent preferences affects the shape of the optimal mechanism---it is no longer possible to rely on an extreme point argument that underlies the simple structure of optimal searchable mechanisms in \Cref{thm:multiplegood} and \Cref{thm:ordeal}. But, perhaps somewhat surprisingly, non-linearity does not affect the characterization of searchable menus, in the following sense. 

\newcommand{\Ulin}{\mathcal U^{\mathrm{lin}}}
\newcommand{\Uconc}{\mathcal U^{\mathrm{conc}}}
\newcommand{\Uvex}{\mathcal U^{\mathrm{vex,sep}}}

\begin{prop}
\label{prop:concave}
Suppose that $\mathcal{X}\subseteq\mathbb{R}^n$ is convex and
consider the following classes of utility functions:
\[
\Ulin
:=
\left\{
u:\mathcal{X}\to\R:
 u\text{ is nondecreasing and linear}
\right\},
\]
\[
\Uconc
:=
\left\{
u:\mathcal{X}\to\R:
 u\text{ is nondecreasing and concave}
\right\}.
\]
Menu $\mathcal{M}$ satisfies the incremental nesting condition on
$\Ulin$ if and only if it satisfies the incremental nesting condition on
$\Uconc$.
\end{prop}

\Cref{prop:concave} (proven in \Cref{app_prop:concave}) implies that if agents' preferences are monotone and concave, it suffices to ``test'' searchability for all linear preferences. Motivated by this observation, we will say that type $\theta$ has \textit{\textbf{linear preferences}} in the taxation problem if their payoff is given by $c - \alpha y$ for some $\alpha \in \R_+$, and that  $\Theta$ is \textit{\textbf{rich}} if it includes all linear preferences. A tax schedule is called \textit{\textbf{progressive}} if $T(\,\cdot\,)$ is convex.

The following result shows that, under a rich type space, progressive tax schedules are exactly the ones that are searchable in our sense. 
\begin{theorem}\label{thm:progressive}
Any progressive tax schedule is a searchable menu (ordered by earnings). If $\Theta$ is rich, then any searchable menu  must correspond to a progressive tax schedule. 
\end{theorem}

\Cref{thm:progressive} predicts that tax progressivity is exactly the condition that makes an income-tax schedule searchable. If $T$ is convex, disposable income $y-T(y)$ is concave, so linear preferences lead to payoffs that are single-peaked along the earnings scale: once moving to a slightly higher earnings level stops being attractive, further moves cannot be optimal. \Cref{prop:concave} then extends this conclusion to any concave utility function. Conversely, if $T$ is not convex, some higher earnings increment is taxed less heavily than a lower one. Under richness, there exist linear preference profiles under which the agent's payoff weakly decreases over the lower increment but increases over the higher increment. We show that, for any ordering of the menu, one can construct such a preference profile so that some earnings level is locally optimal but not globally optimal. Hence, searchability fails.

Even though progressive tax schedules are common in practice and have intuitive appeal, they are not optimal in the classical income taxation framework of \cite{mirrlees1971}, except under very restrictive parametric assumptions. \Cref{thm:progressive} provides a simple microfoundation for tax progressivity: Increasing marginal tax rates are exactly what is needed to make the worker's problem tractable through greedy search. Non-progressive tax schedules imply that there exist preferences under which workers have locally optimal earnings choices that are not globally optimal. This complements recent alternative justifications of income tax progressivity based on threat of wage randomization enabled by labor-market commitment (\citealt{Doligalski2019OptimalIncomeTaxationCommitment}) or robustness (\citealt{Vairo2025OptimalIncomeTaxationAmbiguity}).

Imposing searchability in the income-taxation problem has a further technical benefit. The dominant approach to solving modern taxation models is the perturbation method of \cite{Saez}. However, that method relies on the validity of the first-order approach to the worker's problem. By definition, searchability makes first-order conditions sufficient for global optimality. In light of \Cref{thm:progressive}, a future research direction suggested by our approach is to impose tax progressivity as a \textit{constraint} in the income taxation problem.

\section{Grid-Searchable Menus}\label{sec:grid}

Our definition of searchability asks the menu to be totally ordered, which gives rise to a one-dimensional search space for the agent. In many contexts, however, this requirement appears to be too demanding. In complex multidimensional environments, especially if the agent's utility is separable, it may be more natural to decompose the decision into several ordered dimensions, each of which can be searched independently. For example, with additive valuations and additive prices, a buyer in the multiproduct monopolist problem can separately decide how much of each good to buy.  To capture this idea, we generalize our notion of searchable menus to a concept of \textit{\textbf{grid-searchable menus}}. 

\paragraph{Environment.}\hspace{-2mm}For this section, we focus on a quasilinear environment where both the allocation and the type space can be multidimensional, and restrict attention to finite menus. In particular, let $\mathcal{X} = \prod_{i=1}^n \mathcal{Q}_i \times \R$, where each $\mathcal{Q}_i \subseteq \R$ is an interval, be the outcome space, in which $\mathcal{Q}:=\prod_{i=1}^n \mathcal{Q}_i$ is the space of allocations (e.g., a quantity of each good $i$), and $\R$ is the space of transfers. The agents' payoffs are additively separable across the dimensions, given by 
\[v(q, \theta) = \sum_{i=1}^n v_i(q_i, \theta)\]
for some functions $v_i$. We say that a type $\theta$ has \textit{\textbf{linear preferences}} if
\[v(q, \theta) = \sum_{i=1}^n q_i \alpha_i\]
for some vector $\alpha \in \R^n_+$. The assumption of additive preferences corresponds to the intuitive idea that it may be salient for the agent to consider different dimensions of the choice problem separately. 

\paragraph{Definition.}\hspace{-2mm}Fix a partition $\mathcal{K}$ of $\{1, \dots, n\}$. The designer is allowed to choose any product order $\preceq := \prod_{K \in \mathcal{K}} \preceq_K$, one for each cell in the partition.  A \textit{\textbf{(finite) product menu}} $(\mathcal{Y}, p, \preceq)$ consists of $\mathcal{Y} = \prod_{K \in \mathcal{K}} \mathcal{Y}_K$, where each $\mathcal{Y}_K \subseteq \prod_{i \in K} \mathcal{Q}_i$ is a component menu ordered by $\preceq_K$ and  $p: \mathcal{Y} \rightarrow \R$ is a price function. 
For any product menu $(\mathcal{Y}, p, \preceq)$, the set of \textbf{\textit{$\preceq$-neighbors}} of $y \in \mathcal{Y}$
is 
\[
N_{\preceq}(y):=\Big\{ z\in \mathcal{Y}:\exists\,K\text{ such that }z_{-K}=y_{-K}\text{ and }z_{K}\text{ is the predecessor or successor of }y_{K}\text{ in }\preceq_{K}\Big\} \,.
\]
That is, a neighbor of allocation $y\in \mathcal{Y}$ in the product menu $(\mathcal{Y}, p, \preceq)$ is an allocation $z\in \mathcal{Y}$ that differs from $y$ in the choice from exactly one component menu, in which moreover $y_K$ and $z_K$ are adjacent in the order $\preceq_K$.

For any product menu and any type $\theta$, we write 
\[u_\theta(y) := v(y, \theta) - p(y)\,.\]
We say that $y \in \mathcal{Y}$ is a local maximum of a function $u: \mathcal{Y} \to \R$ if $u(y)\geq u(z)$ for all $z\in N_{\preceq}(y)$. Finally, a product menu $(\mathcal{Y}, p, \preceq)$ is \textit{\textbf{grid-searchable}} if for any type $\theta \in \Theta$: 
\[\text{$y^\star \in \mathcal{Y}$ is a local maximum of } u_\theta(y)  \implies \text{$y^\star \in  \mathcal{Y}$ is a global maximum of } u_\theta(y)\,.\]

\paragraph{Discussion.}\hspace{-2mm}Grid-searchability generalizes searchability: If $\mathcal{K}$ is the trivial partition that has only one element, then every searchable menu is a grid-searchable menu. But if $\mathcal{K}$ contains more than one element, the notion of greedy search is modified relative to the baseline definition. The agent now searches by moving along one component menu at a time. Grid-searchability guarantees convergence (to the global maximum) of a search procedure in which the agent starts with any allocation and then moves sequentially toward any locally-improving option, until no further (local) improvement is possible.\footnote{In the definition of searchability, we also require that  the globally optimal choices form an interval; we could add a similar requirement here for grid-searchability without affecting any of our results.}  

Our definition allows $\mathcal{K}$ to be interpreted either as a choice variable of the designer, in which case grid-searchability expands the set of implementable outcomes, or as an exogenous constraint reflecting the salient dimensions through which the agent organizes the choice problem. The appropriate interpretation depends on the application: the former may be more natural in product design, where a monopolist controls how options are bundled and presented, while the latter may be more appropriate in taxation, where the planner typically takes the relevant dimensions of agents’ choices as given.

\paragraph{Main Characterization Result.}\hspace{-2mm}We say that $p$ is \textit{\textbf{strictly increasing}} if $p(y_K, y_{-K}) < p(y'_K, y_{-K})$ for all $y_K \prec_K y'_{K}$ and all $K \in \mathcal{K}$, where $\prec_K$ is the strict order associated with $\preceq_K$. We say that $p$ has \textit{\textbf{block-additive pricing}} if there exist functions $p_K: \mathcal{Y}_K \rightarrow \R$ for all $K\in \mathcal{K}$ such that 
\[p(y) = \sum_{K \in \mathcal{K}} p_K(y_K)\,.\]
For a product menu with block-additive pricing, the component menus are defined by $(\mathcal{Y}_K, p_K, \preceq_K)$ for each $K$. As in \Cref{sec:application}, we say that $\Theta$ is \textit{\textbf{rich}} if it includes all linear preferences. 
 
\begin{theorem}\label{thm:grid}
Let $(\mathcal{Y}, p, \preceq)$ be a product menu with a strictly increasing $p$. Suppose $v(q, \theta)$ is nondecreasing and concave in $q$. Then: 
\begin{itemize}
    \item[(i)] If $\Theta$ is rich and $(\mathcal{Y}, p, \preceq)$ is grid-searchable, then there exists a product menu with block-additive pricing, where each component menu is searchable, that can induce the same outcomes for every type.
    \item[(ii)]  Conversely, if $(\mathcal{Y}, p, \preceq)$ has block-additive pricing where each component menu is searchable, then $(\mathcal{Y}, p, \preceq)$ is grid-searchable. 
\end{itemize}
\end{theorem}

\Cref{thm:grid} states that---as long as the set of types is rich enough to include all linear preferences---grid-searchability is equivalent to additive pricing across component menus combined with (baseline) searchability of each component menu. In this way, \Cref{thm:grid} allows us to directly leverage the results from \Cref{sec:char} and \Cref{sec:application} to characterize grid-searchable menus in applications.

\paragraph{Proof Sketch of \Cref{thm:grid}.}\hspace{-2mm}Part (ii) of \Cref{thm:grid} is intuitive: Under additively separable preferences, block additive pricing implies that the choices of the agent across different component menus do not interact. As long as any component menu is searchable, the product menu is grid-searchable. 

The more interesting part of \Cref{thm:grid} is (i): Under a richness condition, grid-searchability \textit{requires} that the price function be block-additive. To understand why, it is instructive to consider the simplest case of two component menus with two options each: $\mathcal{Y}_1=\{0,\,1\}$ and $\mathcal{Y}_2=\{0,\,1\}$. Normalize $p(0,0)=0$. Then, block-additive pricing is equivalent to $p(1,1)=p(0,1)+p(1,0)$. Suppose that $p(1,1)<p(0,1)+p(1,0)$, so that the joint upgrade is cheaper than the sum of the two individual upgrades. Then, by richness, we can choose an additive-linear preference type whose value for each single upgrade is below its price, but whose value for taking both upgrades exceeds $p(1,1)$.\footnote{Strict monotonicity of $p$ ensures existence of such preferences with non-negative marginal values.} Then, allocation $(0,0)$ is locally optimal, since each one-step move is unattractive, but allocation $(1,1)$ is globally optimal, contradicting grid-searchability. If instead  $p(1,1)>p(0,1)+p(1,0)$, then the joint upgrade is too expensive: we can choose an additive-linear type for whom $(1,0)$ is better than both of its neighbors $(0,0)$ and $(1,1)$, while the non-neighboring option $(0,1)$ is strictly better. This again creates a local but not global optimum. Hence neither inequality is possible, and $p(1,1)=p(0,1)+p(1,0)$.

The proof of \Cref{thm:grid} relies on this basic logic but is more involved. In a general product menu, verifying that a given allocation is locally optimal requires ruling out profitable adjacent moves in every grid direction. Moreover, an adjacent move within a component menu may be a move between multidimensional bundles rather than a simple scalar upgrade. The proof gets around the second difficulty by first showing that, after deleting dominated options,  the allocations in each component menu are strictly increasing along the component menu order in at least one of the coordinates. Since prices are increasing, this effectively reduces the analysis to a one-dimensional problem within each component menu. 

The rest of the proof proceeds by induction on the ``rank'' of an allocation---the number of one-step moves needed to reach it from the minimal element of the grid. At ranks zero and one, the price is block-additive by the construction of a candidate decomposition, so there is nothing to prove. For the inductive step, take an allocation $y$ at the current rank and let $a$ be the allocation obtained by moving $y$ down one index in \textit{every} component where possible. Since $a$ has strictly lower rank, the additive decomposition is already exact there. We then choose a linear type that (i) attaches no value to the components in which $y$ already sits at the bottom---so that the only moves available there, which are upward and strictly more expensive, are unattractive---and (ii) has marginal value in each remaining component such that it is exactly indifferent to moving back up. Because the boundary price function is discretely convex, this type also does not want to move further down. Hence, $a$ is locally optimal for the type. But if the price of $y$ falls short of the block-additive benchmark, the same type strictly prefers $y$, contradicting grid-searchability. This is the $2\times 2$ argument again: each single upgrade is unattractive while the joint upgrade is attractive. A similar argument, moving $y$ down in a \textit{single} component instead and drawing on the inequality just established, rules out the opposite inequality---together the two deliver the block-additive structure for the inductive step.

\subsection{Applications}

\paragraph{Multiple-Good Monopoly Pricing.}\hspace{-2mm}Consider the multiple-good monopoly pricing problem in \Cref{subsec:multiple-good}. The type space $\Theta$ there consists of all linear preferences: $v(q, \theta) = \sum_i q_i \theta_i$ where $\theta \in \R^n_+$, and $q \in [0, 1]^n$. 

Applying \Cref{thm:grid}, we immediately obtain that for any partition $\mathcal{K}$, the set of grid-searchable menus is outcome-equivalent to the ones with block-additive pricing and within-block searchability. Since the buyer has additive values and the price function is block-additive, the optimal grid-searchable menu can then be solved block-by-block. Combined with \Cref{thm:multiplegood}, the optimal grid-searchable menu then takes the form of additive pricing across blocks and upgrade pricing (with at most $|K|$ tiers) within a block. 

Formally, we say that a product menu is a \textit{\textbf{block-separable upgrade}} menu if it has block-additive pricing and each component menu $(\mathcal{Y}_K, p_K, \preceq_K)$ is an upgrade menu.

\begin{cor}\label{cor:grid-multiplegood}
For any partition $\mathcal{K}$, there exists an optimal menu $(\mathcal{Y}, p, \preceq)$ among all grid-searchable menus with strictly increasing price functions such that $(\mathcal{Y}, p, \preceq)$ is a block-separable upgrade menu where, for each component $K \in \mathcal{K}$, the component upgrade menu has at most $|K|$ tiers with the highest tier being the full bundle of all goods in $K$.
\end{cor}

\paragraph{Screening with Money and Ordeals.}\hspace{-2mm}Consider the problem of screening with money and ordeals in \Cref{subsec:ordeal}. The type space $\Theta$  consists of all linear preferences: $q v + \sum_i c_i (-y_i)$ where $q \in [0, 1]$, $-y_i \in \R_{-}$ for all $i$, and $(v, c) \in \R^{n+1}_+$. In particular, letting $z_i := -y_i$ for all $i$, we have that the utility function is nondecreasing in $(q, z) \in [0, 1]\times \R^{n}_{-}$ for all types $(v, c) \in \R^{n+1}_+$. \Cref{thm:grid} yields that for any partition $\mathcal{K}$, the grid-searchable menus must have block-additive pricing and within-block searchability. It turns out that in this case a posted-requirement menu identified in \Cref{thm:ordeal} continues to be optimal even if we allow for any grid-searchable menus with any partition $\mathcal{K}$.  
\begin{cor}\label{cor:grid-ordeal}
For any partition $\mathcal{K}$ and grid-searchable product menu with a strictly increasing price function, there exists a posted-requirement menu that achieves a weakly higher $\alpha$-Welfare. 
\end{cor}

To see why the corollary is true, note that for any component $K \in \mathcal{K}$ that does not include the productive dimension $q \in [0, 1]$, under any block-additive pricing, almost all types would strictly prefer to choose $0$ over any non-zero $z_K \in \R^K_{-}$.\footnote{Formally, the payoff difference is given by 
\[c_K \cdot (0 -z_K) - \big(p_K(0) - p_K(z_K) \big) = c_K y_K  + p(z_K, 0) - p(0, 0) = c_K y_K  + p(z_K, 0) > 0\,,\]
for Lebesgue-almost-all $c_K \in \R^K_+$, where the first equality holds because of the block additivity of $p$, and the second equality holds because $p(0) = 0$ and $p \in \R_+$ by assumption.} Therefore, instead of offering the original grid-searchable menu, the designer can just offer the searchable component menu that includes the productive dimension $q \in [0, 1]$. Then, by \Cref{thm:ordeal}, a posted-requirement menu is optimal. 

\paragraph{Taxation with Multiple Incomes.}\hspace{-2mm}Consider the taxation problem in \Cref{subsec:taxation}. However, suppose that each agent has $n$ income-generating activities. A classical example is taxation with couples, where an agent is interpreted as a household, and each household member may have an independent source of income.\footnote{See, for example, \cite{KlevenKreinerSaez2009OptimalIncomeTaxationCouples}, \cite{GolosovKrasikov2025OptimalTaxationCouples}, and \cite{BierbrauerBoyerPeichlWeishaar2023TaxationCouples}.} Let $y:= (y_1, \dots, y_n) \in [0, \overline{y}]^n$ denote the generated income levels. Suppose that agents have additively separable preferences over consumption and income sources, i.e., their utility functions are given by: 
\[
u(c, \theta)-\sum_{i=1}^n h_i(y_i, \theta)
\]
where $u$ is strictly increasing and concave in consumption level $c$, and each $h_i$ is nondecreasing and convex in income level $y_i$.

Let  $T:[0, \overline{y}]^n\to \mathbb{R}$ denote the tax schedule. The agent's budget constraint is then 
\[
c\leq \sum_{i=1}^n y_i-T(y).
\]
We assume that the agent's problem, for any tax schedule $T$, is to maximize the Lagrangian with a multiplier $\lambda\geq 0$ on the budget constraint:
\[
(u(c, \theta)-\lambda c)+\sum_{i=1}^n (\lambda y_i-h_i(y_i, \theta))-\lambda T(y)\,,
\]
where the multiplier generally depends on the tax schedule $T$ and the type $\theta$.\footnote{When agents have \textit{\textbf{quasilinear preferences}}, $u(c, \theta) = c$, the Lagrangian formulation is an identity rather than an assumption: the consumption part $u(c,\theta) - \lambda c$ is maximized where $u_c(c, \theta) = \lambda$, so $\lambda = 1$ for every type and every tax schedule. The Lagrangian then reduces to $\sum_{i} \big(y_i - h_i(y_i, \theta)\big) - T(y)$, i.e., the payoff obtained by substituting the binding budget constraint into $c - \sum_i h_i(y_i, \theta)$. In that case, \Cref{cor:grid-tax} applies verbatim to the agent's problem under the budget constraint, and a linear type is simply one with a constant marginal disutility of earnings, $h_i(y_i, \theta) = \alpha_i y_i$.} Because the consumption term $u(c, \theta) - \lambda c$ is concave in $c$ and the tax does not depend on $c$, the consumption choice forms a separate component that is searchable for every type and imposes no restriction on $T$; we may therefore suppress it and let the agent search over income levels alone. This represents the taxation problem as a special case of the general framework of this section, with $\mathcal{Q}$ representing earnings choices and the tax $T$ determining the transfer.\footnote{As in the previous application, we can apply a change of variables to make the agent's preferences nondecreasing in each earnings dimension, which allows us to apply \Cref{thm:grid}.} 

As in \Cref{subsec:taxation}, we assume that the marginal tax rate for any income activity is strictly below $1$: $T(y'_i, y_{-i}) - T(y_i, y_{-i}) < y'_i - y_i$, 
for all $y_i < y'_i$ and all $y_{-i}$. We say that $T$ is a \textit{\textbf{separable tax schedule}} if 
\[T(y) = \sum_i T_i(y_i)\]
for some functions $T_i$.  Recall that $\Theta$ is \textit{\textbf{rich}} if it includes all linear preferences. Since \Cref{sec:grid} restricts attention to finite menus, we consider the finite grids of income levels from which the agent chooses. An \textit{\textbf{(income) grid}} is a product
$Y = \prod_{i=1}^n Y_i$ of finite sets $Y_i \subseteq [0, \overline{y}]$, and it induces the product menu whose $i$-th component menu is $Y_i$ and whose transfer is the restriction of $T$. We refer to the product order that ranks the income levels in each $Y_i$ from lowest to highest as the \textit{\textbf{income product order}}. As in \Cref{thm:main}, where searchability of a menu is characterized through its finite submenus, we ask greedy search to succeed on every grid.

\begin{cor}\label{cor:grid-tax}
Suppose $\Theta$ is rich and let $\mathcal{K}$ be the finest partition. The following are equivalent:
\begin{itemize}
\item[(i)] $T$ is a separable tax schedule and each $T_i$ is progressive;
\item[(ii)] for every grid, the product menu is grid-searchable under the income product order;
\item[(iii)] for every grid, the product menu is grid-searchable under some product order.
\end{itemize}
\end{cor}

The proof of \Cref{cor:grid-tax}, given in the appendix, combines \Cref{thm:progressive} and \Cref{thm:grid}. If agents make decisions about consumption and each income-generating activity separately (i.e., $\mathcal{K}$ is the finest partition), then grid-searchability implies that the tax schedule must be separable. Additionally, within each income source, taxes must be progressive. 

Of course, additional (non-separable) tax schedules can be implemented if $\mathcal{K}$ is a coarser partition. For example, if all income sources are included in a single menu, then the general characterization from \Cref{sec:char} implies that a tax schedule can be implemented as long as the incremental nesting condition holds. 

\section{Concluding Remarks}\label{sec:discussion}

In this paper, we introduced searchability as an agent-side tractability requirement on screening mechanisms. We argued that, although searchability has no bite in one-dimensional problems, it sharply disciplines menu design in multidimensional environments. Our applications to multiproduct monopoly, screening with ordeals, and income taxation illustrate how accounting for the agent’s search problem can generate simple and economically meaningful optimal mechanisms. In this section, we briefly discuss several connections and directions for future research.

\paragraph{Searchability and Single Crossing.}\hspace{-2mm}In one-dimensional screening problems, the single-crossing property provides the structure needed for tractability. As shown in \Cref{cor:1d}, searchability then imposes no additional restriction. In multidimensional problems, we showed that searchability alone provides substantial traction. This suggests a natural question: can searchability be understood as a ``substitute'' for single crossing?\footnote{Relatedly, \cite{KartikKleiner2025} develop a multidimensional generalization of single crossing based on convex choice and directional single-crossing differences.}

To understand the connection between the two concepts, it is useful to define single-crossing preferences in the language of incremental better-off sets. Fix an ordered set $\Theta \subseteq \R$. We say that agents' preferences are \textit{\textbf{single crossing on the menu}} $(\mathcal{M}, \preceq)$ if for any $\theta < \theta'$ and $x \prec x'$ in $\mathcal{M}$, $$u(x, \theta)\,\,\leq (<) \,\,u(x', \theta) \implies u(x, \theta') \,\,\leq (<) \,\,u(x', \theta').$$ A set $A\subseteq\Theta$ is \textit{\textbf{upward-closed}} if $\theta\in A$ and $\theta'>\theta$ imply $\theta'\in A$. For any $x\prec x'$ in $\mathcal{M}$, we use $\mathcal{H}_{x,x'}$ and $\mathcal{G}_{x,x'}$, respectively, to denote the weak and strict incremental better-off sets (from $x$ to $x'$).  
\begin{observation}[Single crossing as exogenous global nesting]\label{obs:single-crossing}
Let $\Theta \subseteq \R$ and let $(\mathcal{M}, \preceq)$ be a finite menu. Agents' preferences are single crossing on  $(\mathcal{M}, \preceq)$  if and only if 
\[\text{both }
\big\{\mathcal H_{x,x'}:x\prec x'\text{ in }\mathcal M\big\}  \text{ and } \big\{\mathcal G_{x,x'}:x\prec x'\text{ in }\mathcal M\big\}   \text{ are nested families of upward-closed sets}\,.\]
\end{observation}

Comparing \Cref{thm:main} with \Cref{obs:single-crossing}, note that searchability relaxes the structure of the single-crossing property in three ways. First, searchability does not require a primitive one-dimensional order over types; the relevant nesting may be \textit{induced} by the menu itself. Second, it does not require the menu order to coincide with a primitive allocation order; the designer may \textit{choose} the order in which options are searched. Third, even restricted to a given menu, searchability differs from single crossing by only imposing a local chain condition; every ordered finite submenu must have nested \textit{adjacent} comparison sets (\Cref{thm:main}), but the comparison sets generated by non-local comparisons need not themselves form a nested family (as in the case of single crossing, by \Cref{obs:single-crossing}). In this sense, searchability preserves the nesting logic of single crossing while endogenizing the relevant order and localizing the restriction.

\paragraph{Future Research Directions.}\hspace{-2mm}Both concepts analyzed in this paper---searchability and grid-searchability---take a particular view of how the agent conducts search. The underlying idea that local optimality should imply global optimality, however, is more general. For example, one could consider an agent who searches along the branches of a tree or, more broadly, over some partially ordered set. We leave the analysis of such alternative search structures for future research.

In this paper, we imposed searchability as a constraint. A natural question is whether searchability would arise endogenously under an explicit behavioral model. Following the approach of \citetaliasyearpar{simplicity}, suppose that agents may select any locally optimal option from the menu and that the designer evaluates the resulting outcome according to a worst-case criterion. Under what conditions does the designer optimally choose a searchable menu?

Finally, the behavioral premise underlying searchability also lends itself naturally to experimental investigation. In a laboratory setting, one could test whether searchability reduces choice errors and decision times, and whether these effects are sufficiently strong to justify imposing searchability as a design constraint.

\setlength\bibsep{2pt}
\bibliographystyle{ecta} 
\bibliography{references}

\appendix

\section{Omitted Proofs}\label{app:proof}

\subsection{Proof of \texorpdfstring{\Cref{lem:char}}{}}

The proof consists of three parts. 

$(a) \implies (b)$: Suppose the menu is searchable. Fix any $\theta$. Let $g(s):= u(x^s, \theta)$. We claim that $g(s)$ is semi-strictly quasi-concave. We first show quasi-concavity. Suppose for contradiction that there exist some $s < t < r \in S$ such that 
\[g(t) < \min\{g(s), g(r)\}\,.\]
Let 
\[s_L \in \argmax_{z\in [\min S, t]} g(z) \quad \text{ and }  \quad s_R \in \argmax_{z\in [t, \max S]} g(z)  \,.\]
The maximizers exist by the compactness of $S$ and the upper-semicontinuity of $g$. Since $g(s) > g(t)$ and $g(r) > g(t)$, it follows immediately that $s_L < t < s_R$. Then, by definition, it follows that both $s_L$ and $s_R$ are local maximizers of $g$. By searchability, they are global maximizers. It follows that 
\[g(t) < g(s_L) = g(s_R) = \max_{z \in S} g(z)\,.\]
However, this implies that the set of global maximizers of $g$ is not an interval, contradicting searchability. Thus, $g$ is quasi-concave. 

We now show the semi-strict part. Suppose for contradiction that there exist $s  < t < r \in S$ such that $g(s) \neq g(r)$ and 
\[g(t) = \min\big\{g(s), g(r)\big\}\,.\]
Consider the case where $g(s) < g(r)$ (the other case is symmetric). We claim that $s$ is a local maximizer of $g$. Indeed, by the above, we have $g(r) > g(t) = g(s)$. Then, by quasi-concavity of $g$, we have  $g(z) \geq g(s) = g(t)$ for all $z \in [s, t]$. Moreover, by quasi-concavity of $g$, we also have 
\[g(z) \leq g(t) = g(s) < g(r)\]
for all $z \in [\min S, t] $. It follows immediately that $g(s) = g(z)$ for all $z  \in [s, t] $ and 
\[g(s) \geq g(z)\]
for all $z \in [\min S, s] $. Thus, by definition, $s$ is a local maximizer of $g$. However, since $g(r) > g(s)$, $s$ is not a global maximizer of $g$, contradicting searchability. \\

$(b) \implies (c)$: Suppose condition $(b)$ holds. Fix any $\theta$. Let $g(s):= u(x^s, \theta)$. Then, $g$ is semi-strictly quasi-concave. Fix any $s < t < r$. We want to show that 
\[g(t) \leq g(r) \implies \text{ either } g(s) < g(t) \text{ or } g(s) = g(t)  = g(r)\,.\]
Suppose $g(t) \leq g(r)$. We claim that $g(s) \leq g(r)$. Indeed, if $g(s) > g(r)$, then by semi-strict quasi-concavity of $g$, we have 
\[g(t) > \min\big\{g(s), g(r)\big\} = g(r) \]
which is impossible. Now, there are two cases. First, suppose $g(s) < g(r)$. Since $g$ is semi-strictly quasi-concave, we have 
\[g(t) > \min\big\{g(s), g(r)\big\} = g(s)\,,\]
as desired. Second, suppose $g(s) = g(r)$. By quasi-concavity of $g$, we have $g(t) \geq g(s) = g(r)$. Combined with $g(t) \leq g(r)$, we conclude 
\[g(t) = g(r) = g(s)\,, \]
as desired. \\

$(c) \implies (a)$: Suppose condition $(c)$ holds. We show that the menu is searchable. Fix any $\theta$. Let $g(s):= u(x^s, \theta)$. We first show that the set of global maximizers of $g$ is an interval. Indeed, fix any $s < t < r$ such that $s$ and $r$ are global maximizers of $g$. Then, we have 
\[g(t) \leq g(r)\,,\]
which by condition $(c)$ implies that 
\[\text{either } g(s) < g(t) \text{ or } g(s) = g(t) = g(r)\,.\]
The first case is impossible given that $s$ is also a global maximizer of $g$. It follows that $g(s) = g(t) = g(r)$ and hence $t$ is also a global maximizer of $g$, as desired. 

We now show that every local maximizer of $g$ must be a global maximizer. Suppose for contradiction that $s^\star$ is a local maximizer of $g$ but not a global maximizer. Let $r \in S$ be a global maximizer of $g$. Then $g(r) > g(s^\star)$. First, consider the case where $s^\star < r$. Then, since $s^\star$ is a local maximizer, there exists some $t \in S$ such that $s^\star < t < r$ and 
\[g(s^\star) \geq g(t)\,.\]
Thus, we have 
\[g(r) > g(s^\star) \geq g(t)\,.\]
Applying condition $(c)$ to $s^\star < t < r$, we obtain 
\[\text{either }  g(s^\star)  < g(t) \text{ or } g(s^\star) = g(t) = g(r)\,.\]
The first possibility contradicts $g(s^\star) \geq g(t)$ and the second possibility contradicts $g(s^\star) < g(r)$. Now, consider the case where $r < s^\star$. The argument is symmetric. Since $s^\star$ is a local maximizer, there exists some $t \in S$ such that $r < t < s^\star$ such that \[g(s^\star) \geq g(t)\,.\] 
Thus, we have 
\[g(r) > g(s^\star) \geq g(t)\]
as before. Applying condition $(c)$ to $r < t < s^\star$, we obtain 
\[\text{either }  g(r)  < g(t) \text{ or } g(r) = g(t) = g(s^\star)\,.\]
The first possibility contradicts $g(r) \geq g(t)$ and the second possibility contradicts $g(s^\star) < g(r)$. This concludes the proof. 

\subsection{Proof of \texorpdfstring{\Cref{thm:main}}{}}

We first prove the first statement of \Cref{thm:main} and then prove the simplification for a finite menu. Fix any menu $(\mathcal{M}, \preceq)$. By \Cref{lem:char}, the menu is searchable if and only if condition $(c)$ of \Cref{lem:char} holds.

First, we claim that, under condition $(c)$ of \Cref{lem:char}, every finite submenu satisfies the incremental nesting condition. Fix any finite submenu $\mathcal{N}= \{x^{s_1}, \dots, x^{s_m}\}$ for $s_1 < \cdots < s_m$. Define the incremental better-off sets $\mathcal{H}^k$ and $\mathcal{G}^k$ for the submenu $\mathcal{N}$. Fix any $k \in \{2, \dots, m-1\}$. Suppose $\theta \in \mathcal{H}^{k+1}$. Then
\[u(x^{s_k}, \theta) \leq u(x^{s_{k+1}}, \theta)\,.\]
Applying condition $(c)$ of \Cref{lem:char} to $s_{k-1} < s_k < s_{k+1}$, we obtain either 
\[u(x^{s_{k-1}}, \theta) < u(x^{s_{k}}, \theta)\]
or 
\[u(x^{s_{k-1}}, \theta) = u(x^{s_{k}}, \theta) = u(x^{s_{k+1}}, \theta)\,.\]
In either case, we have $u(x^{s_{k-1}}, \theta) \leq u(x^{s_{k}}, \theta)$ and hence $\theta \in \mathcal{H}^{k}$, as desired. Similarly, suppose $\theta \in \mathcal{G}^{k+1}$. Then $u(x^{s_k}, \theta) < u(x^{s_{k+1}}, \theta)$. Applying condition $(c)$ of \Cref{lem:char} to $s_{k-1} < s_k < s_{k+1}$, we also obtain the above as before. Moreover, the second possibility cannot occur since $u(x^{s_k}, \theta) < u(x^{s_{k+1}}, \theta)$, and hence we conclude $u(x^{s_{k-1}}, \theta) < u(x^{s_{k}}, \theta)$, i.e., $\theta \in \mathcal{G}^{k}$. Together, these arguments prove that the incremental nesting condition holds for every finite submenu of $\mathcal{M}$.

Second, we claim that if every finite submenu satisfies the incremental nesting condition, then condition $(c)$ of \Cref{lem:char} must hold. To see this, fix any $s < t < r$ in $S$. Consider the three-option submenu $\mathcal{N}=\{x^s, x^t, x^r\}$. Define the incremental better-off sets $\mathcal{H}^k$ and $\mathcal{G}^k$ for the submenu $\mathcal{N}$. The incremental nesting condition for $\mathcal{N}$ states that $\mathcal{H}^3 \subseteq \mathcal{H}^2$ and $\mathcal{G}^3 \subseteq \mathcal{G}^2$. Now, fix any $\theta$ such that $u(x^t, \theta) \leq u(x^r, \theta)$. Then, $\theta \in \mathcal{H}^{3} \subseteq \mathcal{H}^2$, and hence 
\[u(x^s, \theta) \leq u(x^t, \theta)\,.\]
If $u(x^s, \theta) < u(x^t, \theta)$, then the first possibility of condition $(c)$ holds. Otherwise, we have $u(x^s, \theta) = u(x^t, \theta)$. We claim that $u(x^s, \theta) = u(x^t, \theta) = u(x^r, \theta)$. Indeed, if not, we must have $u(x^t, \theta) < u(x^r, \theta)$ and hence $\theta \in \mathcal{G}^3 \subseteq \mathcal{G}^2$---this would imply $u(x^s, \theta) < u(x^t, \theta)$ but that is impossible. This proves that the second possibility of condition $(c)$ holds. Together, these arguments prove that condition $(c)$ holds, as desired. 

Now, fix any finite menu $(\mathcal{M}, \preceq)$. We show that the incremental nesting condition holds for every submenu if and only if it holds for the full menu $(\mathcal{M}, \preceq)$. The necessity direction is obvious. We prove the sufficiency direction. Suppose the incremental nesting condition holds for the full menu $(\mathcal{M}, \preceq)$. Define the incremental better-off sets $\mathcal{H}^k$ and $\mathcal{G}^k$ for the full menu $\mathcal{M}$. Now, fix any submenu $(\mathcal{N}, \preceq)$. We index the options in $(\mathcal{N}, \preceq)$ by an index $l$, with $s_l \in S=\{1,...,m\}$ being the index in the full menu $(\mathcal{M}, \preceq)$. Define the incremental better-off sets ${\mathcal{H}}_\mathcal{N}^l$ and ${\mathcal{G}}_\mathcal{N}^l$ for the submenu $\mathcal{N}$. Fix any $\theta \in {\mathcal{H}}_\mathcal{N}^{l+1}$. We claim that $\theta \in {\mathcal{H}}_\mathcal{N}^{l}$. Indeed, since $\theta \in {\mathcal{H}}_\mathcal{N}^{l+1}$, we have 
\[u(x^{s_{l+1}}, \theta) - u(x^{s_{l}}, \theta)  \geq 0\,.\]
Therefore, we have 
\[ \sum_{s=s_{l}}^{s_{l+1}-1}\Big(u(x^{s+1}, \theta) - u(x^{s}, \theta)\Big) = u(x^{s_{l+1}}, \theta) - u(x^{s_{l}}, \theta) \geq 0\,.\]
Thus, at least one of the differences in the above sum must be non-negative, i.e., there exists some $s \in \{s_{l}, \dots, s_{l+1} - 1\}$ such that 
\[u(x^{s+1}, \theta) - u(x^{s}, \theta) \geq 0\]
and hence $\theta \in \mathcal{H}^{s+1}$. By the incremental nesting condition for the full menu, it follows that $\theta \in \mathcal{H}^{k+1}$ for all $k \in \{1, \dots, s\}$. Thus, for all $k \in \{1, \dots, s\}$, we have 
\[u(x^{k+1}, \theta) \geq u(x^{k}, \theta) \,.\]
Since $s_l-1 \leq s$, chaining the above inequalities, we obtain 
\[u(x^{s_l}, \theta) \geq u(x^{s_l-1}, \theta) \geq u(x^{s_l-2}, \theta) \geq \cdots \geq u(x^{s_{l-1}}, \theta)\,,\]
and hence $\theta \in {\mathcal{H}}_\mathcal{N}^{l}$, as desired. The nestedness of $\{{\mathcal{G}}_\mathcal{N}^{l}\}_{l}$ follows by the same argument as above except that the weak inequalities are replaced with the strict inequalities. Together, these arguments prove that the incremental nesting condition holds for the submenu $(\mathcal{N}, \preceq)$, as desired. 

\subsection{Proof of \texorpdfstring{\Cref{cor:1d}}{}}
Let $(\mathcal{M}, \preceq)$ be an essential menu. First, we claim that distinct options in $\mathcal{M}$ must have distinct $q$'s. Indeed, suppose 
\[(q, p), (q, p') \in \mathcal{M}\,.\]
Suppose, without loss, $p < p'$. Then, since $u((q, p), \theta)$ is strictly decreasing in $p$ for all $\theta$, it follows immediately that $(q, p')$ cannot be optimal for any type $\theta \in \Theta$, contradicting essentiality. Thus, the options in $\mathcal{M}$ must have distinct $q$'s. Now, consider the menu $(\mathcal{M}, \preceq_q)$ where we order the elements in $\mathcal{M}$ by their associated $q$. 

To prove \Cref{cor:1d}, by \Cref{thm:main}, it suffices to show that every finite submenu of $(\mathcal{M}, \preceq_q)$ satisfies the incremental nesting condition. Fix any finite submenu $(\mathcal{N}, \preceq_q)$. Write $\mathcal{N} = \{x^1, \dots, x^m\}$ ordered so that 
\[q^1 < q^2 < \cdots < q^m\,.\]
By essentiality of $\mathcal{M}$, for each $k \in \{1, \dots, m\}$, there exists some $\theta_k \in \Theta$ such that $x^k$ is an optimal choice for $\theta_k$ among all choices in $\mathcal{M}$ and hence among all choices in $\mathcal{N}$. Thus, we have 
\[u(x^k, \theta_k) \geq u(x^j, \theta_k)\]
for all $j$. In particular, we have
\[u(x^k, \theta_k) \geq u(x^{k+1}, \theta_k) \quad \text{ and }  \quad  u(x^k, \theta_k) \geq u(x^{k-1}, \theta_k)\,,\]
provided the neighboring options exist. Define the incremental better-off sets $\mathcal{H}^k$ and $\mathcal{G}^k$ for this finite submenu. Suppose $\theta \in \mathcal{H}^{k+1}$. Then 
\[u(x^k, \theta) \leq u(x^{k+1}, \theta)\,.\]
We claim that $\theta \geq \theta_k$. Indeed, if not, we have $\theta < \theta_k$ and $q^k < q^{k+1}$, which, by the single-crossing property, implies 
\[u(x^k, \theta_k) < u(x^{k+1}, \theta_k)\]
contradicting the optimality of $x^k$ for $\theta_k$. Thus, $\theta \geq \theta_k$. Now, since $q^k > q^{k-1}$ and $\theta \geq \theta_k$, applying the single-crossing property again, we have 
\[u(x^k, \theta) \geq u(x^{k-1}, \theta)\,,\]
and thus $\theta \in \mathcal{H}^k$, as desired. Now, suppose $\theta \in \mathcal{G}^{k+1}$. Then 
\[u(x^k, \theta) < u(x^{k+1}, \theta)\,.\]
The same argument as before implies $\theta \geq \theta_k$. In fact, here we must have $\theta > \theta_k$ because if $\theta = \theta_k$, then $x^k$ would be optimal for $\theta$ and hence the above cannot hold. Then, since $q^k > q^{k-1}$ and $\theta > \theta_k$, the single-crossing property gives 
\[u(x^k, \theta) >  u(x^{k-1}, \theta)\,,\]
and thus $\theta \in \mathcal{G}^k$, as desired. The result follows immediately by \Cref{thm:main}.

\subsection{Proof of \texorpdfstring{\Cref{cor:regular}}{}}

Fix any finite interior menu $(\mathcal{M}, \preceq)$ and write $\mathcal{M} = \{x^1, \dots, x^m\}$. The necessity direction is immediate by \Cref{thm:main}. We prove the sufficiency direction. Suppose $(\mathcal{M}, \preceq)$ satisfies the weak incremental nesting condition. By \Cref{thm:main}, since the menu is finite, it suffices to verify the nestedness of $\{\mathcal{G}^k\}_k$ where the strict incremental better-off sets $\mathcal{G}^k$ are defined for the full menu $(\mathcal{M}, \preceq)$. Fix any $k \in \{2, \dots, m-1\}$. Suppose $\theta \in \mathcal{G}^{k+1}$. Then, 
\[u(x^k, \theta) < u(x^{k+1}, \theta)\,.\]
By the weak incremental nesting condition, this implies that 
\[u(x^{k-1}, \theta) \leq u(x^k, \theta)\,.\]
We claim that the above inequality must be strict. Suppose for contradiction that we have $u(x^{k-1}, \theta) = u(x^k, \theta)$. Since the menu is interior, we have $\mathcal{H}^k \neq \Theta$ and hence there exists some $\theta' \in \Theta$ such that $\theta' \not\in \mathcal{H}^k$. Therefore, the following strict-comparison set is non-empty: 
\[A := \big\{\hat{\theta}: u(x^{k-1}, \hat{\theta}) > u(x^{k}, \hat{\theta})\big\} \neq \emptyset\,.\]
Since the environment is regular, we have 
\[\big\{\hat{\theta}: u(x^{k-1}, \hat{\theta}) = u(x^{k}, \hat{\theta})\big\} \subseteq \cl(A)\,.\]
Since $u(x^{k-1}, \theta) = u(x^k, \theta)$, we have $\theta \in \cl(A)$. Therefore, for any open set $U \subseteq \Theta$ with $\theta \in U$, we have $U \cap A \neq \emptyset$. Moreover, we also have $\theta \in \mathcal{G}^{k+1}$, which is an open set by the continuity of $u(x, \,\cdot\,)$. Thus, $\mathcal{G}^{k+1} \cap A \neq \emptyset$. Fix any $\hat{\theta} \in A \cap \mathcal{G}^{k+1}$. Then, we have $\hat{\theta} \in A \cap \mathcal{H}^{k+1} \subseteq A \cap \mathcal{H}^{k}$ by the weak incremental nesting condition. However, $\hat{\theta} \in A$ means that 
\[ u(x^{k-1}, \hat{\theta}) > u(x^{k}, \hat{\theta}) \]
while $\hat{\theta} \in \mathcal{H}^k$ means that 
\[ u(x^{k-1}, \hat{\theta}) \leq  u(x^{k}, \hat{\theta})\,,\]
a contradiction. Thus, we conclude 
\[u(x^{k-1}, \theta) < u(x^k, \theta)\]
and hence $\theta \in \mathcal{G}^k$, as desired. The result follows immediately by \Cref{thm:main}.

\subsection{Proof of \texorpdfstring{\Cref{cor:payoff}}{}}

Part (ii) is immediate by \Cref{thm:main}. We prove part (i). 

We first verify that the payoff type spaces are regular environments. Clearly, $\Theta \subseteq \R^{d}$ is a topological space, and $u(x, \theta)$ is affine and hence continuous in $\theta$. We verify the reachability condition. Fix any two distinct outcomes $x, x'$ and let 
\[A:= \big\{\theta \in \Theta: (q - q') \cdot \theta > p - p'\big\}\]
and 
\[I:= \big\{\theta  \in \Theta : (q - q') \cdot \theta = p - p'\big\}\,.\]
It suffices to show $I \subseteq \cl(A)$ whenever $A \neq\emptyset$. Suppose $A \neq\emptyset$. Then, there exists some $\hat{\theta} \in A$. Fix any $\theta \in I$. Since $\Theta$ is convex, for every $\varepsilon \in (0, 1)$, we have 
\[\theta_\varepsilon := (1 - \varepsilon) \theta + \varepsilon \hat{\theta} \in \Theta\,.\]
Moreover, by construction for any $\varepsilon \in (0, 1)$, we have 
\[(q - q') \cdot \theta_\varepsilon > p - p'\]
and hence $\theta_\varepsilon \in A$. Moreover $\theta_\varepsilon \rightarrow \theta$ as $\varepsilon \rightarrow 0$. Thus, $\theta \in \cl(A)$. Since $\theta \in I$ is arbitrary, $I \subseteq \cl(A)$, as desired. Therefore, any such environment is regular. 

Now, to prove part (i), let $(\mathcal{M}, \preceq)$ be any finite menu satisfying the weak incremental nesting condition. If the menu is interior, then it is searchable by \Cref{cor:regular}. Now, suppose the menu is not interior. We construct an outcome-equivalent searchable menu. For any $x \in \mathcal{M}$, define its demand set as 
\[D(x; \mathcal{M}) := \Big\{\theta \in \Theta: x \in \argmax_{y \in \mathcal{M}} u(y, \theta)\Big\}\,.\]
Let 
\[\mathcal{N} := \Big\{x \in \mathcal{M}: \text{Lebesgue measure}\big(D(x; \mathcal{M})\big) > 0\Big\}\,.\]
Consider this submenu $(\mathcal{N}, \preceq)$. By assumption, the outside option is optimal for a positive measure of types, so $\underline{x} \in \mathcal{N}$ and hence $(\mathcal{N}, \preceq)$ is itself a menu. By the proof of \Cref{thm:main}, the weak incremental nesting condition holds for the submenu as well. We claim that (i) $(\mathcal{N}, \preceq)$ can induce the same outcomes almost everywhere on $\Theta$ and (ii) $(\mathcal{N}, \preceq)$ is an interior menu. The result then follows by \Cref{cor:regular}. 

To see why $(\mathcal{N}, \preceq)$ induces a.e. the same outcome, let 
\[N := \bigcup_{x \in \mathcal{M} \backslash \mathcal{N}} D(x; \mathcal{M})\,.\]
Since $\mathcal{M}$ is finite, by construction, we have 
\[\text{Lebesgue measure}(N) = 0\,.\]
Consider any $\theta \not \in N$. By definition, type $\theta$ cannot find any $x \in \mathcal{M}\backslash \mathcal{N}$ optimal among the options in $\mathcal{M}$. Therefore, when offered the submenu $\mathcal{N}$, the set of optimal choices for type $\theta$ remains unchanged. Since this holds for all $\theta \not\in N$, under the constructed submenu, there exists a (measurable) selection of optimal choices that coincides with the original choice selection almost everywhere on $\Theta$. Hence, $(\mathcal{N}, \preceq)$ is outcome-equivalent to $(\mathcal{M}, \preceq)$. 

Now, we argue that $(\mathcal{N}, \preceq)$ is an interior menu. Suppose for contradiction that it is not an interior menu. Then, by definition, there exists some $k$ such that
\[u(x^{k-1}, \theta) \leq u(x^{k}, \theta)\]
for all $\theta \in \Theta$. This implies that 
\[D(x^{k-1}; \mathcal{M}) \subseteq \Big\{\theta \in \Theta : u(x^{k-1}, \theta) = u(x^k, \theta) \Big\} \subseteq \Big\{\theta \in \R^{d}: (q^{k} - q^{k-1}) \cdot \theta = p^{k} - p^{k-1}\Big\} =: Z\,.\]
Note that $Z$ defines a hyperplane in $\R^{d}$ and hence either equals the full space $\R^{d}$ or has Lebesgue measure $0$. In the first case, we would have $q^{k} = q^{k-1}$ and $p^{k} = p^{k-1}$, which is impossible since $x^k$ and $x^{k-1}$ are two distinct options. Now, in the second case, $D(x^{k-1}; \mathcal{M}) \subseteq Z$ must have Lebesgue measure $0$, but $x^{k-1} \in \mathcal{N}$ and hence  $D(x^{k-1}; \mathcal{M})$, by construction, has strictly positive Lebesgue measure, which is a contradiction. 

Therefore, $(\mathcal{N}, \preceq)$ is a finite interior menu that satisfies the weak incremental nesting condition, and can induce the same outcomes as $(\mathcal{M}, \preceq)$ almost everywhere on $\Theta$. Moreover, by \Cref{cor:regular},  $(\mathcal{N}, \preceq)$ is searchable, concluding the proof. 

\subsection{Proof of \texorpdfstring{\Cref{thm:multiplegood}}{}}

We prove \Cref{thm:multiplegood} in several steps. We first prove the optimality of an upgrade menu with at most $n$ tiers among finite searchable menus. We then extend the argument to general searchable menus, and finally show that the grand bundle must be the highest tier.

\subsubsection{Proof for Finite Menus}

\begin{prop}\label{prop:finite-screening}
For any searchable menu $\mathcal{M}$ of finite menu size $m$, there exists some searchable upgrade menu $\mathcal{N}$ with at most $n$ tiers such that 
\[\emph{Rev}(\mathcal{M}) \leq \emph{Rev}(\mathcal{N})\,.\]
\end{prop}

We prove \Cref{prop:finite-screening}.  Fix any searchable menu $(\mathcal{M}, \preceq)$ of finite menu size $m$. By \Cref{thm:main}, without loss of generality, we can assume that $(0, 0) \in \mathcal{M}$ and any other options $(q, p)$ in $\mathcal{M}$ satisfy $q \neq  0$ since  any option $(0, p)$ where $p > 0$ can be removed without affecting searchability. Similarly, it is without loss of generality to assume that $\Delta q^k \neq 0$ for all $k$. Let 
\[\Delta q^k := q^k - q^{k-1} \in [-1, 1]^n\,,\]
and 
\[\Delta p^k := p^k - p^{k-1} \in \R\,.\]
Note that 
\[\mathcal{H}^k = \Big\{\theta: \theta \cdot \Delta q^k \geq \Delta p^k \Big\}\,.\]
By \Cref{cor:payoff}, without loss of generality, we can assume that $\{\mathcal{H}^k\}_{k}$ is a nested family of half-spaces, where $\mathcal{H}^{k} \supseteq \mathcal{H}^{k+1}$ for all $k$. 
We normalize the direction of the indices so that $\Delta p^k > 0$ for at least one $k$. Note that revenue under $\mathcal{M}$ is given by 
\[\text{Rev}(\mathcal{M}) = \sum_{k=2}^m \Delta p^k \mu(\mathcal{H}^k) + p^1\,,\]
where $\mu$ is the type distribution.

\paragraph{Step 1.}\hspace{-2mm}We claim that it is without loss of generality to assume that $\Delta p^k \geq 0$ for all $k$ and $p^k > 0$ for all $k > 1$.  

To prove this, we first argue that if $\Delta p^k < 0$, then $\Delta p^j \leq 0$ for all $j < k$. Suppose for contradiction that there exists some $j < k$ such that $\Delta p^k < 0$ and $\Delta p^j > 0$. Since $\Delta p^k < 0$, we must have type $0 \in \mathcal{H}^k$. Since $\Delta p^j > 0$, we have type $0 \not\in \mathcal{H}^j$, contradicting $0 \in  \mathcal{H}^k \subseteq \mathcal{H}^j$. 

Now, fix the largest index $k$ such that $\Delta p^k < 0$. By our normalization, $k < m$. By definition, for all $j > k$, we have $\Delta p^j \geq 0$. By the above claim, for all $j < k$, we have $\Delta p^j \leq 0$. This implies that $p^1 \geq \dots \geq p^k$, and $p^k \leq \cdots \leq p^m$. Moreover, for quasi-concavity to hold for a sufficiently high type $\theta$, it must be that the outside option $(0, 0)$ is at either end of the indices, i.e., $s = 1$ or $s = m$. 

In either case, consider the menu $\mathcal{N}$ constructed by
\[\mathcal{N} = \Big\{(q, p) \in \mathcal{M}:\, p > 0\Big\} \cup \Big\{(0, 0)\Big\}\,.\]
By \Cref{thm:main}, $\mathcal{N}$ is a searchable menu given that $\mathcal{M}$ is a searchable menu (with the inherited ordering of $\mathcal{M}$). Moreover, $\mathcal{N}$ must generate a weakly higher revenue for the designer. In the first case, where the outside option is at index $s = 1$, we immediately have that $\mathcal{N}$ satisfies $\Delta p^k \geq 0$ for all $k$. In the second case, where the outside option is at index $s = m$, we immediately have that $\Delta p^k \leq 0$ for all $k$. But then indexing the options in the reverse direction gives the desired menu. Observe that, in fact, we then have $\Delta p^k > 0$ for all $k$: indeed, $\Delta p^2 = p^2 > 0$, and if $\Delta p^k = 0$ for some $k \geq 3$, then type $0 \in \mathcal{H}^k \subseteq \mathcal{H}^2$ by the nestedness of $\{\mathcal{H}^k\}_k$, which is impossible since $\Delta p^2 > 0$.

\paragraph{Step 2.}\hspace{-2mm}We now claim that it is without loss of generality to assume that $\Delta q^k \geq 0$ for all $k$. 

To prove this, we first argue that if $\Delta q^k_{i} < 0$ for some $k$ and some $i$, then $\Delta q^j_i \leq 0$ for all $j > k$. Suppose for contradiction that there exist some $j > k$ and some $i$ such that $\Delta q^k_{i} < 0$ and $\Delta q^j_i > 0$. Consider the set $\mathcal{H}^j$. Consider type $\theta$ where $\theta_i = K$ and $\theta_j = 0$ for all $j \neq i$. For $K$ large enough, we have 
\[\theta \cdot \Delta q^j = K \Delta q^j_i > \Delta p^j\,.\]
Thus $\theta \in \mathcal{H}^j$. At the same time, 
\[\theta \cdot \Delta q^k = K \Delta q^k_i < 0\,,\]
and hence $\theta \not \in \mathcal{H}^k$ since $\Delta p^k \geq 0$ (by \textbf{Step 1}). But then this contradicts $\theta \in \mathcal{H}^j \subseteq \mathcal{H}^k$ given that $k < j$. 

Now, for each $i$, consider the first index $k(i)$ such that $\Delta q^{k(i)}_i < 0$. It follows from the above that for all $i$, and all $j \geq k(i)$, $\Delta q^{j}_i \leq 0$. We construct a new searchable menu as follows. For all $j$, and all $i$, let  
\[\Delta \bar{q}^j_i = \begin{cases}
    \Delta q^j_i &\text{ if } j < k(i) \,; \\
    0  &\text{ if } j \geq k(i)\,.
\end{cases}
\]
Let $\Delta \bar{p}^j = \Delta p^j$ for all $j$. First, note that for all $i$, and all $r$, by the above, we must have 
\[0 \leq \sum_{j=2}^r \Delta \bar{q}^j_i \leq \max_{k=2,\dots, m} \sum_{j=2}^k \Delta q^j_i \leq 1\,.\]
Therefore, $\{\Delta \bar{q}^{j}\}_j$ is a feasible sequence of incremental quantities. 

Now, we argue that the resulting $\{\bar{\mathcal{H}}^j\}_j$ is still a nested family. To see this, fix any $j < r$. Consider any $\theta \in \mathcal{\bar{H}}^r$. 
We claim that $\theta \in \mathcal{\bar{H}}^j$. Indeed, since $\theta \in \mathcal{\bar{H}}^r$, we have 
\[\theta \cdot \Delta \bar{q}^r = \sum_{i: k(i) > r} \theta_i \Delta q^r_i  \geq \Delta p^r\,.\]
Let $\theta'$ be the type such that $\theta'_i = \theta_i$ for all $i$ such that $k(i) > r$ and $\theta'_i = 0$ otherwise. By construction, we have 
\[\theta' \cdot \Delta q^r \geq \Delta p^r\,,\]
and hence $\theta' \in \mathcal{H}^r$, which by nestedness of $\{\mathcal{H}^k\}_k$ implies that $\theta' \in \mathcal{H}^j$. Therefore, 
\[\theta' \cdot \Delta q^j \geq \Delta p^j\,,\]
and hence 
\[\sum_{i: k(i) > r} \theta_i \Delta q^j_i \geq \Delta p^j\,.\]
Since $r > j$, we have 
\[\theta \cdot \Delta \bar{q}^j = \sum_{i: k(i) > j} \theta_i \Delta q^j_i \geq \sum_{i: k(i) > r} \theta_i \Delta q^j_i \geq \Delta p^j\,,\]
where the first inequality also uses the fact that $\Delta q^j_i \geq 0$ for all $j < k(i)$. Therefore, $\theta \in \mathcal{\bar{H}}^j$ proving the claim. 

Finally, let $\bar{p}^k = \sum_{j\leq k} \Delta \bar{p}^j$ and $\bar{q}^k_i = \sum_{j\leq k} \Delta \bar{q}^j_i$ for all $i$. Since $\Delta \bar{p}^k = \Delta p^k > 0$ for all $k$ (by \textbf{Step 1}), we have $0 \not\in \bar{\mathcal{H}}^k$ for all $k$, and hence the menu $\mathcal{N}:=\{(\bar{q}^k, \bar{p}^k)\}_k$ is interior. By \Cref{cor:regular}, since the menu $\mathcal{N}:=\{(\bar{q}^k, \bar{p}^k)\}_k$ satisfies the weak incremental nesting condition, it is a searchable menu (with the ordering given by $k$). Moreover, this menu must yield a weakly higher revenue because 
\[\text{Rev}(\mathcal{M}) = \sum_{k=2}^m \Delta p^k \mu(\mathcal{H}^k) \leq \sum_{k=2}^m \Delta p^k \mu(\mathcal{\bar{H}}^k) = \text{Rev}(\mathcal{N})\,,\]
where the inequality uses that $\Delta p^k \geq 0$ for all $k$ (by \textbf{Step 1}), and that 
\[\mathcal{H}^k \subseteq \mathcal{\bar{H}}^k\]
which holds since by construction
\[\theta \cdot \Delta \bar{q}^k \geq  \theta \cdot \Delta q^k\,.\]

\paragraph{Step 3.}\hspace{-2mm}We now construct a particular family of searchable menus, and apply an extreme point argument. By \textbf{Step 1} and \textbf{Step 2}, we fix any searchable menu $(\mathcal{M}, \preceq)$ of finite menu size $m$ where (i) $\Delta p^k \geq 0$ and $\Delta q^k_i \geq 0$ for all $k$ and all $i$ and (ii) $p^k > 0$ for all $k > 1$. Let $\{\mathcal{H}^k\}_{k}$ be the nested family of half-spaces defined before. Recall that $k = 1$ is the outside option, and since $\Delta p^k > 0$ for all $k$ (by \textbf{Step 1}), we have $0 \not\in \mathcal{H}^k$ and hence $\mathcal{H}^k \neq \Theta$ for all $k$. Moreover, without loss of generality, we may remove any option with $\Delta q^k_i = 0$ for all $i$, since it must not be chosen by any types. In our construction, we will hold $\{\mathcal{H}^k\}_{k}$ fixed throughout. 

For any $\lambda: \{2,\dots, m\} \rightarrow \R_+$, define 
\[\Delta \bar{p}^k = \lambda(k) \Delta p^k\, \text{ and }\, \Delta \bar{q}^k = \lambda(k) \Delta q^k \,.\]
Note that for any $\lambda$ such that 
\[\sum_{k=2}^m \lambda(k) \Delta q^k_i \leq 1\,\, \text{ for all $i$}\,,\]
we have that $\Delta \bar{q}^k$ is a feasible sequence of incremental quantities. Indeed, for any $i$ and any $k$, given \textbf{Step 2}, we have 
\[\lambda(k) \Delta q^k_i \geq 0\,,\]
and therefore for all $r = 1, \dots, m$, 
\[0 \leq \sum_{k=2}^r \lambda(k) \Delta q^k_i \leq \sum_{k=2}^m \lambda(k) \Delta q^k_i \leq 1 \,.\]
Moreover, for any such $\lambda$, define
\[\bar{p}^r := \sum_{k=2}^r \Delta \bar{p}^k  \text{ and }\, \bar{q}^r_i = \sum_{k = 2}^r \Delta \bar{q}^k_i \text{ for all $i$}\,.\]
Then note that, by \Cref{cor:regular},  $\bar{\mathcal{M}}_\lambda :=\big\{(\bar{q}^r, \bar{p}^r)\big\}$ is a searchable menu, with the same ordering but skipping the indices $k$ such that $\lambda(k) = 0$. In particular, it admits a nested family of half-spaces as indifference curves: 
\[\{\mathcal{H}^k\}_{k: \lambda(k) > 0} \subseteq \{\mathcal{H}^k\}_{k}\,.\]
Therefore, for any such $\lambda$, the revenue under the corresponding searchable menu can be written as 
\[\text{Rev}(\bar{\mathcal{M}}_\lambda) = \sum_{k=2}^m \lambda(k) \Delta p^k \mu(\mathcal{H}^k)\,.\]
Thus, maximizing revenue among this family of searchable menus parameterized by $\lambda$ is equivalent to 
\[\max_{\lambda: \{2,\dots, m\} \rightarrow \R_+} \sum_{k=2}^m \lambda(k) \Delta p^k \mu(\mathcal{H}^k)\]
subject to 
\[\sum_{k=2}^m \lambda(k) \Delta q^k_i \leq 1\,\, \text{ for all $i=1,\dots, n$}\,.\]
Consider the polytope defined by 
\[\Lambda := \Big\{\lambda \in \R^{m-1}_+: \lambda \cdot \Delta q_i \leq 1 \text{ for all $i = 1, \dots, n$}\Big\}\,,\]
which by construction is a compact, convex set (it is bounded since for any $k$, $\Delta q^k_i > 0$ for some $i$). Note that every $\lambda \in \text{ext}(\Lambda)$ must satisfy $\lambda(k) > 0$ for at most $\text{rank}(\{\Delta q^k\}_k) \leq n$ indices given that there are $n$ linear constraints. By Bauer's maximum principle, it follows immediately that there exists some searchable menu $\bar{\mathcal{M}}_{\lambda^{*}}$ with at most $n$ priced tiers (where $\lambda^*(k) > 0$) such that 
\[\text{Rev}(\bar{\mathcal{M}}_\lambda)\leq \text{Rev}(\bar{\mathcal{M}}_{\lambda^{*}})\]
for all $\lambda \in \Lambda$, proving the result. 

\subsubsection{Proof for General Menus}

We now extend the previous argument to arbitrary searchable menus. Let $(\mathcal{M}, \preceq)$ be a searchable menu. We make a few observations. 

First, by \Cref{thm:main}, if $\mathcal{N} \subseteq \mathcal{M}$, then $(\mathcal{N}, \preceq)$ is also a searchable menu. 

Second, the indirect utility function of the agent,
\[U(\theta) := \max_{(q, p) \in \mathcal{M}} q\cdot \theta - p,\]
is a continuous and convex function. Therefore, by a standard argument, there exists a countable collection $\mathcal{M}^{c} \subseteq  \mathcal{M}$ such that 
\[U(\theta) = \sup_{(q, p) \in \mathcal{M}^c} q\cdot \theta - p\,.\]

Third, fix any ordering over the options in $\mathcal{M}^c$ and take a sequence of finite menus $\mathcal{M}^{(m)}$ defined by the collection of first $m$ options in $\mathcal{M}^c$. Let 
\[U^{(m)}(\theta) = \max_{(q, p) \in \mathcal{M}^{(m)}} q\cdot \theta - p\,.\]
It follows immediately that $U^{(m)}$ is a monotone sequence of convex functions converging pointwise to $U$ as $m \rightarrow \infty$. Therefore, $U^{(m)}$ converges uniformly to $U$ on any compact set $Z \subseteq \R^n_+$ by Dini's theorem. By Attouch's theorem, this implies that almost everywhere on $\R^n_+$ where $U$ is differentiable, we have that $\nabla U^{(m)}(\theta)$ (with any measurable selection $x^{(m)}(\theta)$) must converge to $\nabla U(\theta)$ as $m \rightarrow \infty$. 

Let 
\[p^{(m)}(\theta) :=  \nabla U^{(m)}(\theta) \cdot \theta - U^{(m)}(\theta) \,.\]
Then $p^{(m)}$ converges to $p = \nabla U \cdot \theta - U$ almost everywhere. It follows that 
\[\int p^{(m)}(\theta) \d \mu(\theta) \rightarrow \int p(\theta) \d \mu(\theta)\,,\]
by the dominated convergence theorem given that $\mu$ has finite first moment (since $p(\theta) \leq \sum_i \theta_i$). Thus, for any searchable menu and  for any $\varepsilon > 0$, there exists a finite searchable menu that yields at most $\varepsilon$ difference in revenue compared to the original searchable menu. Since this holds for all $\varepsilon > 0$, we have 
\begin{align*}
    \sup_{\text{searchable menus } \mathcal{M}} \text{Rev}(\mathcal{M}) &= \sup_{\text{searchable finite menus } \mathcal{M}} \text{Rev}(\mathcal{M}) \\
    &= \sup_{\text{searchable upgrade menus with $n$ tiers } \mathcal{M}} \text{Rev}(\mathcal{M})\,,
\end{align*}
where the last equality is due to \Cref{prop:finite-screening}. Note that the last supremum is attainable since the objective is bounded and continuous. Moreover, given finite first moment, by dominated convergence, it also follows that the $n$ prices can be restricted to lie in a large enough box, and hence the feasible set is compact. Thus a solution exists, completing the proof. 

\subsubsection{Completion of the Proof}

To complete the proof of \Cref{thm:multiplegood}, it suffices to show that for any optimal upgrade menu with menu size $m \leq n+1$, the highest tier must be the grand bundle. Suppose for contradiction that $(\mathcal{M}, \preceq)$ is an optimal upgrade menu with menu size $m$ where the highest tier is not the grand bundle. This implies that for some good $i$, we have 
\[\sum_{k=2}^{m} \Delta q^k_i < 1\,.\]
Consider the perturbed menu $(\mathcal{N}, \preceq)$, where we keep all the upgrades except that we modify the first upgrade as follows: For some sufficiently small $\varepsilon > 0$, 
\[\Delta \bar{q}^2_j =\begin{cases}
\Delta q^2_j + \varepsilon &\text{ if $j = i$ }\\
\Delta q^2_j &\text{ otherwise}\,.\\     
\end{cases}\]
and $\Delta \bar{p}^2 = \Delta p^2$. This is a feasible modification for small enough $\varepsilon > 0$ since $\sum_{k=2}^{m} \Delta q^k_i < 1$. Note that this must result in a nested family of half-spaces since $\mathcal{\bar{H}}^2 \supset \mathcal{H}^2$ by construction. Therefore, this is a searchable menu and it yields a revenue 
\[\text{Rev}(\mathcal{N}) = \sum_{k=2}^{m} \Delta p^k \mu(\bar{\mathcal{H}}^k) >  \sum_{k=2}^{m} \Delta p^k \mu(\mathcal{H}^k) =\text{Rev}(\mathcal{M})\,,\]
where the strict inequality also uses that $\mu$ has full support. This contradicts the optimality of $\mathcal{M}$, concluding the proof. 

\subsection{Proof of \texorpdfstring{\Cref{prop:undominated}}{}}

Fix any upgrade menu identified in \Cref{thm:multiplegood}. Without loss of generality, we can write the options as $\{(q^1, p^1), \dots, (q^m, p^m)\}$ where $0 = q^1 \leq q^2 \leq \cdots \leq q^m$ and $0 = p^1 < p^2 \leq \cdots \leq p^m$. 

Fix any $\overline{\theta} \in \R_+$ large enough such that 
\[\overline{\theta} (1 - q^k_i) > p^m - p^k\]
for all $i$ and all $k < m$ where $q^k_i < 1$. Since the upgrade menu is finite, such $\overline{\theta}$ always exists. 

Now, consider the type space $\overline{\Theta}:= [0, \overline{\theta}]^n$. We claim that the  mechanism induced by the upgrade menu must be an undominated mechanism on $\overline{\Theta}$ in the sense of \citet{Manelli2007}. To prove this, we apply Lemma 10 of \citet{lahr2025extreme} (which states that every IC and IR mechanism that excludes the lowest type and has no marginal distortion at the top in each dimension is undominated). 

First, we show that the lowest type is excluded. Indeed, since $p^k > 0$ for all $k > 1$, type $0$ strictly prefers the outside option $(0, 0)$ over all the remaining options in the menu. Moreover, since we have an upgrade menu, the set of types who strictly prefer the outside option is given by 
\[\big\{\theta \in \overline{\Theta}: \theta \cdot q^2 < p^2 \big\} \supseteq \big\{\theta \in \overline{\Theta}: \theta \cdot \mathbf{1} < p^2 \big\}\]
which includes a positive Lebesgue-measure set of types in $\overline{\Theta}$.

Second, we show that there is no marginal distortion at the top for each dimension $i$ in the sense of \citet{lahr2025extreme}. Indeed, consider any dimension $i$ and type $\theta$ with $\theta_i = \overline{\theta}$. We claim that every optimal choice $k \in \{1, \dots, m\}$ for type $\theta$ must satisfy $q^k_i = 1$. Recall that the upgrade menu has the grand bundle as the highest tier (i.e., $q^m_j = 1$ for all goods $j$).  Now, fix any tier $k < m$ with $q^k_i < 1$. Consider the payoff difference for type $\theta$ when comparing tier $m$ and tier $k$: 
\begin{align*}
\theta \cdot (q^m - q^k) - (p^m - p^k) &= \theta \cdot (\mathbf{1} - q^k) - (p^m - p^k)    \\
&\geq \overline{\theta} (1 - q^k_i) - (p^m-p^k) > 0\,,
\end{align*}
where the strict inequality is by construction. This proves that under the given upgrade menu, regardless of the selection of the optimal choices by different types, the induced mechanism must have no marginal distortion at the top for each dimension $i$. By Lemma 10 of \citet{lahr2025extreme}, any such induced mechanism is therefore undominated. 

\Cref{prop:undominated} then follows by Lemma 7 of \citet{lahr2025extreme}. 

\subsection{Example for \Cref{subsec:multiple-good}}\label{app_example}

Suppose that there are two goods and a full-support distribution approximating the following discrete type distribution: (i) type $(1, 0)$ with probability $1/2$, and (ii) type $(0, 1/2)$ with probability $1/2$. We claim that any optimal searchable menu must involve a lottery, and hence consists of two tiers. Indeed, if not, then it must be a pure bundling mechanism. Note that the distribution of values for the bundle is constructed to be supported on $\{1/2, 1\}$ with equal probabilities. Therefore, the revenue from pure bundling is $0.5$. Now, consider the following menu 
\[\mathcal{M} = \big\{(0, 0, 0), \,\, (0.25, 0.5, 0.2),\,\, (1, 1, 0.9) \big\}\,.\]
It can be verified that this menu is searchable on the type space $\R^2_+$. Note that type $(1, 0)$ would buy the grand bundle and get payoff $0.1$ (buying the first option gives only payoff $0.25 - 0.2 = 0.05$); type $(0, 1/2)$ would buy the first option and get payoff $0.25 - 0.2 = 0.05$ (buying the bundle gives a negative payoff). Therefore, the total revenue would be 
\[\frac{1}{2} \times 0.9 + \frac{1}{2} \times 0.2 = 0.45 + 0.1 = 0.55 > 0.5\,,\]
strictly higher than the pure bundling revenue. For a full-support continuous type distribution approximating this discrete type distribution in the weak-$^\star$ topology, the same comparison between the two mechanisms would hold. Thus, for such distributions, the optimal searchable menu must include two tiers.

\subsection{Proof of \texorpdfstring{\Cref{cor:fosd}}{}}
By  \Cref{thm:multiplegood}, an upgrade menu is optimal for the seller.  Observe that, for any optimal upgrade menu, type $\theta$ must choose a bundle that is weakly higher than the bundle chosen by any type $\theta'\leq \theta$ (by single crossing). This means that the revenue from type $\theta$ is weakly higher than the revenue from any coordinate-wise lower type (since, without loss of generality, prices are increasing in the index of the bundle within an optimal upgrade menu). An FOSD shift in the distribution can be represented as a coupling under which the higher distribution results in a coordinate-wise higher type almost surely. The conclusion follows. 

\subsection{Proof of \texorpdfstring{\Cref{thm:ordeal}}{}}

The proof follows a similar structure to the proof of \Cref{thm:multiplegood}. The optimality of a posted-requirement menu would continue to hold even for more general objectives as long as the objective is a linear functional of the mechanism.

\subsubsection{Proof for Finite Menus}

\begin{prop}\label{prop:finite}
For any searchable menu $\mathcal{M}$ of finite menu size, there exists a posted-requirement menu $\mathcal{N}$ such that 
\[\emph{$\alpha$-Welfare}(\mathcal{M}) \leq \emph{$\alpha$-Welfare}(\mathcal{N})\,.\]
\end{prop}

We prove \Cref{prop:finite}. Fix any searchable menu $(\mathcal{M}, \preceq)$ of finite menu size $m$. By \Cref{thm:main}, without loss of generality, we may remove any option with $q = 0$ and either $p > 0$ or $y_i > 0$ for some $i$---these options will not be chosen by any type. 

As in the proof of \Cref{thm:multiplegood}, we use the notation 
\[\Delta q^k := q^k - q^{k-1} \in [-1, 1]\,,\]
and 
\[\Delta p^k := p^k - p^{k-1} \in \R\,,\]
and 
\[\Delta y^k := y^k - y^{k-1} \in \R^n\,.\]
Let 
\[\mathcal{H}^k := \Big\{(v, c) : v \Delta q^k - c \cdot \Delta y^k \geq \Delta p^k \Big\}\,.\]
By \Cref{cor:payoff}, without loss of generality, we can assume that $\{\mathcal{H}^k\}_{k}$ is a nested family of non-empty half-spaces, where $\mathcal{H}^{k} \supset \mathcal{H}^{k+1}$ for all $k$. Moreover, without loss of generality, we may assume that each $\mathcal{H}^k$ has a positive measure. Moreover, consider a type $(v, 0)$ for sufficiently high $v$. In order for their choice problem to be quasi-concave, it must be that the outside option is at either end of the index set $k = 1$ or $k = m$. We normalize the direction so that $k = 1$ is the outside option. 

Note that the designer's payoff under $\mathcal{M}$ is given by 
\[\text{$\alpha$-Welfare}(\mathcal{M}) = \alpha \cdot \sum_{k=2}^m\E\Big[w \cdot \big(v \Delta q^k - \Delta p^k - c \cdot \Delta y^k\big) \mid (v, c) \in \mathcal{H}^k\Big] \mu(\mathcal{H}^k) + (1-\alpha) \cdot \sum_{k=2}^m (\Delta p^k-C \Delta q^k) \mu(\mathcal{H}^k)\,,\]
where $\mu$ is the type distribution.

\paragraph{Step 1.}\hspace{-2mm}We first argue that $\Delta p^k \geq 0$ for all $k$ and $\Delta y^k_i \geq 0$ for all $k$ and all $i$. To see $\Delta p^k \geq 0$ for all $k$, suppose for contradiction that $\Delta p^k < 0$ for some $k$. We claim that $\Delta p^j \leq 0$ for all $j < k$. Indeed, if $\Delta p^j > 0$ for some $j < k$, then we must have that the type $0 \in \R^{n+1}_+$ does not belong to $\mathcal{H}^j$ but belongs to $\mathcal{H}^k$ since $\Delta p^k < 0$, which is impossible since $\mathcal{H}^k \subseteq \mathcal{H}^j$. It then follows that $\Delta p^j \leq 0$ for all $j < k$, and hence $p^k = p^1 + \sum_{j=2}^k \Delta p^j < p^1 = 0$, contradicting that any payment $p$ is non-negative.  

Now, we prove that $\Delta y^k_i \geq 0$ for all $k$ and all $i$. Fix any ordeal $i$. Suppose for contradiction that $\Delta y^k_i < 0$ for some $k$. We claim that $\Delta y^j_i \leq 0$ for all $j < k$. Indeed, if $\Delta y^j_i > 0$ for some $j < k$, then considering the type $(0, c_i, 0)$ with sufficiently large $c_i$, we have $(0, c_i, 0)\not\in \mathcal{H}^j$ since $\Delta y^j_i > 0$, but $(0, c_i, 0) \in \mathcal{H}^k$ since $\Delta y^k_i < 0$, which is impossible since $\mathcal{H}^k \subseteq \mathcal{H}^j$.  It then follows that $\Delta y^j_i \leq 0$ for all $j < k$. Then, note that $y^k_i = y^1_i + \sum_{j=2}^k \Delta y^j_i < y^1_i = 0$, where the strict inequality also uses $\Delta y^k_i < 0$, contradicting that any ordeal requirement $y_i$ must be non-negative.

\paragraph{Step 2.}\hspace{-2mm}We now argue that $\Delta q^k \geq 0$ for all $k$. Since each $\mathcal{H}^k$ has a positive measure, for every $k$, there exist some $v \in \R_{++}$ and $c \in \R^n_{++}$ such that 
\[v \Delta q^k \geq c \cdot \Delta y^k + \Delta p^k \geq 0\,,\]
where the second inequality follows from \textbf{Step 1}. It follows that $\Delta q^k \geq 0$ for all $k$, as claimed. 

\paragraph{Step 3.}\hspace{-2mm}We now construct a particular family of searchable menus, and apply an extreme point argument. By \textbf{Step 1} and \textbf{Step 2}, we fix any searchable menu $(\mathcal{M}, \preceq)$ of finite size $m$ where $\Delta y^k \geq 0$, $\Delta p^k \geq 0$, and $\Delta q^k \geq 0$ for all $k$.  Let $\{\mathcal{H}^k\}_{k}$ be the nested family of half-spaces defined before and let $\{\mathcal{G}^k\}_{k}$ be the nested family of strict incremental better-off sets.  Without loss of generality, we may remove any option with $\Delta q^k = 0$ since any such option must not be chosen by any type. In our construction, we will hold $\{\mathcal{H}^k\}_{k}$ and $\{\mathcal{G}^k\}_{k}$ fixed throughout. 

For any $\lambda: \{2,\dots, m\} \rightarrow \R_+$, define 
\[\Delta \bar{y}^k_i = \lambda(k) \Delta y^k_i \text{ for all $i$}\,,\,\, \Delta \bar{p}^k = \lambda(k) \Delta p^k\,,\,\, \text{ and }\, \Delta \bar{q}^k = \lambda(k) \Delta q^k \,.\]
Note that for any $\lambda$ such that 
\[\sum_{k=2}^m \lambda(k) \Delta q^k \leq 1\,,\]
we have that $\Delta \bar{q}^k$ is a feasible sequence of incremental quantities. Indeed, for any $k$, given \textbf{Step 2}, we have 
\[\lambda(k) \Delta q^k \geq 0\,,\]
and therefore for all $r = 2, \dots, m$, 
\[0 \leq \sum_{k=2}^r \lambda(k) \Delta q^k \leq \sum_{k=2}^m \lambda(k) \Delta q^k \leq 1 \,.\]
Moreover, for any such $\lambda$, define
\[\bar{y}^r_i := \sum_{k=2}^r \Delta \bar{y}^k_i \text{ for all $i$}\,,\, \bar{p}^r := \sum_{k=2}^r \Delta \bar{p}^k\,,\,\,  \text{ and }\, \bar{q}^r := \sum_{k = 2}^r \Delta \bar{q}^k\,.\]
By \textbf{Step 1}, $\bar{y}^r \geq 0$ and $\bar{p}^r \geq 0$ for all $r$. Thus, skipping the indices $k$ such that $\lambda(k) = 0$, note that by \Cref{thm:main}, $\bar{\mathcal{M}}_\lambda :=\big\{(\bar{q}^r,\bar{p}^r, \bar{y}^r)\big\}$ is a searchable menu (under the same ordering) with a nested family of better-off sets: 
\[\{\mathcal{H}^k\}_{k: \lambda(k) > 0} \subseteq \{\mathcal{H}^k\}_{k}\,.\]
Therefore, for any such $\lambda$, the objective under the corresponding searchable menu can be written as 
\begin{align*}
\text{$\alpha$-Welfare}(\bar{\mathcal{M}}_\lambda) = \alpha \cdot \sum_{k=2}^m\E\Big[w \cdot \big(v \lambda(k) \Delta q^k - \lambda(k) \Delta p^k - \lambda(k) c \cdot \Delta y^k\big) \mid (v, c) \in \mathcal{H}^k\Big] \mu(\mathcal{H}^k) \\
+ (1-\alpha) \cdot \sum_{k=2}^m (\lambda(k) \Delta p^k-C \lambda(k) \Delta q^k) \mu(\mathcal{H}^k)\,,   
\end{align*}
which is a linear functional of $\lambda$. Let $W_\alpha(\lambda)$ denote this linear functional. Then, maximizing the objective among the family of searchable menus parameterized by $\lambda$ is equivalent to 
\[\max_{\lambda: \{2,\dots, m\} \rightarrow \R_+} W_\alpha(\lambda)\]
subject to 
\[\sum_{k=2}^m \lambda(k) \Delta q^k \leq 1\,.\]
Consider the polytope defined by 
\[\Lambda := \Big\{\lambda \in \R^{m-1}_+: \lambda \cdot \Delta q \leq 1 \Big\}\,,\]
which is a compact, convex set (it is bounded because $\Delta q^k > 0$ for all $k$). Note that every $\lambda \in \text{ext}(\Lambda)$ must satisfy $\lambda(k) > 0$ for at most $\text{rank}(\{\Delta q^k\}_k) = 1$ index given that there is one linear constraint. By Bauer's maximum principle, it follows immediately that there exists some searchable menu $\bar{\mathcal{M}}_{\lambda^{*}}$ with one non-trivial option (where $\lambda^*(k) > 0$) such that 
\[\text{$\alpha$-Welfare}(\bar{\mathcal{M}}_\lambda)\leq \text{$\alpha$-Welfare}(\bar{\mathcal{M}}_{\lambda^{*}})\,\]
for all $\lambda \in \Lambda$.  Moreover, as long as the optimal value is strictly positive, the linear constraint must be binding, because otherwise we can scale up $\lambda$ to strictly increase the objective. It follows that the non-trivial option must involve getting the good for sure, and hence $\bar{\mathcal{M}}_{\lambda^{*}}$ is a posted-requirement menu, proving the result. 

\subsubsection{Proof for General Menus}

The conclusion for general menus follows from the same argument as in the proof of \Cref{thm:multiplegood}. By the same approximation argument, for any searchable menu, we can extract a sequence of finite searchable menus converging to it, in the sense that, treated as outcome functions, 
\[(q^{(m)},p^{(m)},y^{(m)}) \rightarrow (q, p, y)\]
almost everywhere on $\R^{n+1}_+$. By the dominated convergence theorem, the objective is continuous with respect to this converging sequence. The rest of the argument is identical.  

\subsubsection{Completion of the Proof}

To complete the proof, we show that if the designer cares only about the profit, then costly screening is not involved, and hence the posted-requirement menu is a posted-price menu. To see this, fix any optimal posted-requirement menu that involves costly screening
\[\Big\{(0, 0, 0)\,,\,\, (1, p, y)\Big\}\,.\]
Simply consider the following modification: 
\[\Big\{(0, 0, 0)\,,\,\, (1, p, 0)\Big\}\,.\]
Note that any type who was choosing the option $(1, p, y)$ in the original menu must choose option $(1, p, 0)$ in the new menu, and moreover, a positive measure of new types would choose option $(1, p, 0)$ in the new menu. Since the type distribution has full support, any posted price strictly above $C$ yields a strictly positive profit, so the optimal profit is strictly positive; since the profit of the original menu is $(p - C) \cdot \mu\big(v - p - c \cdot y \geq 0\big)$, it follows that $p > C$. Thus, each new buyer contributes $p - C > 0$, and the modification strictly improves the expected profit---a contradiction. 

\subsection{Proof of \Cref{prop:concave}}\label{app_prop:concave}

Fix $k\in\{2,\ldots,m-1\}$. Since every utility function $u(x)=\alpha\cdot x$ with $\alpha\in\R_+^n$ belongs to $\Ulin$, the weak incremental nesting condition on $\Ulin$ implies that
\begin{equation}\label{eq:weakn}
\alpha\cdot (x^{k+1}-x^k)\ge 0
\quad\Longrightarrow\quad
\alpha\cdot (x^k-x^{k-1})\ge 0
\qquad\text{for every }\alpha\in\R_+^n\,,
\end{equation}
while the strict incremental nesting condition on $\Ulin$ implies that
\begin{equation}\label{eq:strongn}
\alpha\cdot (x^{k+1}-x^k)>0
\quad\Longrightarrow\quad
\alpha\cdot (x^k-x^{k-1})>0
\qquad\text{for every }\alpha\in\R_+^n\,.
\end{equation}
Condition \eqref{eq:weakn} says that
$x^k-x^{k-1}$ belongs to the dual cone of
\[
\{\alpha\in\R_+^n:\alpha\cdot (x^{k+1}-x^k)\ge 0\}\,.
\]
The standard polar calculus for polyhedral cones (see Corollary 16.4.2 in \citealp{rockafellar1970convex}) yields
\[
\left(\mathbb R^n_+ \cap \{\alpha:\alpha\cdot (x^{k+1}-x^k)\ge 0\}\right)^*
=
\mathbb R^n_+ + \operatorname{cone}\{x^{k+1}-x^k\}\,,
\]
where $^*$ indicates the dual cone. 
Therefore, condition \eqref{eq:weakn} is equivalent to the existence of
$r\in\mathbb R^n_+$ and $\lambda\ge 0$ such that
\[
x^k-x^{k-1}=r+\lambda(x^{k+1}-x^k)\,.
\]
Equivalently,
\begin{equation}\label{eq:ineqw}
x^k-x^{k-1}\ge \lambda(x^{k+1}-x^k)\,.
\end{equation}
Let
\[
t:=\frac{1}{1+\lambda}\in(0,1]\,.
\]
Rearranging inequality \eqref{eq:ineqw}, we get
\begin{equation}\label{eq:ineqa}
x^k\ge t x^{k-1}+(1-t)x^{k+1}\,.
\end{equation}
Now take any $u\in\Uconc$. We first prove that the weak incremental nesting condition holds for $u$. Suppose that
\[
u(x^{k+1})\ge u(x^k)\,.
\]
By coordinate-wise monotonicity of $u$,  inequality \eqref{eq:ineqa}, and concavity,
\[
u(x^k)\ge u \bigl(t x^{k-1}+(1-t)x^{k+1}\bigr)
\ge
 t u(x^{k-1})+(1-t)u(x^{k+1}) \,.
\]
Thus, we must have
\[
u(x^k)\ge u(x^{k-1})\,,
\]
which verifies that weak incremental nesting for $\Ulin$ implies weak incremental nesting for $\Uconc$.

We now prove that the strict incremental nesting condition holds for $u$. Suppose that
\[
u(x^{k+1})>u(x^k)\,.
\]
If $x^{k+1}-x^k\le 0$ coordinate-wise,  monotonicity would imply $u(x^{k+1})\le u(x^k)$,
contradicting the supposition. Hence there exists at least one coordinate $i$ with
$x_i^{k+1}-x_i^k>0$.

Condition \eqref{eq:strongn} (with the vector $\alpha$ taken to be a unit basis vector) implies
\[
x_i^{k+1}-x_i^k>0\quad\Longrightarrow\quad x_i^k-x_i^{k-1}>0
\qquad\text{for every coordinate }i\,.
\]
Define
\[
\bar\lambda:=\min_{i:\,x_i^{k+1}-x_i^k>0}\quad 
\frac{x_i^k-x_i^{k-1}}{x_i^{k+1}-x_i^k}\,.
\]
By the preceding observations, $\bar\lambda$ is well-defined and strictly positive. We claim that
\begin{equation}\label{equation10}
x^k-x^{k-1}\ge \bar\lambda(x^{k+1}-x^k).
\end{equation}
Indeed, for coordinates $i$ such that  $x_i^{k+1}-x_i^k>0$,  inequality \eqref{equation10} follows from the definition of
$\bar\lambda$. If $x_i^{k+1}-x_i^k=0$, then inequality \eqref{eq:ineqw} implies
$x_i^k-x_i^{k-1}\ge 0=\bar\lambda(x_i^{k+1}-x_i^k)$. Finally, if
$x_i^{k+1}-x_i^k<0$, then inequality \eqref{eq:ineqw} implies
\[
x_i^k-x_i^{k-1}\ge \lambda(x_i^{k+1}-x_i^k)\,.
\]
For every coordinate $j$ with $x_j^{k+1}-x_j^k>0$, inequality \eqref{eq:ineqw} implies
\[
\frac{x_j^k-x_j^{k-1}}{x_j^{k+1}-x_j^k}\ge\lambda\,.
\]
Hence $\bar\lambda\ge\lambda$. Multiplying by
$x_i^{k+1}-x_i^k<0$ gives
\[
\bar\lambda(x_i^{k+1}-x_i^k)
\le \lambda(x_i^{k+1}-x_i^k)
\le x_i^k-x_i^{k-1}\,,
\]
which verifies inequality \eqref{equation10}.

Let
\[
\bar t:=\frac{1}{1+\bar\lambda}.
\]
Since $\bar\lambda>0$, we have $\bar t\in(0,1)$. Rearranging \eqref{equation10},
\[
x^k\ge \bar t x^{k-1}+(1-\bar t)x^{k+1}\,.
\]
This, together with the monotonicity and concavity of $u$, implies that
\[
u(x^k)
\ge u\bigl(\bar t x^{k-1}+(1-\bar t)x^{k+1}\bigr)
\ge \bar t u(x^{k-1})+(1-\bar t)u(x^{k+1})
> \bar t u(x^{k-1})+(1-\bar t)u(x^k)\,,
\]
where the strict inequality uses $1-\bar t>0$. Since $\bar t>0$, this yields
\[
u(x^k)>u(x^{k-1})\,.
\]
This shows that strict incremental nesting for  $\Ulin$ implies strict incremental nesting for $\Uconc$.

We have shown that incremental nesting for $\Ulin$ implies
incremental nesting for $\Uconc$. The converse is immediate since $\Ulin\subseteq\Uconc$.

\subsection{Proof of \Cref{thm:progressive}}

\paragraph{Progressivity $\implies$ Searchability.}\hspace{-2mm}Suppose that $T(\,\cdot\,)$ is convex. To prove searchability, by \Cref{prop:concave}, it suffices to show that the menu induced by $T(\,\cdot\,)$ and ordered by earnings is searchable for any utility function $c-\alpha y$ with $\alpha\geq 0$. Note that $y - T(y)-\alpha y$ is concave in $y$. Every concave function on $[0, \overline{y}]$ is semi-strictly quasi-concave in $y$. The claim follows immediately by \Cref{lem:char}. 

\paragraph{Searchability $\implies$ Progressivity.}\hspace{-2mm}Now, suppose that $(\mathcal{M}, \preceq)$ is a searchable menu, where $\mathcal{M} = \{(y, T(y))\}$ specifies the taxes $T(y)$ for all earnings $y \in [0, \overline{y}]$. Suppose for contradiction that $T(\,\cdot\,)$ is not convex. Then, there exist $y_1 < y_2 < y_3$ such that 
\[T(y_2) > \frac{y_3 - y_2}{y_3 - y_1} T(y_1) + \frac{y_2 - y_1}{y_3 - y_1} T(y_3)\,.\]
Define the secant slopes
\[
s_{12}:=\frac{T(y_2)-T(y_1)}{y_2-y_1},
\qquad
s_{23}:=\frac{T(y_3)-T(y_2)}{y_3-y_2},
\qquad
s_{13}:=\frac{T(y_3)-T(y_1)}{y_3-y_1}.
\]
The non-convexity inequality is equivalent to
$s_{12}>s_{13}>s_{23}$. By assumption, $T(y_2)-T(y_1)<y_2-y_1$, and hence we have $s_{12}< 1$.

For a type with linear preferences with coefficient $\alpha$, write
$\beta:=1-\alpha$. 
Since $\alpha\ge 0$, the admissible values of $\beta$ are precisely $\beta\le 1$. The payoff of this type from earning $y_i$ is
\[
L_i(\beta):=\beta y_i-T(y_i).
\]

We will consider three possible cases, depending on which earnings level is the middle element according to the order $\preceq$. Recall that any finite submenu of a searchable menu is also searchable by \Cref{thm:main}.

\noindent
\textbf{Case (i):} Suppose $y_2$ is between $y_1$ and $y_3$ in the order $\preceq$. Take $\beta=s_{13}$. Then $L_1(\beta)=L_3(\beta)$.
Non-convexity means that $T(y_2)$ lies strictly above the chord connecting
$(y_1,T(y_1))$ and $(y_3,T(y_3))$. Therefore, we have
\[
L_2(\beta)<L_1(\beta)=L_3(\beta).
\]
Thus, the middle option in the search order gives strictly lower payoff than both endpoint options. This violates searchability by \Cref{lem:char}.

\noindent
\textbf{Case (ii):} Suppose $y_3$ is between $y_1$ and $y_2$ in the order $\preceq$. Choose any $\beta<s_{23}$,
which is an admissible choice because $\beta$ may be negative. Since $\beta<s_{23}$, we have $L_3(\beta)<L_2(\beta)$.
Since $s_{23}<s_{13}$, we also have $\beta<s_{13}$, and hence $L_3(\beta)<L_1(\beta)$.
Therefore, we have 
\[
L_3(\beta)<\min\{L_1(\beta),L_2(\beta)\}.
\]
Since $y_3$ is the middle option in the search order, this violates searchability again by \Cref{lem:char}. 

\medskip

\noindent
\textbf{Case (iii):} Suppose $y_1$ is between $y_2$ and $y_3$ in the order $\preceq$. Recall that  $s_{12}<1$.
Choose $\beta=1$.
Then $\beta>s_{12}$, so $L_1(\beta)<L_2(\beta)$.
We also have $s_{12}>s_{13}$, so $\beta>s_{13}$, and hence
$L_1(\beta)<L_3(\beta)$.
Therefore, we have 
\[
L_1(\beta)<\min\{L_2(\beta),L_3(\beta)\}.
\]
Since $y_1$ is the middle option in the search order, this violates searchability again by \Cref{lem:char}. 

Since we derived a contradiction in all possible cases, it must be that $T$ is convex, concluding the proof.

\subsection{Proof of \texorpdfstring{\Cref{thm:grid}}{}}

\subsubsection{Proof of Part (i) of \texorpdfstring{\Cref{thm:grid}}{}}
We first prove part (i) of \Cref{thm:grid}. As explained in the proof sketch in \Cref{sec:grid}, the proof consists of several steps. The first key step is a reduction of the component menus to effectively one-dimensional menus represented by a witness coordinate along which allocations in the menu are increasing. 

We first define a few key objects for the proof. For any grid-searchable product menu $(\mathcal{Y}, p, \preceq)$ and any component $K$, we order the elements in the menu $\mathcal{Y}_K$ according to $\preceq_K$ so that $y^1_K \prec_K y^2_K \prec_K \cdots \prec_K y^{N_K}_K$, and write 
\[\Delta y^s_K := y^s_K - y^{s-1}_K\]
for all $s \in \{2, \dots, N_K\}$. Let
\[P_K(y_K) := p(\underline{y}_{-K}, y_K) - p(\underline{y})\,,\]
be the boundary price function over component menu $K$, where 
\[\underline{y} = (y^1_K)_{K \in \mathcal{K}}\,.\]

We can now formalize the key lemma. We claim that after appropriately deleting strictly dominated options, in each component $K$, there must be a dimension $i \in K$ such that $y^{s}_i$ is strictly increasing in $s$. Moreover, we will show that the boundary price function must be ``discretely convex'' over $y^{s}_i$. Formally, a real-valued function $g$ defined on a finite set of points $a_1 < a_2 < \cdots < a_N \in \R$ is \textit{\textbf{discretely convex}} if $(g(a_{s+1}) - g(a_{s}))/(a_{s+1} - a_s)$ is nondecreasing in $s$ for all $s \in \{1, \dots, N-1\}$.

\begin{lemma}\label{lem:witness}
For any grid-searchable product menu $(\mathcal{Y}, p, \preceq)$ with a strictly increasing price function, there exists a grid-searchable product menu $(\bar{\mathcal{Y}}, \bar{p}, \bar{\preceq})$ such that it can induce the same outcomes for every type, and that for any component $K$, there exists some dimension $i \in K$ such that 
 \[\bar{y}^1_i < \bar{y}^2_i < \cdots < \bar{y}^{\bar{N}_K}_i\,,\]
 and moreover, $\bar{\varphi}_K: \bar{y}^s_i \mapsto \bar{P}_K(\bar{y}^{s}_K)$ is strictly increasing and discretely convex on $\{\bar{y}^1_i, \cdots , \bar{y}^{\bar{N}_K}_i\}$. 
\end{lemma}

\begin{proof}[Proof of \Cref{lem:witness}]

To prove \Cref{lem:witness}, we first need to prove several auxiliary claims. 

We first observe that, given that $v(q, \theta)$ is nondecreasing in $q$, it is without loss of generality to focus on grid-searchable menus in which every adjacent increment in any component $\Delta y^s_K$ has at least one strictly positive coordinate $i$ at any index $s$. This is a much weaker claim than what \Cref{lem:witness} requires, since $i$ can depend on $s$.
\begin{claim}\label{lem:reduction}
For any grid-searchable product menu $(\mathcal{Y}, p, \preceq)$ with a strictly increasing price function, there exists a grid-searchable product menu $(\bar{\mathcal{Y}}, \bar{p}, \bar{\preceq})$ such that it can induce the same outcomes for every type, and for all components $K$ and all $s \in \{2, \dots, \bar{N}_K\}$, there exists at least one dimension $i \in K$ such that 
\[ \Delta \bar{y}^s_i > 0\,.\]
\end{claim}
\begin{proof}[Proof of \Cref{lem:reduction}]
Fix any component $K$ and any $s \in \{2, \dots, N_K\}$ for which 
\[ \Delta y^s_i \leq 0 \text{ for all $i \in K$}\,.\]
Since $p$ is strictly increasing with respect to $\preceq$ and $v(q, \theta)$ is nondecreasing in $q$, we have 
\[v(y_{-K}, y^s_K, \theta) - p(y_{-K}, y^s_K) < v(y_{-K}, y^{s-1}_K, \theta) - p(y_{-K}, y^{s-1}_K) \]
for all $\theta$ and all $y_{-K}$. Therefore, no type chooses any option $(y_{-K}, y^s_K)$ from the product menu. Consider the modified product menu $(\bar{\mathcal{Y}}, \bar{p}, \bar{\preceq})$ where $\bar{\mathcal{Y}}_K = \mathcal{Y}_K \backslash \{y^s_K\}$ and $\bar{\mathcal{Y}}_{-K} = \mathcal{Y}_{-K}$, with the price function $\bar{p}$ and the ordering $\bar{\preceq}$ defined as the restrictions of $p$ and $\preceq$ on this subspace. 

Clearly, $(\bar{\mathcal{Y}}, \bar{p}, \bar{\preceq})$ can induce the same outcome as $(\mathcal{Y}, p, \preceq)$. We argue that the modified menu continues to be grid-searchable. Indeed, fix any $y  \in \bar{\mathcal{Y}}$ that is not globally optimal. Then, by construction, $y$ cannot be globally optimal in $\mathcal{Y}$. Since $\mathcal{Y}$ is grid-searchable, there exists some component $K'$ such that moving to the predecessor or the successor according to $\preceq_{K'}$ yields a strict improvement. If $K' \neq K$, then the same neighbor  $y'$ is also in $\bar{\mathcal{Y}}$ and hence yields a strict improvement there as well. Now, suppose $K' = K$. The same argument works if the neighbor  $y'$ does not have its $K$-component index equal to $s$. If the neighbor  $y'$ has its $K$-component index equal to $s$, then it no longer exists in $\bar{\mathcal{Y}}$. However, then we must have that $y_K$ has index either equal to $s - 1$ or $s + 1$. Note that $y_K$ cannot have its index equal to $s - 1$ because if so $y$ would strictly dominate $y'$ by our earlier argument. Thus, $y_K$'s index must be $s+1$. Then, in the modified menu, $(y^{s-1}_K, y_{-K})$ is a neighbor of $y$, which by construction strictly improves on $ (y^{s}_K, y_{-K}) = y'$, and hence must strictly improve on $y$. Therefore, $y$ cannot be locally optimal in $\bar{\mathcal{Y}}$. This proves that the modified menu $(\bar{\mathcal{Y}}, \bar{p}, \bar{\preceq})$ is indeed grid-searchable. 

The claim follows immediately by repeating this process until no such component $K$ and index $s \in \{2, \dots, N_K\}$ exists---the process must end in finitely many iterations since we start with a finite product menu and in each iteration the total number of options in $\mathcal{Y}$ strictly decreases. 
\end{proof}

By \Cref{lem:reduction}, it is without loss of generality to assume that every adjacent increment in any component $\Delta y^s_K$ has at least one strictly positive coordinate at any index $s$. Fix any such product menu $(\mathcal{Y}, p, \preceq)$. Recall that $P_K(y_K)= p(\underline{y}_{-K}, y_K) - p(\underline{y})$, and let 
\[D^s_K(\tau):= \tau \cdot \Delta y^s_K - \Delta P^s_K\,,\]
where 
\[\Delta P^s_K:= P_K(y^s_K) - P_K(y^{s-1}_K)\,.\]
\begin{claim}\label{lem:D}
For any component $K$, and any $s \in \{2, \dots, N_K - 1\}$, there is no $\tau \in \R^K_+$ such that 
\[D^s_K(\tau) < 0 < D^{s+1}_K(\tau)\,.\]
\end{claim}

\begin{proof}[Proof of \Cref{lem:D}]
Suppose for contradiction that there exists $\tau \in \R^K_+$ such that 
\[D^s_K(\tau) < 0 < D^{s+1}_K(\tau)\,.\]
For any $r \in \{1, \dots, N_K\}$, define 
\[U^r(t) := (t \tau) \cdot y^r_K - P_K(y^r_K)\]
for any $t \in \R_{++}$. Then, we have 
\[U^r(t) - U^{r-1}(t) = D^r_K(t \tau)\,.\]
Since $D^s_K(\tau) < 0 < D^{s+1}_K(\tau)$, by continuity, there exists a non-empty open interval $I \subseteq \R_{++}$ containing $1$ such that for all $t \in I$, we have 
\[D^s_K(t\tau) < 0 < D^{s+1}_K(t\tau)\,.\]
Therefore, for all $t \in I$, we have 
\[U^{s-1}(t) > U^s(t) \quad \text{ and } \quad U^s(t)  < U^{s+1}(t)\,.\]
Let 
\[L:= \{1, \dots, s-1\} \quad \text{ and } \quad R:= \{s+1, \dots, N_K\}\,.\]
Let 
\[A(t) := \max_{l \in L}U^l(t) \quad \text{ and } \quad B(t) := \max_{r \in R} U^r(t)\,.\]
Note that for any $t \in I$, any maximizer of $A(t)$ must be a locally optimal choice for the type $t \tau \in \R^K_+$ on the chain $\mathcal{Y}_K$ with prices given by $P_K$---indeed, this is clear if the maximizer of $A(t)$ has index strictly lower than $s - 1$, and this holds even if the maximizer has index equal to $s - 1$ since $U^{s-1}(t) > U^s(t)$. Similarly, every maximizer of $B(t)$ is a locally optimal choice  for the type $t \tau$ on the chain $\mathcal{Y}_K$ with prices given by $P_K$. 

We claim that $A(t)=B(t)$ for all $t\in I$. Indeed, suppose $A(t) < B(t)$. Let $l^* \in L$ be any maximizer of $A(t)$. By the above argument, $y^{l^*}_K$ is a locally optimal choice for the type $t \tau$ on the chain $\mathcal{Y}_K$ with prices given by $P_K$. Now, consider the type $\theta$ with additive linear preferences specified by $\alpha := (t \tau, 0) \in \R^n$. Consider $y := (y^{l^*}_K, \underline{y}_{-K})$. Since $p$ is strictly increasing and type $\theta$ by construction has value $0$ for any allocation in components $K' \neq K$, there is no local improvement from $y$ for any component $K' \neq K$. Moreover, by construction of $\theta$, it also follows that there is no local improvement from $y$ in component $K$ since the comparisons there are equivalent to the comparisons made by type $t \tau$ on the chain $\mathcal{Y}_K$ with prices given by $P_K$. It follows that $y$ must be a local optimum in the original grid-searchable menu and hence $y$ must be globally optimal. However, $y$ cannot be globally optimal since $y':= (y^{r^*}_K, \underline{y}_{-K})$, where $r^* \in R$ is a maximizer of $B(t)$, would yield a strictly higher payoff for type $\theta$ by construction. It follows that we cannot have $A(t) < B(t)$. Similarly, we cannot have $B(t) < A(t)$. Thus, $A(t) = B(t)$ for all $t \in I$. 

It follows that for any $t \in I$, there exists some $l \in L$ and some $r \in R$ such that 
\[U^l(t) = A(t) = B(t) = U^r(t)\,.\]
Therefore, we must have 
\[I \subseteq \bigcup_{l \in L, r \in R} \{t \in I: U^l(t) = U^r(t)\}\,.\]
Since there are only finitely many pairs $(l, r)$, to cover the interval $I$, it must be that there exists some pair $(l^*, r^*)$ such that 
\[U^{l^*}(t) = U^{r^*}(t)\]
for infinitely many $t \in I$. Since both sides are affine functions of $t$, it follows that the above must hold identically for all $t \in \R_{++}$. In particular, their slopes and intercepts agree: 
\[\tau \cdot y^{l^*}_K  = \tau \cdot y^{r^*}_K  \quad \text{ and } \quad P_K(y^{l^*}_K) = P_K(y^{r^*}_K)\,.\]
However, $l^* < r^*$ by construction. Thus, $P_K(y^{l^*}_K) < P_K(y^{r^*}_K)$ since $p$ is strictly increasing according to $\preceq$, which is a contradiction. 
\end{proof}

Now, define the weak and strict incremental better-off sets for the component menu~$K$:  
\[\mathcal{H}^s_K := \{\tau \in \R^K_+: D^s_K(\tau) \geq 0\} \qquad \text{and} \qquad \mathcal{G}^s_K := \{\tau \in \R^K_+: D^s_K(\tau) > 0\}\,.\]
The next claim shows that the incremental nesting condition holds. 
\begin{claim}\label{lem:HG}
 For any component $K$, we have 
 \[\mathcal{H}^{s+1}_K \subseteq \mathcal{H}^{s}_K \qquad \text{ and } \qquad \mathcal{G}^{s+1}_K \subseteq \mathcal{G}^{s}_K\]
 for all $s \in \{2, \dots, N_K - 1\}$. 
\end{claim}
\begin{proof}[Proof of \Cref{lem:HG}]
Suppose for contradiction that $\mathcal{H}^{s+1}_K \not \subseteq \mathcal{H}^{s}_K$. Then, there exists some $\tau \in \mathcal{H}^{s+1}_K$ and $\tau \not \in \mathcal{H}^{s}_K$. Then, by definition, we have 
\[D^{s+1}_K(\tau) \geq 0 \qquad \text{and} \qquad D^s_K(\tau) < 0\,.\]
Note that, by \Cref{lem:D}, it cannot be that $D^{s+1}_K(\tau) > 0$. Thus, we must have $D^{s+1}_K(\tau) = 0$. Hence, we have 
\[\tau \cdot \Delta y^{s+1}_K = \Delta P^{s+1}_K > 0\,.\]
Note that for any $t > 1$, we have 
\[D^{s+1}_K(t \tau) = (t \tau) \cdot \Delta y^{s+1}_K - \Delta P^{s+1}_K = (t - 1) \Delta P^{s+1}_K > 0\,.\]
By continuity, there exists some $t > 1$ such that $D^s_K(t \tau) < 0$. It then follows that 
\[D^s_K(t \tau) < 0 <  D^{s+1}_K(t \tau)\,,\]
contradicting \Cref{lem:D}. 

Similarly, suppose for contradiction that $\mathcal{G}^{s+1}_K \not \subseteq \mathcal{G}^{s}_K$. Then, there exists some $\tau \in \mathcal{G}^{s+1}_K$ and $\tau \not \in \mathcal{G}^{s}_K$. Then, by definition, we have 
\[D^{s+1}_K(\tau) > 0 \qquad \text{and} \qquad D^s_K(\tau) \leq 0\,.\]
Note that, by \Cref{lem:D}, it cannot be that $D^{s}_K(\tau) < 0$. Thus, we must have $D^{s}_K(\tau) = 0$. Hence, we have 
\[\tau \cdot \Delta y^{s}_K = \Delta P^{s}_K > 0\,.\]
Note that for any $t \in (0, 1)$, we have 
\[D^{s}_K(t \tau) = (t \tau) \cdot \Delta y^{s}_K - \Delta P^{s}_K = (t - 1) \Delta P^{s}_K < 0\,.\]
By continuity, there exists some $t \in (0, 1)$ such that $D^{s+1}_K(t \tau) > 0$. It then follows that 
\[D^s_K(t \tau) < 0 <  D^{s+1}_K(t \tau)\,,\]
contradicting \Cref{lem:D}. 
\end{proof}

We are now ready to finish the proof of \Cref{lem:witness}.
Fix any $K$. We first claim that for all $i \in K$, we have 
\[\Delta y^{s+1}_{K,i} > 0  \implies \Delta y^{s}_{K,i} > 0\]
for all $s \in \{2, \dots, N_K - 1\}$ (where the second subscript denotes the coordinate within $K$). Indeed, suppose not. Then, for some dimension $i$, we have 
\[\Delta y^{s+1}_{K,i} > 0 \quad \text{ and } \quad \Delta y^{s}_{K,i} \leq 0\,\]
for some $s$.  Now consider $\tau = M e_i \in \R^{K}_+$ that puts zero weights on every dimension $j \neq i$ and a (large) positive weight $M$ on dimension $i$. Then, with $M$ large enough, we have
\[D^{s+1}_K(\tau) = M \Delta y^{s+1}_{K,i}  - \Delta P^{s+1}_K > 0\,.\]
However, since $\Delta P^s_K > 0$, we also have 
\[D^s_K(\tau) = M \Delta y^{s}_{K,i}  - \Delta P^{s}_K < 0\,,\]
contradicting \Cref{lem:D}. 

Recall that, by \Cref{lem:reduction}, we can consider only grid-searchable menus for which we know that at any index $s \in \{2, \dots, N_K\}$, there exists some index $j(s)$ such that $\Delta y^{s}_{j(s)} > 0$. In particular, this also holds for $s = N_K$. Let $i = j(N_K)$. Then,  $\Delta y^{N_K}_i > 0$. Moreover, by the above argument, we have that 
\[\Delta y^{s}_i > 0\]
for all $s \in \{2, \dots, N_K - 1\}$ as well. Thus, we have found a dimension $i \in K$ such that 
\[y^1_i < y^2_i < \cdots < y^{N_K}_i\,,\]
proving the first part of the claim. 

Now, to prove the second part, note that 
\[\varphi_K(y^s_i) = P_K(y^s_{K})\]
is strictly increasing in $s$ since $p$ is strictly increasing along the order  $\preceq_K$. It follows immediately that for any $y_i < y'_i$ in $\{y^1_i, \cdots , y^{N_K}_i\}$, we have 
\[\varphi_K(y_i) < \varphi_K(y'_i)\]
and hence $\varphi_K(\,\cdot\,)$ is strictly increasing on the grid. Now, to see that it is discretely convex, note that its adjacent slopes are given by 
\[\sigma^s_K := \frac{\varphi_K(y^{s}_i) - \varphi_K(y^{s-1}_i)}{y^{s}_i - y^{s-1}_i} = \frac{\Delta P^{s}_K}{\Delta y^{s}_i}\,.\]
Consider $\tau = M e_i \in \R^{K}_+$ for any $M \in \R_+$. Note that 
\[\tau \in \mathcal{H}^s_K \iff M \Delta y^s_i - \Delta P^s_K \geq 0 \iff M \geq \sigma^s_K\,.\]
By \Cref{lem:HG}, we have  
\[\mathcal{H}^{s+1}_K \subseteq \mathcal{H}^{s}_K, \]
for all $s \in \{2, \dots, N_K - 1\}$. It follows immediately that 
\[\sigma^{s+1}_K \geq \sigma^{s}_K\]
for all $s \in \{2, \dots, N_K - 1\}$. Hence, $\varphi_K$ is discretely convex, as desired, finishing the proof of \Cref{lem:witness}.
\end{proof}

\paragraph{Block-Additive Pricing.}\hspace{-2mm}Equipped with \Cref{lem:witness}, we can now prove that the menu must feature block-additive pricing.  Define a block-additive price function over the menu $\mathcal{Y}$ built from the boundary price functions as
\[q(y) := p(\underline{y}) + \sum_{K  \in \mathcal{K}} P_K (y_K)\,.\]
We will show that grid-searchability implies $p(y) = q(y)$ for all $y$ and hence the menu $(\mathcal{Y}, p, \preceq)$ must have block-additive pricing.
We prove this equality by induction. 

Specifically, for any $y \in \mathcal{Y}$, let $s_K \in \{1, \dots, N_K\}$ denote its index in component $K$, i.e., 
\[y = (y^{s_K}_K)_{K \in \mathcal{K}}\,.\]
Define the \textit{\textbf{rank}} of $y$ as 
\[|y| := \sum_{K \in \mathcal{K}} (s_K-1)\,.\]
We prove the claim by induction on the rank of $y$.

When $|y| = 0$, we have $y = \underline{y}$ and hence $p(y) = p(\underline{y}) = q(\underline{y}) = q(y)$ by construction. When $|y| = 1$, for some $K$, we have $s_K = 2$ and for all other $K' \neq K$, we have $s_{K'} = 1$. It follows immediately that 
\[p(y) = p(y_K, \underline{y}_{-K}) = p(\underline{y})  + \big(p(y_K, \underline{y}_{-K}) - p(\underline{y}) \big) = p(\underline{y})  + P_K(y_K) = q(y)\,.\]
Now, we start the induction argument. Suppose $p(z) = q(z)$ for all $z$ with $|z| < r$. By the above argument, we only need to consider $|y| = r \geq 2$. Moreover, we also only need to consider $y$ with $s_K > 1$ for at least two components, because otherwise we also have  $p(y) = q(y)$ by the construction of $q$. Now, for any $y$, let 
\[\varepsilon(y) := p(y) - q(y)\,.\]
The induction argument has two steps. In \textbf{Step 1}, we show that for all $|y| = r$, $\varepsilon(y) \geq 0$. In \textbf{Step 2}, we show that for all $|y| = r$, $\varepsilon(y) \leq 0$. This then completes the inductive step. 

\textbf{Step 1:} We first show that $\varepsilon(y) \geq 0$. Suppose for contradiction that $\varepsilon(y) < 0$. Consider $a \in \mathcal{Y}$ defined by moving $y$ one index lower in each component whenever possible: for all $K \in \mathcal{K}$,  
\[a_K := \begin{cases}
 y^{s_K - 1}_K &\text{ if $s_K > 1$}\\
 \underline{y}_K &\text{ otherwise}\\
\end{cases} \,.\]
Clearly, we have $|a| < r$. By the inductive hypothesis, $p(a) = q(a)$. Now, consider the additive-linear type defined by $\alpha \in \R^n_+$ as follows: for all $K \in \mathcal{K}$, and all $i \in K$, 
\[\alpha_i := \begin{cases}
\frac{P_K(y_K) - P_K(a_K)}{y^{s_K}_i - y^{s_K-1}_i} &\text{ if $i = j(K)$ and $s_K > 1$}\\
0 &\text{ otherwise}\\
\end{cases} \,,\]
where $j(K)$ is the ``witness'' dimension in component $K$ identified by \Cref{lem:witness} (along which allocations are increasing). We claim that $a \in \mathcal{Y}$ is locally optimal for type $\alpha \in \R^{n}_+$ but not globally optimal. 

Indeed, to see that $a$ is locally optimal, consider its neighbors by moving in any component $K$. If $s_K = 1$, then $a$ can only move up for one index in $K$. The price for obtaining this neighboring option is strictly higher while the utility to type $\alpha$ remains the same since $\alpha_K = 0$. Thus, this cannot be an improvement. Now, suppose $s_K > 1$. Let $a^+, a^-$ denote the neighbors (if exist) of $a$ when moving one index up and down in dimension $K$. Since $r = |y| \geq 2$, we have $|a| \leq r - 2$ by construction. Therefore, $|a^+| \leq r - 1$ and $|a^-| \leq r - 1$. Thus, by the inductive hypothesis, we have 
\[p(a^+) = q(a^+) \quad \text{ and } \quad p(a^-) = q(a^-)\,.\]
It follows that 
\[p(a^+) - p(a) = q(a^+) - q(a) = P_K(y_K) - P_K(a_K)\,.\]
Since 
\[\alpha \cdot (a^+ - a) = \frac{P_K(y_K) - P_K(a_K)}{y^{s_K}_{j(K)} - y^{s_K-1}_{j(K)}} (a^+_{j(K)} - a_{j(K)}) = \frac{P_K(y_K) - P_K(a_K)}{y^{s_K}_{j(K)} - y^{s_K-1}_{j(K)}} (y^{s_K}_{j(K)} - y^{s_K-1}_{j(K)}) = P_K(y_K) - P_K(a_K)\,,\]
type $\alpha$ is indifferent between $a^+$ and $a$. Now, for $a^-$, we also have 
\[p(a^-) - p(a) = q(a^-) - q(a) = P_K(y^{s_K-2}_K) - P_K(y^{s_K-1}_K)\,.\]
Since 
\begin{align*}
\alpha \cdot (a^- - a) = \frac{P_K(y_K) - P_K(a_K)}{y^{s_K}_{j(K)} - y^{s_K-1}_{j(K)}} (a^-_{j(K)} - a_{j(K)}) &= \frac{P_K(y^{s_K}_K) - P_K(y^{s_K-1}_K)}{y^{s_K}_{j(K)} - y^{s_K-1}_{j(K)}} (y^{s_{K}-2}_{j(K)} - y^{s_K-1}_{j(K)}) \\
& \leq \frac{P_K(y^{s_{K}-1}_K) - P_K(y^{s_K-2}_K)}{y^{s_{K}-1}_{j(K)} - y^{s_{K}-2}_{j(K)}} (y^{s_{K}-2}_{j(K)} - y^{s_K-1}_{j(K)})\\
&= P_K(y^{s_K-2}_K) - P_K(y^{s_{K}-1}_K)\,,
\end{align*}
where the inequality holds because (i) $y^{s_{K}-2}_{j(K)} \leq y^{s_K-1}_{j(K)}$ and (ii) $\varphi_K$ is discretely convex by \Cref{lem:witness}. Therefore, type $\alpha$ weakly prefers $a$ over $a^-$ as well. This proves that $a$ is a locally optimal choice for type $\alpha$. 

We now argue that $a$ is not globally optimal for $\alpha$. In particular, we claim that $\alpha$ strictly prefers $y$ over $a$. To see this, note that 
\[\alpha \cdot (y - a) = \sum_{K \in \mathcal{K}; s_K > 1} \frac{P_K(y_K) - P_K(a_K)}{y^{s_K}_{j(K)} - y^{s_K-1}_{j(K)}} \Big(y^{s_K}_{j(K)} - y^{s_K-1}_{j(K)}\Big) = \sum_{K \in \mathcal{K}} P_K(y_K) - \sum_{K \in \mathcal{K}} P_K(a_K) = q(y) - q(a)\,.\]
Therefore, we have 
\[\alpha \cdot (y - a)  - \big(p(y) - p(a)\big) = \alpha \cdot (y - a)  - \big(q(y) + \varepsilon(y) - q(a)\big) = -\varepsilon(y) > 0\,,\]
and thus type $\alpha$ strictly prefers $y$ over $a$. Thus, $a$ cannot be globally optimal for type $\alpha$. By the richness assumption on $\Theta$, it follows immediately that the original product menu cannot be grid-searchable, which is a contradiction. 

\textbf{Step 2:} We now show that $\varepsilon(y) \leq 0$. Suppose for contradiction that $\varepsilon(y) > 0$. Fix any $H \in \mathcal{K}$ such that $s_H > 1$. Consider $b \in \mathcal{Y}$ defined by moving $y$ down one index only in component $H$: for all $K \in \mathcal{K}$, 
\[b_K = \begin{cases}
    y^{s_K - 1}_K &\text{ if $K = H$} \\
    y^{s_K}_K &\text{ otherwise} \,.
\end{cases}
\]
Now, consider an additive linear preference type defined by $\beta \in \R^n_+$ as follows: for all $K \in \mathcal{K}$ and for all $i \in K$, 
\[\beta_i = \begin{cases}
    \frac{p(y) - p(b)}{y^{s_K}_i - y^{s_K-1}_i} &\text{ if $K = H$ and $i = j(K)$} \\
    \frac{P_K(y^{s_K}_K) - P_K(y^{s_K-1}_K)}{y^{s_K}_i - y^{s_K-1}_i} &\text{ if $K \neq H$ and $i = j(K)$ and $s_K > 1$} \\
    0 &\text{ otherwise} \,,
\end{cases}
\]
where $j(K)$ is the dimension in component $K$ identified by \Cref{lem:witness}. We claim that $b \in \mathcal{Y}$ is locally optimal for type $\beta \in \R^{n}_+$ but not globally optimal. 

Indeed, to see that $b$ is locally optimal, consider its neighbors by moving in any component $K$. There are a few cases. First, if $K \neq H$ and $s_K = 1$, then there is no lower index in component $K$. Consider the neighbor of $b$ with one index higher in component $K$. The price for obtaining this neighboring option is strictly higher while the utility to type $\beta$ remains the same since $\beta_K = 0$. Thus, this cannot be an improvement. 

Second, consider the case where $K \neq H$ and $s_K > 1$. Note that $|b| = r-1$ by construction and hence $p(b) = q(b)$ by the inductive hypothesis. Any downward move in component $K$ would result in some $b^-$ with rank $|b^-| = r-2$, which by the inductive hypothesis must satisfy $p(b^-) = q(b^-)$. It follows immediately that 
\begin{align*}
\beta \cdot (b - b^-) - \big( p(b) - p(b^-) \big) &=  \beta \cdot (b - b^-) - \big( q(b) - q(b^-) \big) \\
& =  \frac{P_K(y^{s_K}_K) - P_K(y^{s_K-1}_K)}{y^{s_K}_{j(K)} - y^{s_K-1}_{j(K)}} \big(y^{s_K}_{j(K)} - y^{s_K-1}_{j(K)}\big) - \big(P_K(y^{s_K}_K) - P_K(y^{s_K-1}_K)\big)   = 0\,.
\end{align*}
Therefore, any such neighbor cannot improve on $b$ for type $\beta$. Now, consider an upward move in component $K$ and denote the neighbor as $b^+$. Note that 
\begin{align*}
\beta \cdot (b - b^+) - \big( p(b) - p(b^+) \big) &\geq  \beta \cdot (b - b^+) - \big( q(b) - q(b^+) \big) \\
& =  \frac{P_K(y^{s_K}_K) - P_K(y^{s_K-1}_K)}{y^{s_K}_{j(K)} - y^{s_K-1}_{j(K)}} \big(y^{s_K}_{j(K)} - y^{s_K+1}_{j(K)}\big) - \big(P_K(y^{s_K}_K) - P_K(y^{s_K+1}_K)\big)  \\
& \geq  \frac{P_K(y^{s_K+1}_K) - P_K(y^{s_K}_K)}{y^{s_K+1}_{j(K)} - y^{s_K}_{j(K)}} \big(y^{s_K}_{j(K)} - y^{s_K+1}_{j(K)}\big) - \big(P_K(y^{s_K}_K) - P_K(y^{s_K+1}_K)\big)  \\
&=  \big(P_K(y^{s_K}_K) - P_K(y^{s_K+1}_K)\big)  - \big(P_K(y^{s_K}_K) - P_K(y^{s_K+1}_K)\big)  = 0\,,
\end{align*}
where the first inequality holds because (i) $p(b) = q(b)$ by the inductive hypothesis and (ii) $|b^+| = (r-1) + 1 = r$ which implies that $\varepsilon(b^+) = p(b^+) -q(b^+) \geq 0$ by \textbf{Step 1}, and the second inequality holds because $\varphi_K$ is discretely convex by \Cref{lem:witness} and $y^{s_K}_{j(K)} < y^{s_K+1}_{j(K)}$. Therefore, any such neighbor also cannot improve on $b$ for type $\beta$. 

Finally, consider the remaining case where $K = H$. First, consider an upward move, which leads to the neighbor $y$. Note that by construction, we have  
\begin{align*}
\beta \cdot (b - y) - \big( p(b) - p(y) \big) 
& =  \frac{p(y) - p(b)}{y^{s_K}_{j(K)} - y^{s_K-1}_{j(K)}} \big(y^{s_K-1}_{j(K)} - y^{s_K}_{j(K)} \big) - \big(p(b) - p(y)\big)  \\
& = \big(p(b) - p(y)\big) - \big(p(b) - p(y)\big) = 0\,.
\end{align*}
Therefore, type $\beta$ is  indifferent between $b$ and its neighbor $y$. Now, consider a downward move to $b^-$ (if exists) in component $K$. Then, since $|b| < r$ and $|b^-| < r$, invoking the inductive hypothesis, we have 
\begin{align*}
\beta \cdot (b - b^-) - \big( p(b) - p(b^-) \big) &=  \beta \cdot (b - b^-) - \big( q(b) - q(b^-) \big) \\
& =  \frac{p(y) -p(b)}{y^{s_K}_{j(K)} - y^{s_K-1}_{j(K)}} \big(y^{s_K-1}_{j(K)} - y^{s_K-2}_{j(K)}\big) - \big(P_K(y^{s_K-1}_K) - P_K(y^{s_K-2}_K)\big)  \\
& \geq  \frac{q(y) -q(b)}{y^{s_K}_{j(K)} - y^{s_K-1}_{j(K)}} \big(y^{s_K-1}_{j(K)} - y^{s_K-2}_{j(K)}\big) - \big(P_K(y^{s_K-1}_K) - P_K(y^{s_K-2}_K)\big)  \\
& =  \frac{P_K(y^{s_K}_K) - P_K(y^{s_K-1}_K)}{y^{s_K}_{j(K)} - y^{s_K-1}_{j(K)}} \big(y^{s_K-1}_{j(K)} - y^{s_K-2}_{j(K)}\big) - \big(P_K(y^{s_K-1}_K) - P_K(y^{s_K-2}_K)\big)  \\
& \geq  \frac{P_K(y^{s_K-1}_K) - P_K(y^{s_K-2}_K)}{y^{s_K-1}_{j(K)} - y^{s_K-2}_{j(K)}} \big(y^{s_K-1}_{j(K)} - y^{s_K-2}_{j(K)}\big) - \big(P_K(y^{s_K-1}_K) - P_K(y^{s_K-2}_K)\big)  \\
& =  \big( P_K(y^{s_K-1}_K) - P_K(y^{s_K-2}_K) \big) - \big(P_K(y^{s_K-1}_K) - P_K(y^{s_K-2}_K)\big) = 0\,,
\end{align*}
where the first inequality holds because (i) the inductive hypothesis implies $p(b) = q(b)$, (ii) \textbf{Step 1} implies $p(y) \geq q(y)$, and (iii) $y^{s_K-1}_{j(K)} > y^{s_K-2}_{j(K)}$, and the second inequality holds because $\varphi_K$ is discretely convex by \Cref{lem:witness} and $y^{s_K-1}_{j(K)} > y^{s_K-2}_{j(K)}$. Therefore, type $\beta$ weakly prefers $b$ over its neighbor $b^-$. Together, combining the three cases, we have that $b$ must be a locally optimal choice for type $\beta$. 

We now argue that $b$ is not globally optimal for $\beta$. Since at $y$, we know that $s_K > 1$ for at least two components $K \in \mathcal{K}$, there exists some $G \neq H \in \mathcal{K}$ such that $s_G > 1$. Consider $c \in \mathcal{Y}$ defined by moving $y$ one index lower in component $G$: for all $K \in \mathcal{K}$, 
\[c_K = \begin{cases}
    y^{s_K - 1}_K &\text{ if $K = G$} \\
    y^{s_K}_K &\text{ otherwise} \,.
\end{cases}
\]
We claim that type $\beta$ strictly prefers $c$ over $b$. Indeed, since $|c| = r - 1$, by the inductive hypothesis, we have $p(c) = q(c)$. Now, note that 
\begin{align*}
\beta \cdot (b - c) - \big( p(b) - p(c) \big) &=  \beta \cdot (b - c) - \big( q(b) - q(c) \big) \\
&=   \frac{p(y) - p(b)}{y^{s_H}_{j(H)} - y^{s_H-1}_{j(H)}} \big(y^{s_H-1}_{j(H)} -  y^{s_H}_{j(H)}\big) +  \frac{P_G(y^{s_G}_{G}) - P_G(y^{s_G-1}_{G})}{y^{s_G}_{j(G)} - y^{s_G-1}_{j(G)}} \big(y^{s_G}_{j(G)} -  y^{s_G-1}_{j(G)}\big) \\
&\qquad -\big(P_H(y^{s_H-1}_H) - P_H(y^{s_H}_H) \big) -\big(P_G(y^{s_G}_G) - P_G(y^{s_G-1}_G) \big)\\
& =   p(b) - p(y) +  \big(P_G(y^{s_G}_{G}) - P_G(y^{s_G-1}_{G})\big)
 -\big(P_H(y^{s_H-1}_H) - P_H(y^{s_H}_H) \big) 
 \\
&\qquad -\big(P_G(y^{s_G}_G) - P_G(y^{s_G-1}_G) \big)\\
 &= q(b) + P_H(y^{s_H}_H)  - P_H(y^{s_H-1}_H)  - p(y)\\
 &=q(y) - p(y) = -\varepsilon(y) < 0\,.
\end{align*}
Therefore, type $\beta$ strictly prefers $c$ over $b$. Thus, $b \in \mathcal{Y}$ cannot be globally optimal for type $\beta$. By the richness assumption on $\Theta$, it follows immediately that the original product menu cannot be grid-searchable, which is a contradiction. 

Together, \textbf{Step 1} and \textbf{Step 2} complete the inductive argument. Therefore, $p(y) = q(y)$ for all $y \in \mathcal{Y}$ and hence $p$ must be block-additive.

\paragraph{Within-Component Searchability.}\hspace{-2mm}To complete the proof of part (i) of \Cref{thm:grid}, it remains to show that for any $K$, the component menu $(\mathcal{Y}_K, p_K, \preceq_K)$ is searchable for all additive linear preferences $\alpha \in \R^n_+$ (searchability for all preferences then follows by \Cref{prop:concave}). For any component $K$, the preferences over the options in the component menu are determined by $\alpha_K \in \R^K_+$. Moreover, for any block-additive price function $p$, we have 
\[P_K(y_K) = p_K(y_K) - p_K(\underline{y}_K)\]
for all $K \in \mathcal{K}$. Therefore, the incremental better-off sets and the strict incremental better-off sets defined for $(\mathcal{Y}_K, p_K, \preceq_K)$ are identical to those defined for $(\mathcal{Y}_K, P_K, \preceq_K)$, which by \Cref{lem:HG} must be nested. By \Cref{thm:main}, it then follows immediately that each component menu is searchable for all additive linear preferences.

\subsubsection{Proof of Part (ii) of \texorpdfstring{\Cref{thm:grid}}{}}
The proof of this direction is straightforward. Let  $(\mathcal{Y}, p, \preceq)$ be a product menu with block-additive pricing where each component menu is searchable. We claim that $(\mathcal{Y}, p, \preceq)$ is grid-searchable on $\Theta$. 

Fix any $\theta \in \Theta$ that has additive preferences and write the utility function as $\sum_{K \in \mathcal{K}} v_K(y_K)$. Then, type $\theta$'s payoff from selecting an option from the product menu is then given by 
\[\sum_K v_K(y_K) - \sum_K p_K(y_K)\,.\]
Fix any $y \in \mathcal{Y}$ such that $y$ is not globally optimal for $\theta$. That is, there exists some $w \in \mathcal{Y}$ such that 
\[\sum_K v_K(y_K) - \sum_K p_K(y_K) < \sum_K v_K(w_K) - \sum_K p_K(w_K)\,.\]
It follows that there exists at least one $K \in \mathcal{K}$ such that 
\[ v_K(y_K) -  p_K(y_K) <  v_K(w_K) -  p_K(w_K)\,.\]
Hence, $y_K$ is also suboptimal for type $\theta$ in the component menu $(\mathcal{Y}_K, p_K, \preceq_K)$. By the searchability of the component menu for type $\theta$, there exists some $y'_K$ that is either the predecessor or the successor to $y_K$ according to $\preceq_K$ such that 
\[ v_K(y_K) -  p_K(y_K) <  v_K(y'_K) -  p_K(y'_K)\,.\]
Then, $y':= (y'_K, y_{-K})$ must be a neighbor of $y$ that strictly improves on $y$ by the additivity of the utility function and the additivity of the price function. Therefore, the product menu $(\mathcal{Y}, p, \preceq)$ is grid-searchable, as desired.

\subsection{Proof of \texorpdfstring{\Cref{cor:grid-multiplegood}}{}}

Throughout, fix the partition $\mathcal{K}$. Recall that any menu contains the outside option $(0,0)$. For any component $K$, let $\mu_K$ denote the marginal
of $\mu$ on $\R^K_+$, which inherits a full-support density and finite first
moments.\footnote{The proof of \Cref{thm:multiplegood} uses only these two properties of the type distribution.} Let $\text{OPT}_K$ denote the optimal revenue among searchable menus in the $|K|$-good monopoly problem with type distribution $\mu_K$, which is attained by \Cref{thm:multiplegood}. We start with a normalization lemma. 

\begin{lemma}\label{lem:normalization}
Let $(\mathcal{N}, \preceq)$ be a finite searchable menu in the $|K|$-good environment, where $\mathcal{N} = \{(q^1, \rho^1), \dots, (q^m, \rho^m)\}$ and $0 = \rho^1 < \rho^2 < \cdots < \rho^m$. Then, there exists a finite searchable menu $\mathcal{N}'$ containing the outside option such that, for almost every $\theta_K$, the payment under $\mathcal{N}'$ is weakly higher than the payment under $\mathcal{N}$. Thus, the revenue under $\mathcal{N}$ is no more than $\emph{OPT}_K$. 
\end{lemma}

\begin{proof}[Proof of \Cref{lem:normalization}]
Let $\mathcal{N}'$ be obtained from $\mathcal{N}$ by replacing the bottom option
$(q^1, 0)$ with the outside option $(0, 0)$, keeping the order. We verify the incremental
nesting condition for $\mathcal{N}'$. All indices $k \geq 3$ are unchanged. For the index
$k = 2$, the new incremental better-off sets are 
\[\mathcal{H}^2_{\mathcal{N}'} = \big\{\theta_K : \theta_K \cdot q^2 \geq \rho^2\big\}
\supseteq \big\{\theta_K: \theta_K \cdot (q^2 - q^1) \geq \rho^2\big\} =
\mathcal{H}^2_{\mathcal{N}}\,,\]
since $\theta_K \geq 0$ and $q^1 \geq 0$, and similarly
$\mathcal{G}^2_{\mathcal{N}'} \supseteq \mathcal{G}^2_{\mathcal{N}}$. Since $\mathcal{N}$
is searchable, \Cref{thm:main} gives $\mathcal{H}^3 \subseteq \mathcal{H}^2_{\mathcal{N}}
\subseteq \mathcal{H}^2_{\mathcal{N}'}$ and $\mathcal{G}^3 \subseteq
\mathcal{G}^2_{\mathcal{N}}\subseteq \mathcal{G}^2_{\mathcal{N}'}$, so $\mathcal{N}'$
satisfies the incremental nesting condition and hence is searchable by \Cref{thm:main}. For almost every $\theta_K$, the optimal choice in $\mathcal{N}$ is unique. If that choice
is some $(q^k, \rho^k)$ with $k \geq 2$, it remains uniquely optimal in $\mathcal{N}'$, and hence the payment is
unchanged. If instead the choice is the bottom option, the payment under $\mathcal{N}$ is $0$, while every price in $\mathcal{N}'$ is non-negative. The claim follows. 
\end{proof}

We now prove the corollary. Fix any grid-searchable product menu $(\mathcal{Y}, p, \preceq)$ with a strictly increasing price function. By \Cref{lem:reduction}, we may
delete, component by component, every option whose allocation is weakly dominated by that of its predecessor in the same component menu. Now, by the proof of \Cref{thm:grid}, on the resulting menu the price function is exactly block-additive, $p(y) = \sum_{K} p_K(y_K)$, each component menu $(\mathcal{Y}_K, p_K, \preceq_K)$ is
searchable, and each $p_K$ is strictly increasing along $\preceq_K$ (by block additivity and strict monotonicity of $p$). Let $y^1_K$ denote the $\preceq_K$-minimal element of
$\mathcal{Y}_K$ and normalize $\tilde{p}_K := p_K - p_K(y^1_K)$, so that each $\tilde{p}_K$ satisfies the hypotheses of \Cref{lem:normalization}. Since each $p_K$ is
strictly increasing, the minimal price on the grid is given by 
\[\min_{y \in \mathcal{Y}} p(y) = \sum_{K} p_K(y^1_K)\,.\]
Because the original menu contains the outside option and the deletion preserves the payoff of each type, type $\theta = 0$ must obtain a non-negative payoff, i.e., $-\min_y p(y) \geq 0$, and hence $\sum_K p_K(y^1_K) \leq 0$. 

Since values are additive and the price is block-additive, the agent's problem separates. Indeed, for every type $\theta$, an option $y \in \mathcal{Y}$ is optimal if and only if $y_K$
maximizes $\theta_K \cdot y_K - \tilde{p}_K(y_K)$ over $\mathcal{Y}_K$ for every $K \in \mathcal{K}$. Moreover, note that for almost every $\theta$, every component problem
has a unique solution, so that $\chi_K(\theta)$ depends on $\theta$ only through $\theta_K$. Therefore, we have 
\[\text{Rev}(\mathcal{Y}, p, \preceq) = \sum_K \E_{\mu}\big[p_K(\chi_K(\theta))\big] \leq \sum_K \E_{\mu_K}\big[\tilde{p}_K(\chi_K(\theta_K))\big] \leq \sum_K \text{OPT}_K\,,\]
where the first inequality uses $\sum_K p_K(y^1_K) \leq 0$ and the second inequality uses \Cref{lem:normalization} applied to each $(\mathcal{Y}_K, \tilde{p}_K, \preceq_K)$. 

Now, for each $K$, \Cref{thm:multiplegood} applied to the $|K|$-good problem with distribution $\mu_K$ yields an optimal searchable upgrade menu $\mathcal{M}^\star_K$ with at most $|K|$ tiers, the highest tier being the full bundle of the goods in $K$, and revenue $\text{OPT}_K$. By Step 1 of the proof of \Cref{thm:multiplegood}, its prices are
strictly increasing along its order. Now, assemble the product menu $\mathcal{Y}^\star := \prod_K \mathcal{M}^\star_K$ with the block-additive price $p^\star(y) := \sum_K p^\star_K(y_K)$ and the product of the component orders. This menu contains the outside option, its price function is strictly increasing, and it is grid-searchable by part (ii) of \Cref{thm:grid}. Since the agent's choice problem separates across components, its revenue is $\sum_K \text{OPT}_K$, which attains the above upper bound. Note that $\mathcal{Y}^\star$ is a block-separable upgrade menu with the desired tier structure, completing the proof. 

\subsection{Proof of \texorpdfstring{\Cref{cor:grid-ordeal}}{}}
Fix any partition $\mathcal{K}$ of the $n+1$ allocation dimensions and any grid-searchable product menu $(\mathcal{Y}, p, \preceq)$ with a strictly increasing price function. Recall that payments are assumed to be non-negative $p \geq 0$. Also recall that the menu contains the outside option $(0, \dots, 0)$ with $p(0) = 0$ and that we use the notation $(q, z) \in [0,1]\times \R^n_-$ where $z = -y$. 

We first claim that, for every component $K$, the option $0_K$ is the $\preceq_K$-minimal element of $\mathcal{Y}_K$. Indeed, if $0_K$ had a $\preceq_K$-predecessor $r_K$, then strict monotonicity of $p$ would give $p(r_K, 0_{-K}) < p(0, \dots, 0) = 0$, contradicting $p \geq 0$. In particular, \Cref{lem:reduction} never deletes the options $0_K$, and by
the proof of \Cref{thm:grid}, after the modification in \Cref{lem:reduction}, the price function is block-additive: 
\[p(y) = \sum_K P_K(y_K), \qquad P_K(y_K) := p(y_K, 0_{-K}) \geq 0, \qquad
P_K(0_K) = 0\,,\]
where each component menu $(\mathcal{Y}_K, P_K, \preceq_K)$ is searchable, and each $P_K$ is strictly increasing. Let $K^\star \in \mathcal{K}$ denote the component
containing the productive dimension $q \in [0,1]$. 

Now we claim that, in any component $K \neq K^\star$, every type optimally chooses $0_K$. Indeed, fix any such $K$ and any $z_K \in \mathcal{Y}_K$ with $z_K \neq 0_K$. Since payoffs
are additively separable across components and the price function is block-additive, replacing $z_K$ by $0_K$ changes the payoff of type $(v, c)$ by
\[\big(c_K \cdot 0_K - P_K(0_K)\big) - \big(c_K \cdot z_K - P_K(z_K)\big) = - c_K \cdot z_K + P_K(z_K) > 0\,,\]
where the inequality holds because $c_K \geq 0$ and $z_K \leq 0$, and because $P_K(z_K) > P_K(0_K) = 0$ since $0_K$ is the $\preceq_K$-minimal element of $\mathcal{Y}_K$ and $P_K$ is strictly increasing along $\preceq_K$.

Thus, the $\alpha$-Welfare generated by the product menu equals the $\alpha$-Welfare generated by the searchable menu $(\mathcal{Y}_{K^\star}, P_{K^\star}, \preceq_{K^\star})$, viewed as a menu in the screening problem with allocation $q$, money $p$, and the ordeals in the set $K^\star \setminus \{q\}$, under the marginal distribution of $(v, c_{K^\star}, w)$.\footnote{The marginal inherits a full-support density, bounded first moments, and the
bounded support of $w$, which is all that the proof of \Cref{thm:ordeal} uses.} By \Cref{thm:ordeal}, there exists a posted-requirement menu $\{(0, 0, 0), (1, p^\star,
y^\star_{K^\star})\}$ in that problem with weakly higher $\alpha$-Welfare. Now, setting $y^\star_i := 0$ for all $i \notin K^\star$ extends it to a posted-requirement menu of the
full problem with the same objective value. 

Finally, any posted-requirement menu is a binary menu and hence searchable. The constructed one is also clearly grid-searchable with the partition given by $\mathcal{K}$. Moreover, whenever $p^\star > 0$, the price function would be strictly increasing and hence attain the upper bound of the $\alpha$-Welfare within this class. If the objective is profit maximization, then this is automatic and moreover, by the proof of \Cref{thm:ordeal}, no costly screening would be involved and a posted price is optimal. This completes the proof of \Cref{cor:grid-ordeal}.

\subsection{Proof of \texorpdfstring{\Cref{cor:grid-tax}}{}}

We prove that $(i) \implies (ii) \implies (iii) \implies (i)$.

The implication $(ii) \implies (iii)$ is immediate. We prove the other two. Throughout, we write $S(y) := \sum_i y_i - T(y)$ for total disposable income. The  assumption that marginal tax rates are strictly below $1$ is equivalent to $S$ being strictly increasing in each coordinate. Note that given the Lagrangian formulation, it is without loss of generality to suppress the search over the consumption dimension: the payoff $u(c,\theta) - \lambda_\theta c$ is concave in $c$ and the tax does not depend on $c$, so this component is searchable for every type, separable from the income components, and imposes no restriction on $T$. Thus, we focus on characterizing grid-searchability via searching on income grids. 

Let $z_i := \overline{y} - y_i$ so that a grid $\prod_i Y_i$ becomes a grid $\prod_i Z_i$. Moreover, let 
\[\pi(z) := - S\big(\overline{y}\mathbf{1} - z\big)\,.\]
Since $u$ is strictly increasing in consumption, the
budget constraint binds and the multiplier satisfies $\lambda_\theta > 0$. Now, note that type $\theta$'s payoff, after suppressing the consumption dimension, is given by  
\[\lambda_\theta S(y) -  \sum_i h_i\big(y_i, \theta \big)\,.\]
Equivalently, we can write the payoff as 
\[v(z, \theta) - \pi(z)\,, \,\,\, \text{ where } v(z, \theta) := - \frac{1}{\lambda_\theta}\sum_i h_i(\overline{y} - z_i,
\theta)\,.\]
In the $z$-coordinates, the taxation problem is therefore a quasilinear grid problem in the sense of \Cref{sec:grid}. Moreover, note that $v$ is nondecreasing and concave in $z$ because each $h_i$ is nondecreasing and convex in income. Also, note that $\pi(z)$ is strictly increasing since $S(y)$ is strictly increasing.  Finally, a type with linear preferences $c - \sum_i \alpha_i y_i$ has $\lambda_\theta = 1$,
so that her payoff is
\[G_\alpha(y) := S(y) - \alpha \cdot y \qquad \text{ and } \qquad v(z, \theta) = - \sum_i \alpha_i \big(\overline{y} - z_i\big)  = \alpha \cdot z - \overline{y} \sum_i \alpha_i\,.\]
Since an additive constant does not affect how a type ranks options, this is the linear preference with coefficient vector $\alpha$ in the sense of \Cref{sec:grid}. In particular, richness delivers all $\alpha \in \R^n_+$, so that the transformed problem has a rich type space.

\paragraph{Part 1: $(i) \implies (ii)$.}\hspace{-2mm}Suppose $T = \sum_i T_i$ with each $T_i$ convex. Fix any grid. Note that it suffices to establish grid-searchability when each component menu is ordered by decreasing income. In that case, we have that $z_i$ increases along $\preceq_i$ and $\pi$ is strictly increasing along every component order. Since $T$ is separable, so is $\pi$. Indeed, writing 
\[\pi_i(z_i) := z_i + T_i(\overline{y} - z_i) - \overline{y}\,,\] 
the change of variables gives
$\pi(z) = \sum_i \pi_i(z_i)$.  Moreover, note that each component menu is searchable. Indeed, fix any type $\theta$. The payoff over component $i$ is
\[ - \frac{1}{\lambda_\theta}h_i(\overline{y} - z_i,
\theta) - \pi_i(z_i)\]
which is concave in $z_i$ because $T_i$ and $h_i$ are convex. Thus, each type $\theta$'s payoff over component menu options is semi-strictly quasi-concave. Therefore, each component menu is searchable by \Cref{lem:char}. Part (ii) of \Cref{thm:grid} then yields that the product menu is grid-searchable, as desired.

\paragraph{Part 2: $(iii) \implies (i)$.}\hspace{-2mm}We first show that grid-searchability pins down the component orders. 

\begin{lemma}\label{lem:monotone-order}
For any grid, if the induced product menu is grid-searchable under a product order $\preceq$, then each component order $\preceq_i$ coincides with the income order on $Y_i$
or with its reverse.
\end{lemma}

\begin{proof}[Proof of \Cref{lem:monotone-order}]
Since the grid is finite, we may fix $M > 0$ such that $S(y'_k, y_{-k}) - S(y_k, y_{-k}) < M(y'_k - y_k)$ for every component $k$, every $y_k < y'_k$ in $Y_k$, and every $y_{-k}$. For a linear type $\alpha$, moving within component $k$ from $y_k$ to $y'_k$ changes the payoff by $S(y'_k, y_{-k}) - S(y_k, y_{-k}) - \alpha_k (y'_k - y_k)$. Note that the sign of this change coincides with the sign of $y'_k - y_k$ when $\alpha_k = 0$, and is opposite to it when $\alpha_k > M$.

Fix any component $i$. Consider linear types with $\alpha_k > M$ for every $k \neq i$. Note that for any such type, a profile can be locally optimal only if, in every component $k \neq i$, its income level is a local minimum of income level along $\preceq_k$. Now, fix any profile $\hat{y}_{-i}$ with this property, which exists because every finite
sequence has a local minimum. First, consider any such linear type with $\alpha_i = 0$. Then $(y_i, \hat{y}_{-i})$ is locally optimal if and only if $y_i$ is a local maximum of income levels along $\preceq_i$. If the income levels along $\preceq_i$ had two local maxima, at income levels $a < a'$, then $(a, \hat{y}_{-i})$ and $(a', \hat{y}_{-i})$ would both be locally optimal for this linear type with $\alpha_i = 0$, while the former yields a strictly lower payoff because $S$ is strictly increasing in $y_i$---hence $(a, \hat{y}_{-i})$ would be a local but not a global maximum for the type, contradicting grid-searchability. Now, consider any such linear type with $\alpha_i > M$. A similar argument shows that $(y_i, \hat{y}_{-i})$ is locally optimal if and only if $y_i$ is a local \textit{minimum} of income levels along $\preceq_i$. Therefore, if the income levels along $\preceq_i$ had two local minima, at income levels $b < b'$, then $(b, \hat{y}_{-i})$ and $(b', \hat{y}_{-i})$ would both be locally optimal, while the latter yields a strictly lower payoff because $\alpha_i > M$---again contradicting grid-searchability.

Finally, note that a finite sequence of distinct values with a unique local maximum must be strictly increasing and then strictly decreasing. Now, if both regions were non-degenerate, then the first and the last indices would both be local minima, which we have just ruled out. Therefore, the income levels along $\preceq_i$ are either strictly increasing or strictly decreasing, as desired.
\end{proof}

\textbf{Step 1.} Fix any grid and suppose that the induced product menu is grid-searchable under a product order $\preceq$. By \Cref{lem:monotone-order} and reversing the ordering if necessary, we may assume that each $\preceq_i$ orders $Y_i$ by decreasing income. Then, $z_i$ increases along $\preceq_i$, and $\pi$ is strictly increasing along every component order. Since the transformed problem has a rich type space and $v$ is nondecreasing and concave in $z$, part (i) of \Cref{thm:grid} applies. Moreover, every component menu is
one-dimensional and every adjacent step strictly increases $z_i$, so that no option is
deleted by \Cref{lem:reduction}. Therefore, the proof of that part applies to the menu itself and yields that $\pi$ is exactly block-additive with each component menu
searchable. Now, undoing the change of variables, there
exist functions $\{T^Y_i\}_i$ such that
\[T(y) = \sum_i T^Y_i(y_i) \qquad \text{ for all $y \in Y$}\,,\]
and, for every $i$, the menu $\{(y_i, T^Y_i(y_i)) : y_i \in Y_i\}$ is searchable for all linear types with payoffs $\beta y_i - T^Y_i(y_i)$, where $\beta = 1 - \alpha_i \leq 1$. Consequently, $T^Y_i$ is discretely convex by the proof of 
\Cref{thm:progressive}. 

\textbf{Step 2.} It remains to pass from grids to the tax schedule. Define $T_i(y_i) := T(y_i, 0_{-i}) - T(0)$ for every $i$ and every $y_i \in [0, \overline{y}]$. Now, fix any $y \in [0, \overline{y}]^n$ and apply \textbf{Step 1} to the grid
$\prod_i \{0, y_i\}$: there exist functions $\{\tau_i\}_i$ with $T = \sum_i \tau_i$ on this grid, and evaluating at $0$ and at $(y_i, 0_{-i})$ gives
\[\tau_i(y_i) - \tau_i(0) =  T(y_i, 0_{-i}) - T(0) = T_i(y_i)\]
for all $i$. Therefore, we have 
\[T(y) = \sum_i \tau_i(y_i) = \sum_i \tau_i(0) + \sum_i T_i(y_i) = T(0) + \sum_i
T_i(y_i)\,.\]
Thus, absorbing the constant $T(0)$ into any one of the functions $T_i$, the tax schedule $T$ is additively separable. Finally, fix any $i$ and any $y^1 < y^2 < y^3$ in $[0, \overline{y}]$. Now, apply
\textbf{Step 1} to the grid $Q := \{y^1, y^2, y^3\} \times \prod_{j \neq i} \{0\}$. It yields functions $\{\sigma_k\}_k$ with $T = \sum_k \sigma_k$ on $Q$, where $\sigma_i$ is
discretely convex on $\{y^1, y^2, y^3\}$. Evaluating this decomposition at the points $(y_i, 0_{-i})$ for $y_i \in \{y^1, y^2, y^3\}$ gives
\[\sigma_i(y_i) = T(y_i, 0_{-i}) - \sum_{j \neq i} \sigma_j(0) = T_i(y_i) + T(0) -
\sum_{j \neq i} \sigma_j(0)\,,\]
where the second equality uses the definition of $T_i$. Therefore, $\sigma_i$ and $T_i$ differ by a constant on $\{y^1, y^2, y^3\}$. Since discrete convexity is unaffected by adding a constant, $T_i$ is discretely convex on $\{y^1, y^2, y^3\}$. Thus, we have 
\[\frac{T_i(y^2) - T_i(y^1)}{y^2 - y^1} \;\leq\; \frac{T_i(y^3) - T_i(y^2)}{y^3 - y^2}\,,\]
or, rearranging,
\[T_i(y^2) \;\leq\; \frac{y^3 - y^2}{y^3 - y^1} T_i(y^1) + \frac{y^2 - y^1}{y^3 - y^1}
T_i(y^3)\,.\]
To conclude, fix any $y^1 < y^3$ in $[0, \overline{y}]$ and any $t \in (0,1)$, and apply the above with $y^2 := t y^1 + (1-t) y^3$. Since $\frac{y^3 - y^2}{y^3 - y^1} = t$ and $\frac{y^2 - y^1}{y^3 - y^1} = 1 - t$, the displayed
inequality is then 
\[T_i\big(t y^1 + (1-t) y^3\big) \;\leq\; t\, T_i(y^1) + (1-t)\, T_i(y^3)\,.\]
Since $y^1 < y^3$ and $t \in (0,1)$ were arbitrary, $T_i$ is progressive, as desired.

\section{A Search-Game Foundation}\label{app:one-step}

In this section, we formalize a notion of search simplicity that characterizes searchable menus. Consider a finite menu $(\mathcal{M}, \preceq)$ with size $m$. Consider a pair of extensive-form games with a single player defined as follows. 

In the \textit{\textbf{ascending search game}}, the nodes consist of $(1, \dots, m)$. The starting node is index $1$. At each node $s$ before reaching index $m$, the agent has the choice to either clinch the option $x^s \in \mathcal{M}$, or move to the next node $s + 1$. At node $m$, the agent can choose either to clinch $x^m$ or choose any other option from the menu $\{1, \dots, m-1\}$. 

Symmetrically, in the \textit{\textbf{descending search game}}, the nodes consist of $(m, \dots, 1)$. The starting node is index $m$. At each node $s$ before reaching index $1$, the agent has the choice to either clinch the option $x^s \in \mathcal{M}$, or move to the next node $s - 1$. At node $1$, the agent can choose either to clinch $x^1$ or choose any other option from the menu $\{2, \dots, m\}$. 

Clearly, any dominant strategy in either search game picks an optimal choice from the menu. We adopt the notion of \textit{\textbf{one-step dominance}} from \citet{pycia2023theory}. Formally, we say that a strategy is \textit{\textbf{one-step dominant}} if at each on-path information set $I^\star$, there exists a strategic plan for all subsequent immediate information sets $I$ such that the worst outcome under the one-step-ahead plan is weakly preferred to the best outcome under any deviation at the current information set $I^\star$. We say that an option $x \in \mathcal{M}$ is \textit{\textbf{one-step implementable}} for type $\theta$ if there exists a one-step dominant strategy for type $\theta$ whose outcome is $x$. Following \citet{pycia2023theory}, we say that a search game is \textit{\textbf{one-step simple}} if every type $\theta$ has a one-step dominant strategy. 

We will state our result for regular environments, defined in \Cref{sec:char}. We additionally say that a finite menu $(\mathcal{M}, \preceq)$ is \textit{\textbf{irredundant}} if, for any two distinct options $x \neq x' \in \mathcal{M}$, some type strictly prefers $x'$ to $x$---that is, no option is weakly dominated by another option. If some option were weakly dominated by another, then deleting it would change no type's optimal value and preserve searchability. Thus, irredundancy is without loss of generality for our purposes. Note that irredundancy implies $\mathcal{G}^k \neq \emptyset$ for all $k \in \{2, \dots, m\}$. Recall also that a menu is \textit{\textbf{interior}} if $\mathcal{H}^k \neq \Theta$ for all $k$.  

\begin{prop}\label{prop:search}
Suppose $(\mathcal{M},\preceq)$ is a finite, interior, and irredundant menu in a regular environment. Then, $(\mathcal M,\preceq)$ is searchable if and only if both the ascending and the descending search games are one-step simple.
\end{prop}

Intuitively, searchable menus are exactly the ones for which one-step foresight---either clinching the current option, or planning to clinch the adjacent one---is enough for the agent to find an optimal choice.

The proof proceeds in three steps. First, \Cref{lem:onestep} characterizes the one-step dominant strategies in the two search games. Second, \Cref{lem:qc} translates one-step simplicity of the two games into a ``pseudo quasi-concavity'' property of the agent's payoff over the menu: the payoff must be nondecreasing below the lowest optimal option and nonincreasing above the highest optimal option. Third, \Cref{lem:upgrade} shows that, in a regular environment, the ``pseudo quasi-concavity'' property for every type is nevertheless equivalent to the weak incremental nesting condition. Consequently, the result follows from \Cref{cor:regular}. 

Throughout the proof, we fix a type $\theta$ and write 
\[g(s) := u(x^s, \theta) \qquad \text{ for all $s \in S = \{1, \dots, m\}$}\,.\]
For any type $\theta$, we also write $O(\theta) \subseteq S$ for the set of optimal indices for type $\theta$, and let $k^-(\theta)$ and $k^+(\theta)$ denote, respectively, the smallest and the largest element of $O(\theta)$. 

\begin{lemma}\label{lem:onestep}
Fix any type $\theta$ and any $k \in S$. In the ascending search game, $x^k$ is one-step implementable for $\theta$ if and only if $x^k$ is an optimal choice for $\theta$ in $\mathcal{M}$ and 
\[g(1) \leq g(2) \leq \cdots \leq g(k)\,.\]
In the descending search game, $x^k$ is one-step implementable for $\theta$ if and only if $x^k$ is an optimal choice for $\theta$ in $\mathcal{M}$ and 
\[g(k) \geq g(k+1) \geq \cdots \geq g(m)\,.\]
\end{lemma}

\begin{proof}[Proof of \Cref{lem:onestep}]
We prove the statement for the ascending search game. The argument for the descending search game is symmetric. 

First, note that any strategy in the ascending search game falls into one of two classes: either the agent clinches at some node $j < m$, in which case the outcome is $x^{j}$ and the on-path information sets are $1, \dots, j$; or the agent moves on at every node $s < m$ and selects some option $x^{j}$ at node $m$, in which case the outcome is $x^{j}$ and the on-path information sets are $1, \dots, m$. Note also that from any node $s \leq m$, every option in $\mathcal{M}$ is attainable, since the agent can move on to node $m$ and select it there.

\textbf{Necessity:} Fix a one-step dominant strategy with outcome $x^{j}$. We first claim that $x^{j}$ is an optimal choice for $\theta$. Suppose the agent clinches at node $j < m$. At the on-path information set $j$, the only deviation is to move on to node $j+1$, from which every option in $\mathcal{M}$ is attainable; hence, the best outcome under a deviation is a $\theta$-optimal option, and one-step dominance requires  
\[u(x^{j}, \theta) \geq u(x, \theta) \qquad \text{ for all $x \in \mathcal{M}$}\,.\]
Suppose instead the agent selects $x^{j}$ at node $m$. At the on-path information set $m$, the deviations are exactly the selections of the other options, and one-step dominance requires the same inequality. In either case, $x^{j}$ is an optimal choice for $\theta$. 

We now claim that $g$ is nondecreasing on $\{1, \dots, j\}$. Fix any on-path information set $s$ at which the strategy prescribes ``Continue''. The only deviation at $s$ is to ``Clinch'' $x^s$, whose outcome is $x^s$ with certainty. Consider the possible plans at the immediate successor node $s+1$. If $s + 1 < m$, the plan to ``Clinch'' $x^{s+1}$ yields the outcome $x^{s+1}$ with certainty, whereas the plan to ``Continue'' keeps it feasible to select every option in which case the worst outcome under that plan is a $\theta$-worst option of $\mathcal{M}$. Therefore, regardless of which plan is used, we must have $g(s) \leq g(s+1)$, whenever $s \leq m - 2$.  If the agent clinches at node $j < m$, the on-path information sets at which the strategy prescribes moving on are $s = 1, \dots, j-1$, all of which satisfy $s \leq m-2$; thus, $g$ is nondecreasing on $\{1, \dots, j\}$, as desired. If instead the agent selects $x^{j}$ at node $m$, the above shows that $g$ is nondecreasing on $\{1, \dots, m-1\}$. This proves the claim when $j \leq m - 1$. When $j = m$, we have $g(m) \geq g(m-1)$ because $x^{m}$ is an optimal choice by the previous paragraph, and hence $g$ is nondecreasing on $\{1, \dots, m\}$, as desired. 

\textbf{Sufficiency:} Suppose $x^k$ is an optimal choice for $\theta$ and $g$ is nondecreasing on $\{1, \dots, k\}$. Consider the strategy that selects ``Continue'' at every node $s < k$ and ``Clinch'' $x^k$ at node $k$ (when $k = m$, selects $x^m$ at node $m$). Fix any on-path information set $s < k$. Consider the plan to ``Clinch'' $x^{s+1}$ at node $s+1$ (when $s + 1 = m$, to select $x^{m}$ at node $m$). The outcome under this plan is $x^{s+1}$ with certainty, the only deviation at $s$ yields $x^s$, and $g(s+1) \geq g(s)$ by assumption. At the on-path information set $k$, any deviation attains at best a $\theta$-optimal option, and $u(x^k, \theta)$ is the optimal value by assumption. Therefore, this strategy that yields outcome $x^k$ is one-step dominant.
\end{proof}

\begin{lemma}\label{lem:qc}
Both search games are one-step simple if and only if, for every type $\theta$, $u(x^s, \theta)$ is nondecreasing in $s$ on $\{1, \dots, k^-(\theta)\}$ and nonincreasing in $s$ on $\{k^+(\theta), \dots, m\}$. 
\end{lemma}

\begin{proof}[Proof of \Cref{lem:qc}]
Fix any type $\theta$ and write $k^- = k^-(\theta)$ and $k^+ = k^+(\theta)$. By \Cref{lem:onestep}, the ascending search game admits a one-step dominant strategy for type $\theta$ if and only if $g$ is nondecreasing on $\{1, \dots, k\}$ for some optimal index $k$. If $g$ is nondecreasing on $\{1, \dots, k^-\}$, this holds with $k = k^-$; conversely, since every optimal index satisfies $k \geq k^-$, if $g$ is nondecreasing on $\{1, \dots, k\}$ for some optimal index $k$, then $g$ is nondecreasing on $\{1, \dots, k^-\}$. Hence the ascending search game admits a one-step dominant strategy for type $\theta$ if and only if $g$ is nondecreasing on $\{1, \dots, k^-\}$. Symmetrically, the descending search game admits a one-step dominant strategy for type $\theta$ if and only if $g$ is nonincreasing on $\{k^+, \dots, m\}$. The claim follows. 
\end{proof}

\begin{lemma}\label{lem:upgrade}
In a regular environment, a finite and irredundant menu $(\mathcal{M}, \preceq)$ satisfies the weak incremental nesting condition if and only if, for every type $\theta$, $u(x^s, \theta)$ is nondecreasing in $s$ on $\{1, \dots, k^-(\theta)\}$ and nonincreasing in $s$ on $\{k^+(\theta), \dots, m\}$. 
\end{lemma}

\begin{proof}[Proof of \Cref{lem:upgrade}]
Suppose first that $(\mathcal{M}, \preceq)$ satisfies the weak incremental nesting condition, and fix any type $\theta$. Since $\mathcal{H}^{k+1} \subseteq \mathcal{H}^{k}$ for all $k$, there exists some $j \in S$ such that $\theta \in \mathcal{H}^k$ for all $k \leq j$ and $\theta \not\in \mathcal{H}^k$ for all $k > j$. By the definition of the incremental better-off sets, this means  
\[g(1) \leq \cdots \leq g(j) \quad \text{ and } \quad g(j) > \cdots > g(m)\,.\]
In particular, the maximal payoff is attained at $j$, and hence $k^+(\theta) = j$ and $k^-(\theta) \leq j$ (so $g$ stays flat from $k^-(\theta)$ to $k^+(\theta) = j$). Therefore, we have that $g$ is nondecreasing on $\{1, \dots, k^-(\theta)\}$ and nonincreasing on $\{k^+(\theta), \dots, m\} = \{j, \dots, m\}$, as desired. 

Conversely, suppose that, for every type $\theta$, $g$ is nondecreasing on
$\{1, \dots, k^-(\theta)\}$ and nonincreasing on $\{k^+(\theta), \dots, m\}$---we refer to this as the \textit{\textbf{pseudo quasi-concavity property}}. We first note an auxiliary observation: since the menu is finite and $u(x, \,\cdot\,)$ is continuous in $\theta$, for every type $\hat{\theta}$ there exists an open set $B(\hat{\theta})$ containing $\hat{\theta}$ such that 
\[O(\theta') \subseteq O(\hat{\theta}) \text{ for all } \theta' \in B(\hat{\theta})\,.\]
Indeed, when $O(\hat{\theta}) \neq S$, take the set of types whose payoffs from every option differ from those of $\hat{\theta}$ by strictly less than
\[\frac{1}{2}\min_{j \notin O(\hat{\theta})} \big\{ \max_{s} u(x^s, \hat{\theta}) \,\, -  \,\,u(x^j, \hat{\theta}) \big\} > 0\,,\]
and when $O(\hat{\theta}) = S$, take $B(\hat{\theta}) = \Theta$. 

Now, fix any $k \in \{2, \dots, m-1\}$ and any $\theta \in \mathcal{H}^{k+1}$. We claim that $\theta \in \mathcal{H}^{k}$. Suppose for contradiction that this is not the case, and thus
\[g(k+1) \geq g(k) \quad \text{ and } \quad g(k-1) > g(k)\,.\]
Since $g(k) < g(k-1) \leq \max_s g(s)$, the index $k$ is not optimal for $\theta$. We consider three cases. 

\textbf{Case (i):} Suppose $k < k^-(\theta)$. Then $\{k-1, k\} \subseteq \{1, \dots, k^-(\theta)\}$, and hence $g(k-1) \leq g(k)$ by the pseudo quasi-concavity property, a contradiction. 

\textbf{Case (ii):} Suppose $k > k^+(\theta)$. Then $\{k, k+1\} \subseteq \{k^+(\theta), \dots, m\}$, and hence $g(k) \geq g(k+1)$ by the pseudo quasi-concavity property, and hence $g(k) = g(k+1)$. Since the menu is irredundant, $\mathcal{G}^{k+1} \neq \emptyset$, and
the reachability condition applied to $x^{k+1}$ and $x^{k}$ gives
$\theta \in \cl(\mathcal{G}^{k+1})$. Since $\theta \in \cl(\mathcal{G}^{k+1})$ and $B(\theta)$ is an open set containing $\theta$, the intersection $\mathcal{G}^{k+1} \cap B(\theta)$ is non-empty. Fix any $\theta' \in \mathcal{G}^{k+1} \cap B(\theta)$. Then, $k^+(\theta') \leq k^+(\theta) < k$ because $O(\theta') \subseteq O(\theta)$, and hence $\{k, k+1\} \subseteq \{k^+(\theta'), \dots, m\}$, while $u(x^{k+1}, \theta') > u(x^{k}, \theta')$. This contradicts the pseudo quasi-concavity property for type $\theta'$. 

\textbf{Case (iii):} Suppose $k^-(\theta) < k < k^+(\theta)$. We first construct a finite sequence of types $\theta_0  = \theta,\, \theta_1,\, \theta_2,\, \dots$ with strictly shrinking optimal sets, $O(\theta_{i+1}) \subsetneq O(\theta_i) \subseteq O(\theta)$, such that each $\theta_i$ satisfies the two properties 
\[u(x^{k-1}, \theta_i) > u(x^{k}, \theta_i) \qquad \text{and} \qquad
u(x^{b_i}, \theta_i) > u(x^{k}, \theta_i) \, \text{ for some index } b_i > k\,.\]
Both properties hold for $\theta_0 = \theta$ with $b_0 := k^+(\theta)$. Now, we describe an iterative construction. Consider any $\theta_i$ that satisfies the two properties with $O(\theta_i)$ containing indices on both sides of $k$. Fix any $a \in O(\theta_i)$ with $a < k$ and any $b \in O(\theta_i)$ with $b > k$. Since the menu is irredundant, the set $V := \{\theta' : u(x^{b}, \theta') >
u(x^{a}, \theta')\}$ is non-empty, and since both $x^a$ and $x^b$ are optimal for $\theta_i$, type $\theta_i$ is indifferent between them. Therefore, the reachability condition applied to $x^b$ and $x^a$ gives $\theta_i \in \cl(V)$. The two properties, with $b_{i+1} := b$, are strict inequalities satisfied by $\theta_i$ (note that $u(x^{b}, \theta_i) > u(x^{k}, \theta_i)$ because $b$ is optimal for $\theta_i$ while $k$ is not)---hence, by continuity, they define an open set containing $\theta_i$. Intersecting this set with $B(\theta_i)$, we may choose $\theta_{i+1} \in V$ satisfying both properties and $O(\theta_{i+1}) \subseteq O(\theta_i)$. Moreover, $a \notin O(\theta_{i+1})$ because $u(x^{b}, \theta_{i+1}) > u(x^{a}, \theta_{i+1})$, so that $O(\theta_{i+1}) \subsetneq O(\theta_i)$, as desired. 

Since the optimal sets are non-empty and strictly shrink, after finitely many steps we must reach a type $\theta^\star = \theta_i$ whose optimal indices lie entirely on one side of $k$ (recall that $k \notin O(\theta_i) \subseteq O(\theta)$ throughout). If $O(\theta^\star) \subseteq \{k+1, \dots, m\}$, then $\{k-1, k\} \subseteq \{1, \dots, k^-(\theta^\star)\}$. But then $u(x^{k-1}, \theta^\star) > u(x^{k}, \theta^\star)$ contradicts the pseudo quasi-concavity property for $\theta^\star$. If $O(\theta^\star) \subseteq \{1, \dots, k-1\}$, then $k^+(\theta^\star) < k < b_i$, and hence $\{k, b_i\} \subseteq \{k^+(\theta^\star),\dots, m\}$. But then $u(x^{b_i}, \theta^\star) > u(x^{k}, \theta^\star)$ contradicts the pseudo quasi-concavity property for $\theta^\star$. 

We obtain a contradiction in all cases. Therefore, $\mathcal{H}^{k+1} \subseteq \mathcal{H}^{k}$ for all $k$, as desired. 
\end{proof}

\begin{proof}[Completion of Proof of \Cref{prop:search}]
Suppose first that $(\mathcal{M}, \preceq)$ is searchable. By \Cref{lem:char},
$u(x^s, \theta)$ is semi-strictly quasi-concave, and hence quasi-concave, in $s$ for every type $\theta$. Now, fix any type $\theta$ and any $s < t \leq k^+(\theta)$. If $t = k^+(\theta)$, then $g(t) \geq g(s)$ because $x^{k^+(\theta)}$ is an optimal choice. If $t < k^+(\theta)$, then quasi-concavity applied to $s < t < k^+(\theta)$ gives $g(t) \geq \min \{g(s), g(k^+(\theta)) \} = g(s)$. Thus, $g$ is nondecreasing on $\{1, \dots, k^+(\theta)\}$. A symmetric argument shows that $g$ is nonincreasing on $\{k^-(\theta), \dots, m\}$. In particular, we must have that $g$ is nondecreasing on $\{1, \dots, k^-(\theta)\}$ and nonincreasing on $\{k^+(\theta), \dots, m\}$. Therefore, both search games are one-step simple by \Cref{lem:qc}. Conversely, suppose that both search games are one-step simple. By \Cref{lem:qc}, for every type $\theta$, we have that $g$ is nondecreasing on $\{1, \dots, k^-(\theta)\}$ and nonincreasing on $\{k^+(\theta), \dots, m\}$. Since the environment is regular and the menu is finite and irredundant, \Cref{lem:upgrade} then implies that $(\mathcal{M}, \preceq)$
satisfies the weak incremental nesting condition. Since the menu is also interior, \Cref{cor:regular} implies that $(\mathcal{M}, \preceq)$ is searchable, concluding the proof. 
\end{proof}

\section{Local-to-Global Property Implies Interval Property in Regular Environments}\label{app:interval}

The analysis in \Cref{app:one-step} also shows that, under the hypotheses of \Cref{prop:search}, condition (i) in the definition of searchability---every locally optimal choice is globally optimal---by itself implies condition (ii)---the set of globally optimal choices is an interval. In this sense, the interval requirement in the definition of searchability is not an additional restriction in regular environments. 

\begin{prop}\label{prop:interval}
Suppose $(\mathcal{M},\preceq)$ is a finite, interior, and irredundant menu in a regular environment. If, for every type $\theta \in \Theta$, every local maximum of $u(x^s, \theta)$ is a global maximum, then $(\mathcal{M}, \preceq)$ is searchable. 
\end{prop}

\begin{proof}[Proof of \Cref{prop:interval}]
We follow the same notation as in \Cref{app:one-step}. Fix any type $\theta$. We claim that $g$ is nondecreasing on $\{1, \dots, k^-(\theta)\}$ and nonincreasing on $\{k^+(\theta), \dots, m\}$. Suppose for contradiction that $g(s) > g(s+1)$ for some $s \in \{1, \dots, k^-(\theta) - 1\}$. Let $t^\star$ be a maximizer of $g$ on $\{1, \dots, s\}$. If $t^\star < s$, then every index adjacent to $t^\star$ lies in $\{1, \dots, s\}$, and hence $g(t^\star)$ weakly exceeds the values of $g$ at the adjacent indices. If $t^\star = s$, the same conclusion holds because $g(s) > g(s+1)$. In either case, $t^\star$ is a local maximum of $g$ and hence, by hypothesis, a global maximum. But $t^\star \leq s < k^-(\theta)$, contradicting the minimality of $k^-(\theta)$ among the maximizers of $g$. Therefore, $g$ must be nondecreasing on $\{1, \dots, k^-(\theta)\}$. A symmetric argument shows that $g$ must be nonincreasing on $\{k^+(\theta), \dots, m\}$. Since this holds for every type $\theta$, \Cref{lem:upgrade} implies that $(\mathcal{M}, \preceq)$ satisfies the weak incremental nesting condition. Then, by \Cref{cor:regular}, $(\mathcal{M}, \preceq)$ is searchable, as desired. 
\end{proof}

\end{document}